\documentclass[onecolumn, 12pt]{article}

\usepackage{amsfonts}
\usepackage{epsfig}
\usepackage{graphicx}
\usepackage{float} 
\usepackage{subfigure}
\usepackage{indentfirst}
\usepackage{setspace}
\usepackage[justification=centering]{caption}
\usepackage{latexsym}
\usepackage{multirow}
\usepackage{amsmath}
\usepackage{epstopdf}
\usepackage{cases}
\usepackage{bbding}

\usepackage{bigdelim}
\usepackage{amssymb}
\usepackage{mathrsfs}
\usepackage{caption}
\usepackage{amsthm}
\usepackage{pstricks}
\usepackage[bookmarks,colorlinks,linkcolor=blue,anchorcolor=blue,citecolor=blue]{hyperref}
\usepackage[left=2cm, right=2cm, top=2cm, bottom=2cm]{geometry}
\usepackage{amsthm}
\usepackage[title]{appendix}
\usepackage{caption}
\usepackage{longtable}
\usepackage{epstopdf}

\newtheorem{definition}{Definition}

\newtheorem{theorem}{Theorem}
\newtheorem{corollary}{Corollary}
\newtheorem{example}{Example}
\allowdisplaybreaks[4]

\begin{document}

\title{\Large A Characterization of Reny's Weakly Sequentially Rational Equilibrium through $\varepsilon$-Perfect $\gamma$-Weakly Sequentially Rational Equilibrium\footnote{This work was partially supported by GRF: CityU 11215123 of Hong Kong SAR Government.}}

\author{Yiyin Cao\textsuperscript{\ref{fnote1}} and Chuangyin Dang\textsuperscript{\ref{fnote2}}\footnote{Corresponding Author}}
\date{ }

\maketitle

\footnotetext[1]{School of Management, Xi'an Jiaotong University, Xi'an, China, yiyincao2-c@my.cityu.edu.hk\label{fnote1}}
\footnotetext[2]{Department of Systems Engineering, City University of Hong Kong, Kowloon, Hong Kong, mecdang@cityu.edu.hk \label{fnote2}}

\date{ }
\maketitle

\begin{abstract}

A weakening of sequential rationality of sequential equilibrium yields Reny's (1992) weakly sequentially rational equilibrium (WSRE) in extensive-form games. WSRE requires Kreps and Wilson's (1982) consistent assessment to satisfy global rationality of nonconvex payoff functions at every information set reachable by a player's own strategy. The consistent assessment demands a convergent sequence of totally mixed behavioral strategy profiles and associated Bayesian beliefs. Nonetheless, due to the nonconvexity, proving the existence of WSRE required invoking the existence of a normal-form perfect equilibrium, which is sufficient but not necessary. Furthermore, Reny’s WSRE definition does not fully specify how to construct the convergent sequence. To overcome these challenges, this paper develops a characterization of WSRE through $\varepsilon$-perfect $\gamma$-WSRE with local sequential rationality, which is accomplished by incorporating an extra behavioral strategy profile.  For any given $\gamma>0$, we generate a perfect $\gamma$-WSRE as a limit point of a sequence of $\varepsilon_k$-perfect $\gamma$-WSRE with $\varepsilon_k\to 0$. A WSRE is then acquired from a limit point of a sequence of perfect $\gamma_q$-WSRE with $\gamma_q\to 0$. This characterization enables analytical identification of all WSREs in small extensive-form games and a direct proof of the existence of WSRE.
An application of the characterization yields a polynomial system that serves as a necessary and sufficient condition for verifying whether a totally mixed assessment is an $\varepsilon$-perfect $\gamma$-WSRE.  Exploiting the system, we devise differentiable path-following methods to compute WSREs by establishing the existence of smooth paths, which are secured from the equilibrium systems of barrier and penalty extensive-form games. Comprehensive numerical results further confirm the efficiency of the methods.
 
\end{abstract}



\section{Introduction}

Nash equilibrium in behavioral strategies (NashEBS) of an extensive-form game is formulated through global rationality of nonconvex payoff functions in Nash~\cite{Nash (1951)}. Given a Bayesian belief system, it is shown in Mas-Colell et al.~\cite{Mas-Colell et al. (1995)} that NashEBS essentially demands global rationality of nonconvex conditional expected payoff functions at every information set reachable by all the players' strategies. NashEBS prescribes a notion of rational behavior and becomes one of the most important and elegant ideas in game theory. However, as pointed out in the literature such as Kreps and Wilson~\cite{Kreps and Wilson (1982)}, Milgrom and Mollner~\cite{Milgrom and Mollner (2021)}, Myerson~\cite{Myerson (1978)}, and Selten~\cite{Selten (1975)},  there may exist many NashEBSs in an extensive-form game and some of these equilibria may be inconsistent with our intuitive notions about what should be the outcome of a game, especially along the off-the-equilibrium paths. To mitigate this ambiguity and eradicate some of these counterintuitive equilibria, Selten~\cite{Selten (1975)} formalizes the concept of perfect equilibrium through strategy perturbations. To meet the requirement that a player's own strategy is optimal starting from any point in a game tree, Kreps and Wilson~\cite{Kreps and Wilson (1982)} explicitly utilize a belief system and establish the notion of sequential equilibrium. By relaxing Kreps and Wilson's consistency conditions on beliefs, the concept of weak sequential equilibrium is presented in Myerson~\cite{Myerson (1991)} and the concept of perfect Bayesian equilibrium is proposed in the literature such as Fudenberg and Tirole~\cite{Fudenberg and Tirole (1991)}.  

Sequential equilibrium requires players to maintain the assumption that all players adhere to equilibrium strategies even after observing a deviation. Reny~\cite{Reny (1992)} critiques this ``inference assumption", arguing that once a player deviates, there is no rationale for others to believe the deviator will subsequently act optimally against the original equilibrium strategy profile. Building on the premise of independent deviations, Reny proposes no restrictions on inferences about a deviator's future play. To formalize an equilibrium concept with this property, Reny~\cite{Reny (1992)} introduces weakly sequential rationality or weakly sequentially rational equilibrium (WSRE) in behavioral strategies. This concept slightly weakens Kreps and Wilson's requirement of global rationality at every information set by demanding rationality only at information sets reachable by a player's own strategy, while retaining the consistent assessments.
WSRE has become the backbone of forward induction equilibrium proposed by Govindan and Wilson~\cite{Govindan and Wilson (2009)}. To define a relevant strategy of a player for a given outcome, Govindan and Wilson~\cite{Govindan and Wilson (2009)} explicitly utilize a WSRE with that outcome for which the strategy at every reachable information set prescribes an optimal continuation given the player's equilibrium belief there. Recently, Siniscalchi~\cite{Siniscalchi (2022)} introduces the notion of structural rationality, which implies WRSE.
According to Reny's definition, determining whether a consistent assessment is a WSRE in behavioral strategies involves verifying global rationality defined by nonconvex conditional payoff functions at each reachable information set of a player. Due to this difficulty, Reny~\cite{Reny (1992)} can only establish the existence of a WSRE through Selten's~\cite{Selten (1975)} existence of normal-form perfect equilibrium of an extensive-form game, which is a sufficient but not necessary condition for the existence of WSRE. Furthermore, consistent assessment demands a convergent sequence of totally mixed behavioral strategy profiles and associated Bayesian beliefs. However, Reny’s WSRE definition does not fully specify how to construct the convergent sequence.
To overcome these deficiencies with Reny's definition, this paper acquires a characterization of WSRE through $\varepsilon$-perfect $\gamma$-WSRE with local sequential rationality, which is achieved by incorporating an extra behavioral strategy profile.

The global rationality of nonconvex payoff functions is embedded in the definition of NashEBS of an extensive-form game. Due to this nonconvexity, determining whether a specified behavioral strategy profile qualifies as a Nash equilibrium or not with this definition becomes challenging. To overcome this difficulty, a general approach is to invoke the associated normal-form game. In this approach, an extensive-form game is converted to a normal-form game by taking into account all pure strategies for each player and the resulting payoffs.  It is well known that one Nash equilibrium of the associated normal-form game of an extensive-form game may correspond to infinitely many NashEBSs of the extensive-form game.  Since a pure strategy specifies a move for each information set of a player,  the number of pure strategies is often exponential in the size of the extensive-form game.
 To address this issue, the sequence form of an extensive-form game with perfect recall is introduced as a strategic description in the literature such as Koller and Megiddo~\cite{Koller and Megiddo (1992)} and von Stengel~\cite{von Stengel (1996)}. Although the sequence form of an extensive-form game is linear in the size of the game tree, it misses some unique features an extensive-form game has and a realization plan of the sequence form may correspond to infinitely many behavioral strategy profiles. Moreover, the advantage of the sequence form vanishes in the computation of conditional expected payoffs. To completely resolve the problems, Cao and Dang~\cite{Cao and Dang (2024)} develop a characterization of NashEBS with local sequential rationality, which is accomplished by introducing an extra behavioral strategy profile and self-independent beliefs. The characterization achieves global rationality through local sequential rationality, which only demands the optimality of a linear payoff function at every information set.  This development for NashEBS inspires the characterization of WSRE in this paper. 
 
Equilibrium computation plays an essential role in the applications of game theory (Harsanyi and Selten~\cite{Harsanyi and Selten (1988)}). The computation of WSREs is closely intertwined with the computation of Nash equilibrium and its refinements in extensive-form games.  Herings and Peeters~\cite{Herings and Peeters (2010)} presents an excellent review of the existing path-following methods for computing Nash equilibria in normal-form games.  However, as pointed out above, such a scheme can be very inefficient since the problem size increases exponentially in the size of game trees. To benefit from special structures of extensive-form games, several methods have been proposed to compute a NashEBS in the literature. A modified Lemke-Howson method was proposed by Wilson~\cite{Wilson (1972)}. A generalization of the Lemke-Howson method to the sequence form of an extensive two-person game was realized in Koller et al.~\cite{Koller et al. (1996)}. As an application of the structure theorem, a global Newton method was described in Govindan and Wilson~\cite{Govindan and Wilson (2002)} for the sequence form of a perturbed extensive-form game. There exist very few methods to compute the refinements of Nash equilibrium in extensive-form games.  As an application of the sequence form, a pivoting procedure was attained in von Stengel et al.~\cite{von Stengel et al. (2002)} to find a normal-form perfect equilibrium in an extensive two-person game. A pivoting algorithm was given in Miltersen and Sorensen~\cite{Miltersen and Sorensen (2010)} to find a quasi-perfect equilibrium of a two-person extensive-form game.  By exploiting the agent quantal response equilibrium in McKelvey and Palfrey~\cite{McKelvey and Palfrey (1998)}, a homotopy method was given in Turocy~\cite{Turocy (2010)} to compute sequential equilibria but the method is very time-consuming and may have difficulties to numerically find such an equilibrium with reasonable accuracy. The existing methods have significantly advanced the applications of equilibria, but none of them can be employed for effectively computing WSREs in extensive-form games. Given the favorable attributes of path-following methods demonstrated in existing literature, particularly their ability to tackle global tasks by iteratively employing local approximations, we exploit the characterizations of WSRE to develop differentiable path-following methods to compute such an equilibrium. The methods are devised by establishing the existence of smooth paths, which are derived from the construction of logarithmic-barrier and convex-quadratic-penalty extensive-form games in which each player at each information set solves against two given behavioral strategy profiles a convex optimization problem.

The remainder of this paper is structured as follows. Section 2 presents a characterization of Reny's WSRE with local sequential rationality for a given consistent assessment. In Section 3, we establish a characterization of WSRE through $\varepsilon$-perfect $\gamma$-WSRE. Section 4 provides illustrative examples demonstrating how to apply our characterization in Section 3 to analytically identify all WSREs in small extensive-form games. Section 5 develops differentiable path-following methods for computing WSREs, while Section 6 presents comprehensive numerical results. Finally, Section 7 concludes the paper with discussion and remarks.

\section{\large A Characterization of WSRE through Local Sequential Rationality for a Given Consistent Assessment}

\subsection{Preliminaries}

This paper is only concerned with finite extensive-form games with perfect recall. To describe an extensive-form game,  some necessary notations are required in accordance with Osborne and Rubinstein~\cite{Osborne and Rubinstein (1994)}.  For easy reference, we summarize these notations in Table~\ref{Table} and elaborate a selection of these notations in Fig.~\ref{Notation}. Given these notations, we represent  
by $\Gamma=\langle N, H, P, f_c, \{{\cal I}_i\}_{i\in N}\rangle$ an extensive-form game. A finite extensive-form game means an extensive-form game with a finite number of histories. Let $X_i(h)$ be the record of player $i$'s experience along the history $h$. Then, $X_i(h)$ is the sequence consisting of the information sets that player $i$  encounters in the history $h$ and the actions she takes at them in the order that these events occur. An extensive-form game has perfect recall if, for each player $i$, we have $X_i(h')=X_i(h'')$ whenever the histories $h'$ and $h''$ are in the same information set of player $i$. 
\begin{table}[ht]\setlength{\abovedisplayskip}{1.2pt}
\setlength{\belowdisplayskip}{1.2pt}
\linespread{1.5} 
\footnotesize
\centering
\caption{Notation for Extensive-Form Games}
\label{Table}
\begin{tabular}{l|l}
\hline
Notation & Terminology\\
\hline
$N=\{1,2,\ldots,n\}$ & Set of players without the chance player\\ 
$h=\langle a_1,a_2,\ldots,a_L\rangle$ &  A history, which is a sequence of actions taken by players\\ 
$H$, $\emptyset\in H$ &  Set of histories, $\langle a_1,\ldots,a_L\rangle\in H$ if $\langle a_1,\ldots,a_K\rangle\in H$ and $L<K$\\ 
$Z$ & Set of terminal histories\\ 
$A(h)=\{a:\langle h,a\rangle\in H\}$ & Set of actions after a nonterminal history $h\in H$\\ 
$P(h)$ & Player who takes an action after a history $h\in H$\\ 
${\cal I}_i$ & Collection of information partitions of player $i$\\ 
$I_i^j\in {\cal I}_i$, $j\in M_i=\{1,\ldots,m_i\}$ & $j$th information set of player $i$, $A(h)=A(h')$ whenever $h,h'\in I_i^j$\\ 
$A(I^j_i)=A(h)$ for any $h\in I^j_i$ & Set of actions of player $i$ at information set $I^j_i$\\ 
$\beta=(\beta^i_{I_i^j}(a):i\in N,I_i^j\in{\cal I}_i,a\in A(I_i^j))$ & Profile of behavioral strategies\\ 
$\beta^i=(\beta^i_{I_i^j}:j\in M_i)$ & Behavioral strategy of player $i$\\  
$\beta^{-i}=(\beta^p_{I^p_q}:p\in N\backslash\{i\},q\in M_p)$ & Profile of behavioral strategies without $\beta^i$\\ 
$\beta^i_{I_i^j}=(\beta^i_{I_i^j}(a):a\in A(I_i^j))^\top$ & Probability measure over $A(I_i^j)$ and $\beta^i_h=\beta^i_{I^j_i}$ for any $h\in I^j_i$\\
$\beta^{-I^j_i}=(\beta^p_{I^p_q}:p\in N,q\in M_p,I^q_p\ne I^j_i)$ & Profile of behavioral strategies without $\beta^i_{I^j_i}$\\ 
$f_c(\cdot|h)=(f_c(a|h):a\in A(h))^\top$ & Probability measure of the chance player $c$ over $A(h)$ \\ 
$\mu=(\mu^i_{I_i^j}(h):i\in N,I_i^j\in{\cal I}_i,h\in I_i^j)$ & Belief system, $\sum_{h\in I_i^j}\mu^i_{I_i^j}(h)=1$ and $\mu^i_{I_i^j}(h)\ge 0$ for all $h\in I_i^j$\\ 
$h\cap A(I^j_i)$; $a\in h$ & $\{a_1,\ldots,a_L\}\cap A(I^j_i)$; $a\in \{a_1,\ldots,a_L\}$ for $h=\langle a_1,\ldots,a_L\rangle$\\
$u^i:Z\rightarrow\mathbb{R}$ & Payoff function of player $i$\\ \hline
\end{tabular}
\end{table}
\begin{figure}[ht!]
    \centering
    \begin{minipage}{0.40\textwidth}
        \centering
        \includegraphics[width=0.9\textwidth, height=0.15\textheight]{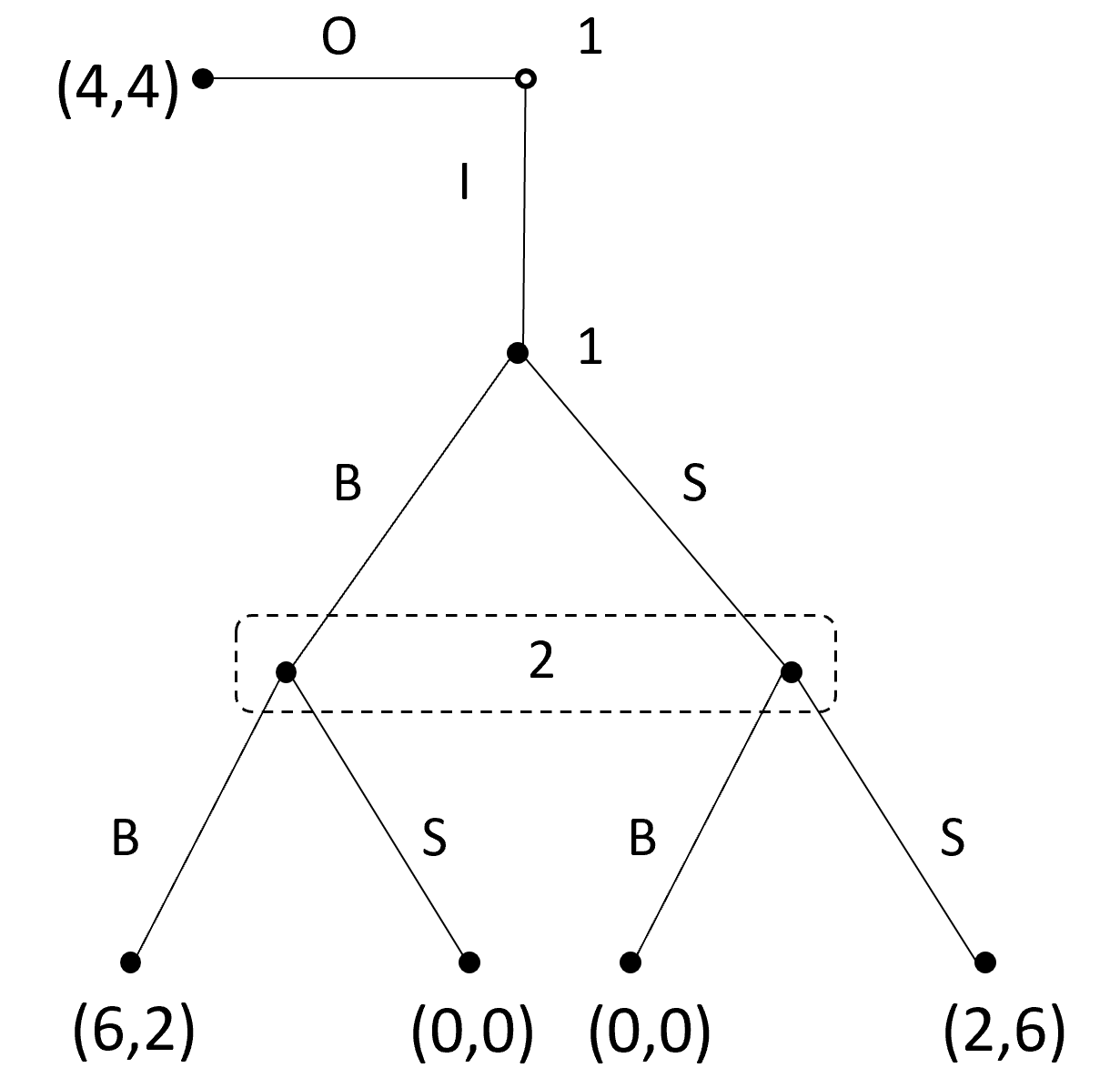}
        % first figure itself
\end{minipage}\hfill
    \begin{minipage}{0.60\textwidth}
       {\scriptsize
       Information sets: ${\cal I}_1=\{I^1_1, I^2_1\}$, ${\cal I}_2=\{I^1_2\}$,  $I^1_1=\{\emptyset\}$, $I^2_1=\{\langle I\rangle\}$, $I^1_2=\{\langle I, B\rangle, \langle I, S\rangle\}$.\newline Player who takes an action after a history: $P(\emptyset)=1$, $P(\langle I\rangle)=1$, $P(\langle I, B\rangle)= 2$, $P(\langle I, S\rangle)=2$.\newline
       Action sets: $A(I^1_1)=\{O, I\}$, $A(I^2_1)=\{B, S\}$,  $A(I^1_2)=\{B, S\}$. \newline
       Behavioral strategies:
       $\beta^1_{I^1_1}=(\beta^1_{I^1_1}(O),\beta^1_{I^1_1}(I))^\top$,  $\beta^1_{I^2_1}=(\beta^1_{I^2_1}(B),\beta^1_{I^2_1}(S))^\top$, $\beta^2_{I^1_2}=(\beta^2_{I^1_2}(B),\beta^2_{I^1_2}(S))^\top$.\newline
       Beliefs: $\mu^2_{I^1_2}=(\mu^2_{I^1_2}(\langle I, B\rangle), \mu^2_{I^1_2}(\langle I, S\rangle))^\top$.}
        \end{minipage}
\caption{\label{Notation}\scriptsize Illustrations of Some Notations}
\end{figure}
 
 For $i\in N$ and $j\in M_i$, we often write a behavioral strategy profile $\beta$ as $(\beta^i_{I_i^j},\beta^{-I^j_i})$. The probability assigned by $\beta$ to any  history $h=\langle a_1,\ldots,a_K\rangle\in H$ equals \begin{equation}\label{wsreeq1}\setlength{\abovedisplayskip}{1.2pt}
\setlength{\belowdisplayskip}{1.2pt} \omega(h|\beta)=\prod\limits_{k=0}^{K-1}\beta^{P(\langle a_1,\ldots,a_k\rangle)}_{\langle a_1,\ldots,a_k\rangle}(a_{k+1}).\end{equation} 
 For $i\in N$, $j\in M_i$, $h=\langle a_1,\ldots,a_K\rangle\in H$, and $a\in A(I^j_i)$, we have \begin{equation}\label{wsreeq1A}\setlength{\abovedisplayskip}{1.2pt}
\setlength{\belowdisplayskip}{1.2pt} \omega(h|a,\beta^{-I^j_i})=\mathop{\prod\limits_{k=0}^{K-1}}\limits_{\langle a_1,\ldots,a_k\rangle\notin I^j_i}\beta^{P(\langle a_1,\ldots,a_k\rangle)}_{\langle a_1,\ldots,a_k\rangle}(a_{k+1}).\end{equation} 
 The expected payoff of player $i$ at $\beta$ is given by
\begin{equation}\label{wsreeq2}\setlength{\abovedisplayskip}{1.2pt}
\setlength{\belowdisplayskip}{1.2pt} u^i(\beta)=\sum\limits_{h\in Z}u^i(h)\omega(h|\beta).\end{equation}
For the game in Fig.~\ref{Notation}, we have $\omega(\langle I, B, S\rangle |\beta)=\beta^1_{I^1_1}(I)\beta^1_{I^2_1}(B)\beta^2_{I^1_2}(S)$,  $\omega(\langle I, B, S\rangle |B, \beta^{-I^2_1})=\beta^1_{I^1_1}(I)\beta^2_{I^1_2}(S)$, and $u^2(\beta)=4\beta^1_{I^1_1}(O)+2\beta^1_{I^1_1}(I)\beta^1_{I^2_1}(B)\beta^2_{I^1_2}(B)+6\beta^1_{I^1_1}(I)\beta^1_{I^2_1}(S)\beta^2_{I^1_2}(S)$.

 A behavioral strategy profile $\beta^*$ is 
a {\bf Nash equilibrium} of an extensive-form game if $u^i(\beta^*)\ge u^i(\beta^i,\beta^{*-i})$ for any $\beta^i$. 
 It is evident that $u^i(\beta^i,\beta^{*-i})$ is a nonconvex function when player $i$ has at least one information set following another. Consequently, Nash's definition of Nash equilibrium in behavioral strategies involves global rationality of nonconvex payoff functions. To tackle this problem, Cao and Dang~\cite{Cao and Dang (2024)} present a characterization of Nash equilibrium in behavioral strategies through local sequential rationality by introducing an extra behavioral strategy profile and self-independent consistency of beliefs. The idea in this characterization inspires the development of this paper.

We denote by $\text{int}(C)$ and $|C|$ the  interior of a set $C$ and the cardinality of a finite set $C$, respectively. Let
$\triangle=\mathop{\times}\limits_{i\in N,\;j\in M_i}\triangle^i_{I^j_i}$ and $\triangle^i=\mathop{\times}\limits_{j\in M_i}\triangle^i_{I^j_i}$, where $\triangle^i_{I^j_i}=\{\beta^i_{I^j_i}\in\mathbb{R}_+^{|A(I^j_i)|}|\sum\limits_{a\in A(I^j_i)}\beta^i_{I^j_i}(a)=1\}$. Let $\Xi=\mathop{\times}\limits_{i\in N,\;j\in M_i}\Xi^i_{I^j_i}$, where $\Xi^i_{I^j_i}=\{\mu^i_{I^j_i}=(\mu^i_{I^j_i}(h):h\in I^j_i)^{\top}|\sum\limits_{h\in I^j_i}\mu^i_{I^j_i}(h)=1,\;0\le\mu^i_{I^j_i}(h)\}$.  An {\bf assessment} in an extensive-form game is a pair $(\beta,\mu)\in\triangle\times\Xi$, where $\beta$ is a behavioral strategy profile and $\mu$ is a function that assigns to every information set a probability measure on the set of histories in the information set. 
We describe $\mu$ as a {\bf belief system}. A player's belief at an information set $I^j_i$ along a specific history $h \in I^j_i$, denoted as $\mu^i_{I^j_i}(h)$, represents the probability that the history $h$ has occurred when it is player $i$'s turn to choose an action at $I^j_i$. A behavioral strategy profile $\beta=(\beta^i_{I^j_i}(a):i\in N,j\in M_i,a\in A(I^j_i))$ is said to be totally mixed if $\beta^i_{I^j_i}(a)>0$ for all $i\in N$, $j\in M_i$, and $a\in A(I^j_i)$. In this paper, $\beta>0$ indicates that $\beta$ is totally mixed. Moreover, an assessment is said to be totally mixed if $\beta^i_{I^j_i}(a)>0$ and $\mu^i_{I^j_i}(h)>0$ for all $i\in N,j\in M_i,a\in A(I^j_i),h\in I^j_i$.

For $i\in N$ and $j\in M_i$,
we define \begin{equation}\label{wsreeq3}\setlength{\abovedisplayskip}{1.2pt}
\setlength{\belowdisplayskip}{1.2pt} u^i(\beta\land I^j_i)=\sum\limits^
{h\cap A(I^j_i)\ne\emptyset}_{h\in Z}u^i(h)\omega(h|\beta)\text{\;\;\;\;
and\;\;\;\;}
 u^i((a,\beta^{-I^j_i})\land I^j_i)=\sum\limits_{a\in h\in Z}u^i(h)\omega(h|a,\beta^{-I^j_i}).\end{equation}
To represent the conditional expected payoff with a belief system $\mu$,
for $i\in N$, $j\in M_i$, and $h=\langle a_1,\ldots,a_K\rangle\in Z$, let
\[\setlength{\abovedisplayskip}{1.2pt}
\setlength{\belowdisplayskip}{1.2pt}
\nu^i_{I^j_i}(h|\beta,\mu)=\left\{\begin{array}{ll} 
\mu^i_{I^j_i}(\hat h)\prod\limits_{k=L}^{K-1}\beta^{P(\langle a_1,\ldots,a_k\rangle)}_{\langle a_1,\ldots,a_k\rangle}(a_{k+1}) & \text{if $\hat h=\langle a_1,\ldots,a_L\rangle\in I^j_i$,}\\

0 & \text{if there is no subhistory of $h$ in $I^j_i$,}
\end{array}\right.\]
and
\[\setlength{\abovedisplayskip}{1.2pt}
\setlength{\belowdisplayskip}{1.2pt}
\nu^i_{I^j_i}(h|a,\beta^{-I^j_i},\mu)=\left\{\begin{array}{ll} 
\mu^i_{I^j_i}(\hat h)\prod\limits_{k=L+1}^{K-1}\beta^{P(\langle a_1,\ldots,a_k\rangle)}_{\langle a_1,\ldots,a_k\rangle}(a_{k+1}) & \text{if $\hat h=\langle a_1,\ldots,a_L\rangle\in I^j_i$,}\\

0 & \text{if there is no subhistory of $h$ in $I^j_i$,}
\end{array}\right.\]
where $\nu^i_{I^j_i}(h|\beta,\mu)$ denotes the probability that the moves along $h$ are played if $\beta$ is taken by players given that  $I^j_i$ has been reached and the belief system is $\mu$. Within $\nu^i_{I^j_i}(h|\beta,\mu)$, $\mu^i_{I^j_i}(\hat{h})$ represents  the probability that a subhistory $\hat{h}\in I^j_i$ along $h$ has occurred conditioned on information set $I^j_i$ has been reached. 
With these notations,  the conditional expected payoff of player $i$ on $I^j_i$ at $(\beta,\mu)$ can be denoted as
\begin{equation}\label{wsreeq4}\setlength{\abovedisplayskip}{1.2pt}\setlength{\belowdisplayskip}{1.2pt}
u^i(\beta,\mu|I^j_i)=\sum\limits^
{h\cap A(I^j_i)\ne\emptyset}_{h\in Z}u^i(h)\nu^i_{I^j_i}(h|\beta,\mu),\end{equation}
and when player $i$ takes a pure action $a\in A(I^j_i)$ at $I^j_i$,  the conditional expected payoff of player $i$ on $I^j_i$ at $(a,\beta^{-I^j_i},\mu)$ can be denoted as
\begin{equation}\label{wsreeq5}\setlength{\abovedisplayskip}{1.2pt}\setlength{\belowdisplayskip}{1.2pt}
u^i(a,\beta^{-I^j_i},\mu|I^j_i)=\sum\limits_{a\in h\in Z}u^i(h)\nu^i_{I^j_i}(h|a,\beta^{-I^j_i},\mu).\end{equation}

For $i\in N$ and $j\in M_i$, the realization probability of information set $I^j_i$ is given by
\[\setlength{\abovedisplayskip}{1.2pt}
\setlength{\belowdisplayskip}{1.2pt}\omega(I^j_i|\beta)=\sum\limits_{h\in I^j_i}\omega(h|\beta).\] For $\beta>0$, we define $\mu(\beta)=(\mu^i_{I^j_i}(h|\beta):i\in N,j\in M_i,h\in I^j_i)$ with $\mu^i_{I^j_i}(h|\beta)=\frac{\omega(h|\beta)}{\omega(I^j_i|\beta)}$.
Let \[\setlength{\abovedisplayskip}{1.2pt}
\setlength{\belowdisplayskip}{1.2pt}\widetilde\Psi=\{(\beta,\mu)\in\triangle\times\Xi|\beta>0\text{ and }\mu=\mu(\beta)\}\] and $\Psi$ be the closure of $\widetilde\Psi$. We say $(\beta,\mu)$ is a {\bf KW-consistent assessment} and $\mu$ is a {\bf KW-consistent belief system} if $(\beta,\mu)\in\Psi$.
Let \[\setlength{\abovedisplayskip}{1.2pt}
\setlength{\belowdisplayskip}{1.2pt}
{\cal Y}(I^j_i)=\{(a, q)|a\in h\cap A(I^q_i)\text{ for some $q\in M_i$ and $h\in I^j_i$}\}.
\]
${\cal Y}(I^j_i)$ consists of all the superscripts of information sets and the actions player $i$ encounters before reaching information set $I^j_i$, which can be an empty set. We define 
\begin{equation}\label{eqAwsre}\setlength{\abovedisplayskip}{1.2pt}
\setlength{\belowdisplayskip}{1.2pt}
{\cal D}(I^j_i|\beta)=\left\{\begin{array}{ll}
\prod\limits_{(a,q)\in {\cal Y}(I^j_i)}\beta^i_{I^q_i}(a) & \text{if ${\cal Y}(I^j_i)\ne\emptyset$,}\\
1 &\text{otherwise.}
\end{array}\right.
\end{equation}
A weakening of Kreps and Wilson's requirement of global rationality at every information set yields the notion of weakly sequential rationality in Reny~\cite{Reny (1992)}.
\begin{definition}[{\bf Weakly Sequential Rationality} or {\bf Weakly Sequentially Rational Equilibrium (WSRE)}, Reny~\cite{Reny (1992)}]\label{wsred1}{\em A consistent assessment $(\beta^*,\mu^*)\in\Psi$ is  a weakly sequentially rational equilibrium (WSRE) if, for every player $i$ and every $j\in M_i$ with ${\cal D}(I^j_i|\beta^*)>0$, we have 
\(u^i(\beta^*,\mu^*|I^j_i)\ge u^i(\beta^i,\beta^{*-i},\mu^*|I^j_i)\)
for every $\beta^i$ of player $i$.}
\end{definition}  

One can see that $u^i(\beta^i,\beta^{*-i},\mu^*|I^j_i)$ is a nonconvex payoff function when $I^j_i$ is followed by at least one more information set of player $i$. Therefore, 
to determine according to Definition~\ref{wsred1} whether an assessment $(\beta^*,\mu^*)\in\Psi$ is a WSRE or not,  one needs to verify the global rationality of nonconvex conditional payoff functions at each information set that is reached by a player's own strategy or solve a dynamic programming problem. To overcome this deficiency with Definition~\ref{wsred1}, we introduce an auxiliary behavioral strategy profile to acquire a characterization of WSRE from local sequential rationality. 

For $i\in N$ and $j\in M_i$, let
$\varrho^i_{I^j_i}(\beta,\tilde\beta)=(\varrho^i_{I^j_i}(\beta^p_{I^q_p},\tilde\beta):p\in N,q\in M_p)$,  where  \begin{equation}\label{esedeqA}\setlength{\abovedisplayskip}{1.2pt}
\setlength{\belowdisplayskip}{1.2pt}
\varrho^i_{I^j_i}(\beta^p_{I^q_p},\tilde\beta)=\left\{\begin{array}{ll}
\tilde\beta^i_{I^q_i} & \text{if both $p=i$ and $h\cap A(I^j_i)\ne\emptyset$ for some $h\in I^q_i$,}\\

\beta^p_{I^q_p} & \text{otherwise.}
\end{array}\right.\end{equation}
Given these notations, we arrive at the following characterization of WSRE through local sequential rationality.

\begin{theorem}\label{ed1wsre}{\em An assessment $(\beta^*,\mu^*)\in\Psi$ is  a WSRE  if and only if  $(\beta^*,\mu^*)$ together with $\tilde\beta^*$  satisfies the properties:\newline
 (i). $\beta^{*i}_{I^j_i}(a')=0$ for any $i\in N$, $j\in M_i$ and $a',a''\in A(I^j_i)$ with 
 \[\setlength{\abovedisplayskip}{1.2pt}
\setlength{\belowdisplayskip}{1.2pt}{\cal D}(I^j_i|\beta^*)(u^i(a'',\varrho^i_{I^j_i}(\beta^{*-I^j_i},\tilde\beta^*),\mu^*|I^j_i)
-u^i(a',\varrho^i_{I^j_i}(\beta^{*-I^j_i},\tilde\beta^*),\mu^*|I^j_i))>0,\] 
(ii). $\tilde\beta^{*i}_{I^j_i}(a')=0$ for any $i\in N$, $j\in M_i$ and $a',a''\in A(I^j_i)$ with \[\setlength{\abovedisplayskip}{1.2pt}
\setlength{\belowdisplayskip}{1.2pt} u^i(a'',\varrho^i_{I^j_i}(\beta^{*-I^j_i},\tilde\beta^*),\mu^*|I^j_i)
-u^i(a',\varrho^i_{I^j_i}(\beta^{*-I^j_i},\tilde\beta^*),\mu^*|I^j_i)>0.\]
}\end{theorem}

\begin{proof} Let $(\beta^*,\mu^*)\in\Psi$. Then there exists a convergent sequence of $\{(\beta^\ell,\mu^\ell),\ell=1,2,\ldots\}$ such that $(\beta^*,\mu^*)=\lim\limits_{\ell\to\infty}(\beta^\ell,\mu^\ell)$, where $\beta^\ell=(\beta^{\ell i}_{I^j_i}(a):i\in N,j\in M_i,a\in A(I^j_i))>0$ and $\mu^\ell=(\mu^{\ell i}_{I^j_i}(h):i\in N, j\in M_i, h\in I^j_i)=\mu(\beta^\ell)=(\mu^i_{I^j_i}(h|\beta^\ell):i\in N,j\in M_i,h\in I^j_i)$ with  $\mu^{\ell i}_{I^j_i}(h)=\mu^i_{I^j_i}(h|\beta^\ell)=\frac{\omega(h|\beta^\ell)}{\omega(I^j_i|\beta^\ell)}$. For $i\in N$, $j\in M_i$, and $a\in A(I^j_i)$, let 
 \[\setlength{\abovedisplayskip}{1.2pt}
\setlength{\belowdisplayskip}{1.2pt}
\begin{array}{l}
\Lambda(a, I^j_i)=\left\{q\in M_i|\text{for any
 $h=\langle a_1,\ldots,a_L\rangle\in I^q_i$, there exists $1\le \ell\le L$ such that $a_\ell=a$}\right\},\\

M(a,I^j_i)=\left\{q\in M_i\left|\begin{array}{l}\text{for any
 $h=\langle a_1,\ldots,a_L\rangle\in I^q_i$, there exists $1\le \ell\le L$ such}\\
 \text{that $a_\ell=a$ and $\{a_{\ell+1},\ldots,a_L\}\cap A(I^p_i)=\emptyset$ for all $p\in M_i$}\end{array}\right.\right\},\\
 
 Z^0(a, I^j_i)=\left\{h=\langle a_1,\ldots,a_K\rangle\in Z\left|\begin{array}{l}
 \text{$a_\ell=a$ for some $1\le \ell\le K$ and}\\
 \text{$\{a_{\ell+1},\ldots,a_K\}\cap A(I^q_i)=\emptyset$ for all $q\in M_i$}\end{array}\right.\right\},\\
 
%Z^i=\{h\in Z|h\cap A(I^j_i)=\emptyset\text{ for all $j\in M_i$}\},\\

\Lambda(I^j_i)=\mathop{\cup}\limits_{a\in A(I^j_i)}\Lambda(a, I^j_i),\;\;
M(I^j_i)=\mathop{\cup}\limits_{a\in A(I^j_i)}M(a,I^j_i),\;\;  
Z^0(I^j_i)=\mathop{\cup}\limits_{a\in A(I^j_i)}Z^0(a, I^j_i). 
\end{array}\]
For $i\in N$ and $j\in M_i$, if ${\cal D}(I^j_i|\beta)=0$, one can deduce from Eq.~(\ref{eqAwsre}) that ${\cal D}(I^q_i|\beta)=0$ for any $q\in\Lambda(I^j_i)$.

\noindent $(\Rightarrow)$ Let $(\beta^*,\mu^*, \tilde\beta^*)$ with $(\beta^*,\mu^*)\in\Psi$ satisfy the properties in Theorem~\ref{ed1wsre}. Then,
for any $i\in N$ and $j\in M_i$ with ${\cal D}(I^j_i|\beta^*)>0$, it follows from the properties in Theorem~\ref{ed1wsre} that
\begin{equation}\setlength{\abovedisplayskip}{1.2pt}
\setlength{\belowdisplayskip}{1.2pt}\label{wsrethmeq1}\sum\limits_{a\in A(I^j_i)}\beta^{*i}_{I^j_i}(a)u^i(a,\varrho^i_{I^j_i}(\beta^{*-I^j_i},\tilde\beta^*),\mu^*|I^j_i)=\sum\limits_{a\in A(I^j_i)}\tilde\beta^{*i}_{I^j_i}(a)u^i(a,\varrho^i_{I^j_i}(\beta^{*-I^j_i},\tilde\beta^*),\mu^*|I^j_i).\end{equation}
For $i\in N$ and $j\in M_i$, we have 
\begin{equation}\setlength{\abovedisplayskip}{1.2pt}
\setlength{\belowdisplayskip}{1.2pt}\label{wsrethmeq2}
\begin{array}{rl}
\max\limits_{\beta^i} u^i(\beta^i,\beta^{\ell,-i},\mu^\ell|I^j_i)= & \max\limits_{\beta^i} \sum\limits_{a\in A(I^j_i)}\beta^i_{I^j_i}(a)u^i(a,\beta^{i,-I^j_i},\beta^{\ell,-i},\mu^\ell|I^j_i)\\
= & \max\limits_{\beta^i_{I^j_i}} \sum\limits_{a\in A(I^j_i)}\beta^i_{I^j_i}(a) \max\limits_{\tilde\beta^i} u^i(a,\tilde\beta^{i,-I^j_i},\beta^{\ell,-i},\mu^\ell|I^j_i)\\
= & \max\limits_{\tilde\beta^i_{I^j_i}} \sum\limits_{a\in A(I^j_i)}\tilde\beta^i_{I^j_i}(a) \max\limits_{\tilde\beta^i} u^i(a,\varrho^i_{I^j_i}(\beta^{\ell,-I^j_i},\tilde\beta),\mu^\ell|I^j_i),
\end{array}
\end{equation}
where the second equality comes from the fact that $\Lambda(a',I^j_i)\cap\Lambda(a'',I^j_i)=\emptyset$ for any $a',a''\in A(I^j_i)$ with $a'\ne a''$.
Thus, as $\ell\to\infty$, we get from Eq.~(\ref{wsrethmeq2}) that \begin{equation}\setlength{\abovedisplayskip}{1.2pt}
\setlength{\belowdisplayskip}{1.2pt}\label{wsrethmeqA}\max\limits_{\beta^i} u^i(\beta^i,\beta^{*,-i},\mu^*|I^j_i)= \max\limits_{\tilde\beta^i_{I^j_i}} \sum\limits_{a\in A(I^j_i)}\tilde\beta^i_{I^j_i}(a) \max\limits_{\tilde\beta^i} u^i(a,\varrho^i_{I^j_i}(\beta^{*,-I^j_i},\tilde\beta),\mu^*|I^j_i).\end{equation} In the following,
we acquire from the backward induction the conclusion that,  for any $i\in N$ and $j\in M_i$,   \begin{equation}\setlength{\abovedisplayskip}{1.2pt}
\setlength{\belowdisplayskip}{1.2pt}\label{wsrethmeq3}\max\limits_{\tilde\beta^i_{I^j_i}}\sum\limits_{a\in  A(I^j_i)}\tilde\beta^{i}_{I^j_i}(a)\max\limits_{\tilde\beta^i}u^i(a,\varrho^i_{I^j_i}(\beta^{*,-I^j_i},\tilde\beta), \mu^* |I^j_i)=\sum\limits_{a\in  A(I^j_i)}\tilde\beta^{*i}_{I^j_i}(a)u^i(a,\varrho^i_{I^j_i}(\beta^{*,-I^j_i},\tilde\beta^*), \mu^* |I^j_i).\end{equation}
{\bf Case (1)}. Consider $i\in N$ and $j\in M_i$ with $M(I^j_i)=\emptyset$. We have
\begin{equation}\setlength{\abovedisplayskip}{1.2pt}
\setlength{\belowdisplayskip}{1.2pt}\label{wsrethmeq4} \begin{array}{rl}
 & \max\limits_{\tilde\beta^i_{I^j_i}}\sum\limits_{a\in A(I^j_i)}\tilde\beta^i_{I^j_i}(a)\max\limits_{\tilde\beta^i}u^i(a,\varrho^i_{I^j_i}(\beta^{*,-I^j_i},\tilde\beta),\mu^*|I^j_i)\\
 = & \max\limits_{\tilde\beta^i_{I^j_i}}\sum\limits_{a\in A(I^j_i)}\tilde\beta^i_{I^j_i}(a)u^i(a,\beta^{*,-I^j_i},\mu^*|I^j_i)
 =  \sum\limits_{a\in A(I^j_i)}\tilde\beta^{*i}_{I^j_i}(a)u^i(a,\beta^{*,-I^j_i},\mu^*|I^j_i),
\end{array}
\end{equation}
where the first equality comes from the fact that $u^i(a,\varrho^i_{I^j_i}(\beta^{*,-I^j_i},\tilde\beta),\mu^*|I^j_i)$ is independent of $\tilde\beta$.\newline
{\bf Case (2)}. 
Consider $i\in N$ and $j\in M_i$ with $M(I^j_i)\ne\emptyset$ such that, for any $q\in M(I^j_i)$,  {\small
 \begin{equation}\setlength{\abovedisplayskip}{1.2pt}
\setlength{\belowdisplayskip}{1.2pt}\label{wsrethmeq5}\max\limits_{\tilde\beta^i_{I^q_i}}\sum\limits_{a'\in  A(I^q_i)}\tilde\beta^{i}_{I^q_i}(a')\max\limits_{\tilde\beta^i}u^i(a',\varrho^i_{I^q_i}(\beta^{*,-I^q_i},\tilde\beta), \mu^* |I^q_i)=\sum\limits_{a'\in  A(I^q_i)}\tilde\beta^{*i}_{I^q_i}(a')u^i(a',\varrho^i_{I^q_i}(\beta^{*,-I^q_i},\tilde\beta^*), \mu^* |I^q_i).\end{equation}}

\noindent
 For $a\in A(I^j_i)$ with $M(a, I^j_i)\ne\emptyset$, {\footnotesize
\begin{equation}\setlength{\abovedisplayskip}{1.2pt}
\setlength{\belowdisplayskip}{1.2pt}\label{wsrethmeq6}\begin{array}{rl}
 & u^i(a,\varrho^i_{I^j_i}(\beta^{\ell,-I^j_i},\tilde\beta),\mu^\ell|I^j_i)\\
 
 = &  \sum\limits_{q\in M(a, I^j_i)}\sum\limits_{h\in Z}^{h\cap A(I^q_i)\ne\emptyset}u^i(h)\nu^i_{I^j_i}(h|a, \varrho^i_{I^j_i}(\beta^{\ell,-I^j_i},\tilde\beta),\mu^\ell) + \sum\limits_{h\in Z^0(a, I^j_i)}u^i(h)\nu^i_{I^j_i}(h|a,\varrho^i_{I^j_i}(\beta^{\ell,-I^j_i},\tilde\beta),\mu^\ell)\\
 
  = &  \sum\limits_{q\in M(a, I^j_i)}\frac{1}{\omega(I^j_i|\beta^{\ell})}\sum\limits_{h\in Z}^{h\cap A(I^q_i)\ne\emptyset}u^i(h)\omega(h|a, \varrho^i_{I^j_i}(\beta^{\ell,-I^j_i},\tilde\beta)) + \sum\limits_{h\in Z^0(a, I^j_i)}u^i(h)\nu^i_{I^j_i}(h|a,\beta^{\ell,-I^j_i},\mu^\ell)\\
  
  = & \sum\limits_{q\in M(a, I^j_i)} \frac{\omega(I^q_i|\beta^{\ell,-I^j_i})}{\omega(I^j_i|\beta^\ell)}\sum\limits_{h\in Z}^{h\cap A(I^q_i)\ne\emptyset}u^i(h)\frac{\omega(h|a, \varrho^i_{I^j_i}(\beta^{\ell,-I^j_i},\tilde\beta))}{\omega(I^q_i|\beta^{\ell,-I^j_i})} + \sum\limits_{h\in Z^0(a, I^j_i)}u^i(h)\nu^i_{I^j_i}(h|a,\beta^{\ell,-I^j_i},\mu^\ell)\\

  = & \sum\limits_{q\in M(a, I^j_i)} \frac{\omega(I^q_i|\beta^{\ell,-I^j_i})}{\omega(I^j_i|\beta^\ell)}\sum\limits_{a'\in  A(I^q_i)}\tilde\beta^{i}_{I^q_i}(a')u^i(a',\varrho^i_{I^q_i}(\beta^{\ell,-I^q_i},\tilde\beta), \mu^\ell |I^q_i)+ \sum\limits_{h\in Z^0(a, I^j_i)}u^i(h)\nu^i_{I^j_i}(h|a,\beta^{\ell,-I^j_i},\mu^\ell).
 \end{array}\end{equation}} 
 
 \noindent 
 Thus,
 \begin{equation}\setlength{\abovedisplayskip}{1.2pt}
\setlength{\belowdisplayskip}{1.2pt}\label{wsrethmeq7}\begin{array}{rl} 
& \sum\limits_{a\in A(I^j_i)}\tilde\beta^i_{I^j_i}(a)u^i(a,\varrho^i_{I^j_i}(\beta^{\ell,-I^j_i},\tilde\beta),\mu^\ell|I^j_i)\\
= & \sum\limits_{a\in A(I^j_i)}\tilde\beta^i_{I^j_i}(a)( \sum\limits_{q\in M(a, I^j_i)} \frac{\omega(I^q_i|\beta^{\ell,-I^j_i})}{\omega(I^j_i|\beta^\ell)}\sum\limits_{a'\in  A(I^q_i)}\tilde\beta^{i}_{I^q_i}(a')u^i(a',\varrho^i_{I^q_i}(\beta^{\ell,-I^q_i},\tilde\beta), \mu^\ell |I^q_i)\\
& + \sum\limits_{h\in Z^0(a, I^j_i)}u^i(h)\nu^i_{I^j_i}(h|a,\beta^{\ell,-I^j_i},\mu^\ell)).
\end{array}\end{equation}
 Furthermore, since \begin{equation}
 \label{wsrethmeqC}\setlength{\abovedisplayskip}{1.2pt}
\setlength{\belowdisplayskip}{1.2pt}
\frac{\omega(I^q_i|\beta^{\ell,-I^j_i})}{ \omega(I^j_i|\beta^\ell)}=\sum\limits^{\hat h =\langle a_1,\ldots,a_k\rangle\in I^j_i}_{h=\langle\hat h, a_{k+1}=a,\ldots,a_L\rangle
\in I^q_i}\mu^{\ell i}_{I^j_i}(\hat h)\prod\limits_{g=k+1}^{L-1}\beta^{\ell, P(\langle a_1,\ldots,a_{g}\rangle)}_{\langle a_1,\ldots,a_{g}\rangle}(a_{g+1}),
\end{equation} 
as a result of substituting Eq.~(\ref{wsrethmeqC}) into Eq.~(\ref{wsrethmeq6}), it holds that
\begin{equation}\setlength{\abovedisplayskip}{1.2pt}
\setlength{\belowdisplayskip}{1.2pt}\label{wsrethmeq8}\begin{array}{rl} 
& \max\limits_{\tilde\beta^i}u^i(a,\varrho^i_{I^j_i}(\beta^{*,-I^j_i},\tilde\beta),\mu^*|I^j_i)
=  \max\limits_{\tilde\beta^i} \lim\limits_{\ell\to\infty}  u^i(a,\varrho^i_{I^j_i}(\beta^{\ell,-I^j_i},\tilde\beta),\mu^\ell|I^j_i) \\

= & \sum\limits_{q\in M(a, I^j_i)}\sum\limits^{\hat h =\langle a_1,\ldots,a_k\rangle\in I^j_i}_{h=\langle\hat h, a_{k+1}=a,\ldots,a_L\rangle
\in I^q_i}\mu^{* i}_{I^j_i}(\hat h)\prod\limits_{g=k+1}^{L-1}\beta^{*, P(\langle a_1,\ldots,a_{g}\rangle)}_{\langle a_1,\ldots,a_{g}\rangle}(a_{g+1})\\

&\max\limits_{\tilde\beta^i_{I^q_i}} \sum\limits_{a'\in  A(I^q_i)}\tilde\beta^{i}_{I^q_i}(a')\max\limits_{\tilde\beta^i}u^i(a',\varrho^i_{I^q_i}(\beta^{*,-I^q_i},\tilde\beta), \mu^* |I^q_i)
 + \sum\limits_{h\in Z^0(a, I^j_i)}u^i(h)\nu^i_{I^q_i}(h|a,\beta^{*,-I^j_i},\mu^*)\\

= & \sum\limits_{q\in M(a, I^j_i)} \sum\limits^{\hat h =\langle a_1,\ldots,a_k\rangle\in I^j_i}_{h=\langle\hat h, a_{k+1}=a,\ldots,a_L\rangle
\in I^q_i}\mu^{* i}_{I^j_i}(\hat h)\prod\limits_{g=k+1}^{L-1}\beta^{*, P(\langle a_1,\ldots,a_{g}\rangle)}_{\langle a_1,\ldots,a_{g}\rangle}(a_{g+1})\\
& \sum\limits_{a'\in  A(I^q_i)}\tilde\beta^{*i}_{I^q_i}(a')u^i(a',\varrho^i_{I^q_i}(\beta^{*,-I^q_i},\tilde\beta^*), \mu^* |I^q_i) + \sum\limits_{h\in Z^0(a, I^j_i)}u^i(h)\nu^i_{I^j_i}(h|a,\beta^{*,-I^j_i},\mu^*),
\end{array}\end{equation}
where the second equality comes from Eq.~(\ref{wsrethmeq6}) and Eq.~(\ref{wsrethmeqC}) and the third equality comes from Eq.~(\ref{wsrethmeq5}).
Let
\begin{equation}\setlength{\abovedisplayskip}{1.2pt}
\setlength{\belowdisplayskip}{1.2pt}\label{wsrethmeqD}\begin{array}{rl} 
w(\beta,\mu)= & \sum\limits_{q\in M(a, I^j_i)} \sum\limits^{\hat h =\langle a_1,\ldots,a_k\rangle\in I^j_i}_{h=\langle\hat h, a_{k+1}=a,\ldots,a_L\rangle
\in I^q_i}\mu^i_{I^j_i}(\hat h)\prod\limits_{g=k+1}^{L-1}\beta^{P(\langle a_1,\ldots,a_{g}\rangle)}_{\langle a_1,\ldots,a_{g}\rangle}(a_{g+1})\\
& \sum\limits_{a'\in  A(I^q_i)}\tilde\beta^{*i}_{I^q_i}(a')u^i(a',\varrho^i_{I^q_i}(\beta^{-I^q_i},\tilde\beta^*), \mu|I^q_i) + \sum\limits_{h\in Z^0(a, I^j_i)}u^i(h)\nu^i_{I^j_i}(h|a,\beta^{-I^j_i},\mu),
\end{array}\end{equation}
which is a continuous function of $(\beta,\mu)$.
Then, as a result of Eq.~(\ref{wsrethmeq8}), we have \begin{equation}\label{Sqrsethm2} \setlength{\abovedisplayskip}{1.2pt}
\setlength{\belowdisplayskip}{1.2pt}w(\beta^*,\mu^*)=\max\limits_{\tilde\beta^i}u^i(a,\varrho^i_{I^j_i}(\beta^{*,-I^j_i},\tilde\beta),\mu^*|I^j_i).\end{equation} Thus,
due to the continuity of $w(\beta, \mu)$ and $\lim\limits_{\ell\to\infty}(\beta^\ell,\mu^\ell)=(\beta^*,\mu^*)$, it follows that \begin{equation} \setlength{\abovedisplayskip}{1.2pt}
\setlength{\belowdisplayskip}{1.2pt}\label{Swsrethm1}\lim\limits_{\ell\to\infty}w(\beta^\ell,\mu^\ell)=w(\beta^*,\mu^*)=\max\limits_{\tilde\beta^i}u^i(a,\varrho^i_{I^j_i}(\beta^{*,-I^j_i},\tilde\beta),\mu^*|I^j_i).\end{equation} Moreover,
\begin{equation}\setlength{\abovedisplayskip}{1.2pt}
\setlength{\belowdisplayskip}{1.2pt}\label{wsrethmeq9}\begin{array}{rl} 
w(\beta^\ell,\mu^\ell)= & \sum\limits_{q\in M(a, I^j_i)} \sum\limits^{\hat h =\langle a_1,\ldots,a_k\rangle\in I^j_i}_{h=\langle\hat h, a_{k+1}=a,\ldots,a_L\rangle
\in I^q_i}\mu^{\ell i}_{I^j_i}(\hat h)\prod\limits_{g=k+1}^{L-1}\beta^{\ell, P(\langle a_1,\ldots,a_{g}\rangle)}_{\langle a_1,\ldots,a_{g}\rangle}(a_{g+1})\\
& \sum\limits_{a'\in  A(I^q_i)}\tilde\beta^{*i}_{I^q_i}(a')u^i(a',\varrho^i_{I^q_i}(\beta^{\ell,-I^q_i},\tilde\beta^*), \mu^\ell |I^q_i)+ \sum\limits_{h\in Z^0(a, I^j_i)}u^i(h)\nu^i_{I^j_i}(h|a,\beta^{\ell,-I^j_i},\mu^\ell)\\
= &  \sum\limits_{q\in M(a, I^j_i)} \frac{\omega(I^q_i|\beta^{\ell,-I^j_i})}{\omega(I^j_i|\beta^\ell)}\sum\limits_{a'\in  A(I^q_i)}\tilde\beta^{*i}_{I^q_i}(a')u^i(a',\varrho^i_{I^q_i}(\beta^{\ell,-I^q_i},\tilde\beta^*), \mu^\ell |I^q_i)\\
& + \sum\limits_{h\in Z^0(a, I^j_i)}u^i(h)\nu^i_{I^j_i}(h|a,\beta^{\ell,-I^j_i},\mu^\ell)\\
 = & u^i(a,\varrho^i_{I^j_i}(\beta^{\ell,-I^j_i},\tilde\beta^*),\mu^\ell|I^j_i).
\end{array} \end{equation}
 Eq.~(\ref{wsrethmeq9}) and Eq.~(\ref{Swsrethm1}) together bring us the conclusion that
\[\setlength{\abovedisplayskip}{1.2pt}
\setlength{\belowdisplayskip}{1.2pt}\begin{array}{rl} & \lim\limits_{\ell\to\infty}w(\beta^\ell,\mu^\ell)=\lim\limits_{\ell\to\infty}u^i(a,\varrho^i_{I^j_i}(\beta^{\ell,-I^j_i},\tilde\beta^*),\mu^\ell|I^j_i)
= u^i(a,\varrho^i_{I^j_i}(\beta^{*,-I^j_i},\tilde\beta^*),\mu^*|I^j_i)\\ 
= & w(\beta^*,\mu^*)=\max\limits_{\tilde\beta^i}u^i(a,\varrho^i_{I^j_i}(\beta^{*,-I^j_i},\tilde\beta),\mu^*|I^j_i).\end{array}\]
This result together with the properties in  Theorem~\ref{ed1wsre} implies that
\begin{equation}\setlength{\abovedisplayskip}{1.2pt}
\setlength{\belowdisplayskip}{1.2pt}\label{wsrethmeq10}
\begin{array}{rl}
 & \max\limits_{\tilde\beta^i_{I^j_i}}\sum\limits_{a\in  A(I^j_i)}\tilde\beta^{i}_{I^j_i}(a)\max\limits_{\tilde\beta^i}u^i(a,\varrho^i_{I^j_i}(\beta^{*,-I^j_i},\tilde\beta), \mu^* |I^j_i)\\
 = & \max\limits_{\tilde\beta^i_{I^j_i}}
\sum\limits_{a\in  A(I^j_i)}\tilde\beta^i_{I^j_i}(a)u^i(a,\varrho^i_{I^j_i}(\beta^{*,-I^j_i},\tilde\beta^*), \mu^* |I^j_i)\\
= & \sum\limits_{a\in  A(I^j_i)}\tilde\beta^{*i}_{I^j_i}(a)u^i(a,\varrho^i_{I^j_i}(\beta^{*,-I^j_i},\tilde\beta^*), \mu^* |I^j_i),\end{array}\end{equation}
which is exactly the desired conclusion in Eq.~(\ref{wsrethmeq3}). Hence it follows from Eq.~(\ref{wsrethmeqA}) that 
\begin{equation}\setlength{\abovedisplayskip}{1.2pt}
\setlength{\belowdisplayskip}{1.2pt}\label{wsrethmeqB}\max\limits_{\beta^i} u^i(\beta^i,\beta^{*,-i},\mu^*|I^j_i)= \sum\limits_{a\in  A(I^j_i)}\beta^{*i}_{I^j_i}(a)u^i(a,\varrho^i_{I^j_i}(\beta^{*,-I^j_i},\tilde\beta^*), \mu^* |I^j_i).\end{equation}
We prove below that, for any $i\in N$ and $j\in M_i$ with ${\cal D}(I^j_i|\beta^*)>0$, 
\begin{equation}\setlength{\abovedisplayskip}{1.2pt}
\setlength{\belowdisplayskip}{1.2pt}\label{wsrethmeq11}
\sum\limits_{a\in  A(I^j_i)}\beta^{*i}_{I^j_i}(a)u^i(a,\varrho^i_{I^j_i}(\beta^{*,-I^j_i},\tilde\beta^*), \mu^* |I^j_i)=\sum\limits_{a\in  A(I^j_i)}\beta^{*i}_{I^j_i}(a)u^i(a,\beta^{*,-I^j_i}, \mu^* |I^j_i).
\end{equation}
{\bf Case (a)}. Consider $i\in N$ and $j\in M_i$ with ${\cal D}(I^j_i|\beta^*)>0$ and $M(I^j_i)=\emptyset$. Then, $\varrho^i_{I^j_i}(\beta^{*-I^j_i},\tilde\beta^*)=\beta^{*-I^j_i}$. Thus, $u^i(a,\varrho^i_{I^j_i}(\beta^{*,-I^j_i},\tilde\beta^*), \mu^* |I^j_i)=u^i(a, \beta^{*,-I^j_i}, \mu^* |I^j_i)$. Therefore, \[\setlength{\abovedisplayskip}{1.2pt}
\setlength{\belowdisplayskip}{1.2pt}
\sum\limits_{a\in  A(I^j_i)}\beta^{*i}_{I^j_i}(a)u^i(a,\varrho^i_{I^j_i}(\beta^{*,-I^j_i},\tilde\beta^*), \mu^* |I^j_i)=\sum\limits_{a\in  A(I^j_i)}\beta^{*i}_{I^j_i}(a)u^i(a,\beta^{*,-I^j_i}, \mu^* |I^j_i).\]
{\bf Case (b)}. Consider $i\in N$ and $j\in M_i$ with ${\cal D}(I^j_i|\beta^*)>0$ and $M(I^j_i)\ne \emptyset$ such that, for any $q\in M(I^j_i)$ with ${\cal D}(\beta^*|I^q_i)>0$, 
\begin{equation}\setlength{\abovedisplayskip}{1.2pt}
\setlength{\belowdisplayskip}{1.2pt}\label{wsrethmeq12}\sum\limits_{a'\in  A(I^q_i)}\beta^{*i}_{I^q_i}(a')u^i(a',\varrho^i_{I^q_i}(\beta^{*,-I^q_i},\tilde\beta^*), \mu^* |I^q_i)=\sum\limits_{a'\in  A(I^q_i)}\beta^{*i}_{I^q_i}(a')u^i(a',\beta^{*,-I^q_i}, \mu^* |I^q_i).\end{equation}
For $a\in A(I^j_i)$ with $\beta^{*i}_{I^j_i}(a)>0$, we have ${\cal D}(\beta^*|I^q_i)=\beta^{*i}_{I^j_i}(a){\cal D}(I^j_i|\beta^*)>0$ for any $q\in M(a, I^j_i)$ and {\footnotesize
\begin{equation}\setlength{\abovedisplayskip}{1.2pt}
\setlength{\belowdisplayskip}{1.2pt}\label{wsrethmeq13}
\begin{array}{rl}
 & u^i(a,\varrho^i_{I^j_i}(\beta^{\ell,-I^j_i},\tilde\beta^*), \mu^\ell |I^j_i)\\
 
 = &  \sum\limits_{q\in M(a, I^j_i)}\sum\limits_{h\in Z}^{h\cap A(I^q_i)\ne\emptyset}u^i(h)\nu^i_{I^j_i}(h|a, \varrho^i_{I^j_i}(\beta^{\ell,-I^j_i},\tilde\beta^*),\mu^\ell) + \sum\limits_{h\in Z^0(a, I^j_i)}u^i(h)\nu^i_{I^j_i}(h|a,\varrho^i_{I^j_i}(\beta^{\ell,-I^j_i},\tilde\beta^*),\mu^\ell)\\
 
  = &  \sum\limits_{q\in M(a, I^j_i)}\frac{1}{\omega(I^j_i|\beta^{\ell})}\sum\limits_{h\in Z}^{h\cap A(I^q_i)\ne\emptyset}u^i(h)\omega(h|a, \varrho^i_{I^j_i}(\beta^{\ell,-I^j_i},\tilde\beta^*)) + \sum\limits_{h\in Z^0(a, I^j_i)}u^i(h)\nu^i_{I^j_i}(h|a,\beta^{\ell,-I^j_i},\mu^\ell)\\
  
  = & \sum\limits_{q\in M(a, I^j_i)} \frac{\omega(I^q_i|\beta^{\ell,-I^j_i})}{\omega(I^j_i|\beta^\ell)}\sum\limits_{h\in Z}^{h\cap A(I^q_i)\ne\emptyset}u^i(h)\frac{\omega(h|a, \varrho^i_{I^j_i}(\beta^{\ell,-I^j_i},\tilde\beta^*))}{\omega(I^q_i|\beta^{\ell,-I^j_i})} + \sum\limits_{h\in Z^0(a, I^j_i)}u^i(h)\nu^i_{I^j_i}(h|a,\beta^{\ell,-I^j_i},\mu^\ell)\\

  = & \sum\limits_{q\in M(a, I^j_i)} \frac{\omega(I^q_i|\beta^{\ell,-I^j_i})}{\omega(I^j_i|\beta^\ell)}\sum\limits_{a'\in  A(I^q_i)}\tilde\beta^{*i}_{I^q_i}(a')u^i(a',\varrho^i_{I^q_i}(\beta^{\ell,-I^q_i},\tilde\beta^*), \mu^\ell |I^q_i)+ \sum\limits_{h\in Z^0(a, I^j_i)}u^i(h)\nu^i_{I^j_i}(h|a,\beta^{\ell,-I^j_i},\mu^\ell).
\end{array}
\end{equation}}

\noindent Then, \begin{equation}\setlength{\abovedisplayskip}{1.2pt}
\setlength{\belowdisplayskip}{1.2pt} \begin{array}{rl}
& u^i(a,\varrho^i_{I^j_i}(\beta^{*,-I^j_i},\tilde\beta^*), \mu^* |I^j_i)=\lim\limits_{\ell\to\infty} u^i(a,\varrho^i_{I^j_i}(\beta^{\ell,-I^j_i},\tilde\beta^*), \mu^\ell |I^j_i)\\
= & \sum\limits_{q\in M(a, I^j_i)} \sum\limits^{\hat h =\langle a_1,\ldots,a_k\rangle\in I^j_i}_{h=\langle\hat h, a_{k+1}=a,\ldots,a_L\rangle
\in I^q_i}\mu^{* i}_{I^j_i}(\hat h)\prod\limits_{g=k+1}^{L-1}\beta^{*, P(\langle a_1,\ldots,a_{g}\rangle)}_{\langle a_1,\ldots,a_{g}\rangle}(a_{g+1})\\
& \sum\limits_{a'\in  A(I^q_i)}\tilde\beta^{*i}_{I^q_i}(a')u^i(a',\varrho^i_{I^q_i}(\beta^{*,-I^q_i},\tilde\beta^*), \mu^* |I^q_i) + \sum\limits_{h\in Z^0(a, I^j_i)}u^i(h)\nu^i_{I^j_i}(h|a,\beta^{*,-I^j_i},\mu^*)\\
= & \sum\limits_{q\in M(a, I^j_i)} \sum\limits^{\hat h =\langle a_1,\ldots,a_k\rangle\in I^j_i}_{h=\langle\hat h, a_{k+1}=a,\ldots,a_L\rangle
\in I^q_i}\mu^{* i}_{I^j_i}(\hat h)\prod\limits_{g=k+1}^{L-1}\beta^{*, P(\langle a_1,\ldots,a_{g}\rangle)}_{\langle a_1,\ldots,a_{g}\rangle}(a_{g+1})\\
& \sum\limits_{a'\in  A(I^q_i)}\beta^{*i}_{I^q_i}(a')u^i(a', \beta^{*,-I^q_i}, \mu^* |I^q_i) + \sum\limits_{h\in Z^0(a, I^j_i)}u^i(h)\nu^i_{I^j_i}(h|a,\beta^{*,-I^j_i},\mu^*),
\end{array}
\end{equation}
where the second equality comes from Eq.~(\ref{wsrethmeq13}) and Eq.~(\ref{wsrethmeqC}) and 
the last equality comes from Eq.~(\ref{wsrethmeq12}). Let 
\begin{equation}\setlength{\abovedisplayskip}{1.2pt}
\setlength{\belowdisplayskip}{1.2pt}\label{wsrethmeqF}\begin{array}{rl} 
f(\beta,\mu)= & \sum\limits_{q\in M(a, I^j_i)} \sum\limits^{\hat h =\langle a_1,\ldots,a_k\rangle\in I^j_i}_{h=\langle\hat h, a_{k+1}=a,\ldots,a_L\rangle
\in I^q_i}\mu^i_{I^j_i}(\hat h)\prod\limits_{g=k+1}^{L-1}\beta^{P(\langle a_1,\ldots,a_{g}\rangle)}_{\langle a_1,\ldots,a_{g}\rangle}(a_{g+1})\\
& \sum\limits_{a'\in  A(I^q_i)}\beta^i_{I^q_i}(a')u^i(a',\beta^{-I^q_i}, \mu|I^q_i) + \sum\limits_{h\in Z^0(a, I^j_i)}u^i(h)\nu^i_{I^j_i}(h|a,\beta^{-I^j_i},\mu).
\end{array}\end{equation}
Then, $f(\beta,\mu)$ is a continuous function of $(\beta, \mu)$ and \begin{equation}\setlength{\abovedisplayskip}{1.2pt}
\setlength{\belowdisplayskip}{1.2pt}\label{Swsrethm3}\lim\limits_{\ell\to\infty}f(\beta^\ell,\mu^\ell)=f(\beta^*,\mu^*)=u^i(a,\varrho^i_{I^j_i}(\beta^{*,-I^j_i},\tilde\beta^*), \mu^* |I^j_i).\end{equation} Moreover,
\begin{equation}\setlength{\abovedisplayskip}{1.2pt}
\setlength{\belowdisplayskip}{1.2pt}\label{wsrethmeq14}\begin{array}{rl} 
f(\beta^\ell,\mu^\ell)= & \sum\limits_{q\in M(a, I^j_i)} \sum\limits^{\hat h =\langle a_1,\ldots,a_k\rangle\in I^j_i}_{h=\langle\hat h, a_{k+1}=a,\ldots,a_L\rangle
\in I^q_i}\mu^{\ell i}_{I^j_i}(\hat h)\prod\limits_{g=k+1}^{L-1}\beta^{\ell, P(\langle a_1,\ldots,a_{g}\rangle)}_{\langle a_1,\ldots,a_{g}\rangle}(a_{g+1})\\
& \sum\limits_{a'\in  A(I^q_i)}\beta^{\ell i}_{I^q_i}(a')u^i(a',\beta^{\ell,-I^q_i}, \mu^\ell|I^q_i) + \sum\limits_{h\in Z^0(a, I^j_i)}u^i(h)\nu^i_{I^j_i}(h|a,\beta^{\ell,-I^j_i},\mu^\ell)\\

= & \sum\limits_{q\in M(a, I^j_i)}  \frac{\omega(I^q_i|\beta^{\ell,-I^j_i})}{\omega(I^j_i|\beta^\ell)}\sum\limits_{a'\in  A(I^q_i)}\beta^{\ell i}_{I^q_i}(a')u^i(a',\beta^{\ell,-I^q_i}, \mu^\ell|I^q_i)\\
& + \sum\limits_{h\in Z^0(a, I^j_i)}u^i(h)\nu^i_{I^j_i}(h|a,\beta^{\ell,-I^j_i},\mu^\ell)\\

 = & u^i(a,\beta^{\ell,-I^j_i},\mu^\ell|I^j_i).
\end{array}\end{equation} 
Thus it follows from Eq.~(\ref{wsrethmeq14}) and  Eq.~(\ref{Swsrethm3}) that  {\small \begin{equation}\setlength{\abovedisplayskip}{1.2pt}
\setlength{\belowdisplayskip}{1.2pt}\label{wsrethmeq15}
u^i(a,\beta^{*,-I^j_i},\mu^*|I^j_i)= \lim\limits_{\ell\to\infty}u^i(a,\beta^{\ell,-I^j_i},\mu^\ell|I^j_i)=\lim\limits_{\ell\to\infty}f(\beta^\ell,\mu^\ell)= f(\beta^*,\mu^*)=u^i(a,\varrho^i_{I^j_i}(\beta^{*,-I^j_i},\tilde\beta^*), \mu^* |I^j_i),
\end{equation}}

\noindent 
Multiplying $\beta^{*i}_{I^j_i}(a)$ to Eq.~(\ref{wsrethmeq15}) yields
\begin{equation}\setlength{\abovedisplayskip}{1.2pt}
\setlength{\belowdisplayskip}{1.2pt}\label{wsrethmeq16}
\sum\limits_{a\in  A(I^j_i)}\beta^{*i}_{I^j_i}(a)u^i(a,\varrho^i_{I^j_i}(\beta^{*,-I^j_i},\tilde\beta^*), \mu^* |I^j_i)=\sum\limits_{a\in  A(I^j_i)}\beta^{*i}_{I^j_i}(a)u^i(a,\beta^{*,-I^j_i}, \mu^* |I^j_i),
\end{equation}
which is precisely the same as Eq.~(\ref{wsrethmeq11}).
This result together with Eq.~(\ref{wsrethmeqB}) brings us the conclusion that, for any $i\in N$ and $j\in M_i$ with ${\cal D}(I^j_i|\beta^*)>0$,
\begin{equation}\setlength{\abovedisplayskip}{1.2pt}
\setlength{\belowdisplayskip}{1.2pt}\label{wsrethmeq17}
\max\limits_{\beta^i} u^i(\beta^i,\beta^{*-i},\mu^*|I^j_i)=\sum\limits_{a\in  A(I^j_i)}\beta^{*i}_{I^j_i}(a)u^i(a,\beta^{*,-I^j_i}, \mu^* |I^j_i)=u^i(\beta^*,\mu^*|I^j_i).
\end{equation}
Hence, $(\beta^*,\mu^*)$ is a WSRE according to Definition~\ref{wsred1}.

$(\Rightarrow)$ Let $(\beta^*,\mu^*)\in\Psi$ be a WSRE according to Definition~\ref{wsred1}. 
For any $i\in N$ and $j\in M_i$ with ${\cal D}(I^j_i|\beta^*)>0$, let $\tilde\beta^{*i}_{I^j_i}=\beta^{*i}_{I^j_i}$. For any $i\in N$ and $j\in M_i$ with ${\cal D}(I^j_i|\beta^*)=0$, we solve through the backward induction
 a sequence of the following linear optimization problems to find $\tilde\beta^{*i}_{I^j_i}$, \begin{equation}\setlength{\abovedisplayskip}{1.2pt}
\setlength{\belowdisplayskip}{1.2pt}
\label{wsrethmeq18}\begin{array}{rl}
\max\limits_{\tilde\beta^i_{I^j_i}} &\sum\limits_{a\in A(I^j_i)}\tilde\beta^i_{I^j_i}(a)u^i(a,\varrho(\beta^{*-I^j_i},\tilde\beta^*),\mu^*| I^j_i)\\
\text{s.t.} & \sum\limits_{a\in A(I^j_i)}\tilde\beta^i_{I^j_i}(a)=1,\;0\le\tilde\beta^i_{I^j_i}(a),\;a\in A(I^j_i).
\end{array}\end{equation}
We next prove that $(\beta^*,\mu^*,\tilde\beta^*)$  meets the properties in Theorem~\ref{ed1wsre}.

\noindent {\bf Case (I)}. Consider $i\in N$ and $j\in M_i$ with ${\cal D}(I^j_i|\beta^*)>0$. As a result of the choice of $\tilde\beta^*$, one can derive in the same way as the proof of Eq.~(\ref{wsrethmeq11}) that
\begin{equation}\setlength{\abovedisplayskip}{1.2pt}
\setlength{\belowdisplayskip}{1.2pt}\label{wsrethmeq19}
\begin{array}{rl}
 & u^i(\varrho^i_{I^j_i}(\beta^{*},\tilde\beta^*), \mu^* |I^j_i)=\sum\limits_{a\in  A(I^j_i)}\beta^{*i}_{I^j_i}(a)u^i(a,\varrho^i_{I^j_i}(\beta^{*,-I^j_i},\tilde\beta^*), \mu^* |I^j_i)\\
 = & \sum\limits_{a\in  A(I^j_i)}\beta^{*i}_{I^j_i}(a)u^i(a,\beta^{*,-I^j_i}, \mu^* |I^j_i)=u^i(\beta^{*}, \mu^* |I^j_i).
 \end{array}
\end{equation}
Suppose that there exist $p\in N$, $g\in M_p$, and $a',a''\in A(I^g_p)$ with ${\cal D}(I^g_p|\beta^*)>0$ such that $\beta^{*p}_{I^{g}_{p}}(a')>0$ and
$u^{p}(a'',\beta^{*-I^{g}_{p}},\mu^*|I^{g}_{p})>u^{p}(a',\beta^{*-I^{g}_{p}},\mu^*|I^{g}_{p})$.
Let $\hat\beta=(\hat\beta^{*i}_{I^j_i}(a):i\in N, j\in M_i, a\in A(I^j_i))$ with $\hat\beta^{i}_{I^j_i}(a)=\beta^{*i}_{I^j_i}(a)$ for all $i\in N$, $j\in M_i$ and $a\in A(I^j_i)$ with the exception of $\hat\beta^{p}_{I^{g}_{p}}(a'')=\beta^{*p}_{I^{g}_{p}}(a')+\beta^{*p}_{I^{g}_{p}}(a'')$ and $\hat\beta^{p}_{I^{g}_{p}}(a')=0$. Then, $u^p(
\hat\beta,\mu^*|I^{g}_{p})>u^p(\beta^{*},
\mu^*|I^{g}_{p})$, which contradicts Definition~\ref{wsred1}. Thus, as a result of Eq.~(\ref{wsrethmeq19}) that $\beta^{*i}_{I^{j}_{i}}(a')=\tilde\beta^{*i}_{I^{j}_{i}}(a')=0$ for any $i\in N$, $j\in M_i$, and $a',a''\in A(I^j_i)$ with ${\cal D}(I^j_i|\beta^*)>0$ and $u^{i}(a'',\varrho^i_{I^j_i}(\beta^{*-I^{j}_{i}},\tilde\beta^*),\mu^*|I^{j}_{i})>u^{i}(a',\varrho^i_{I^j_i}(\beta^{*-I^{j}_{i}},\tilde\beta^*),\mu^*|I^{j}_{i})$. 

\noindent {\bf Case (II)}. Consider $i\in N$ and $j\in M_i$ with ${\cal D}(I^j_i|\beta^*)=0$. It follows from the linear optimization problem~(\ref{wsrethmeq18}) that $\tilde\beta^{*i}_{I^{j}_{i}}(a')=0$ for any $i\in N$, $j\in M_i$, and $a',a''\in A(I^j_i)$ with  $u^{i}(a'',\varrho^i_{I^j_i}(\beta^{*-I^{j}_{i}},\tilde\beta^*),\mu^*|I^{j}_{i})>u^{i}(a',\varrho^i_{I^j_i}(\beta^{*-I^{j}_{i}},\tilde\beta^*),\mu^*|I^{j}_{i})$. 

Cases (I) and (II) together ensure that  $(\beta^*, \mu^*,\tilde\beta^*)$ meets the properties in Theorem~\ref{ed1wsre}. The proof is completed.
\end{proof}

In Thereom~\ref{ed1wsre}, the requirements (i) and (ii) are named as local sequential rationality.

{\setlength{\parskip}{1em}
\noindent\textit{\textbf{Local Rationality at an Information Set $I^j_i$}}: An assessment $(\beta^*,\mu^*)$ together with $\tilde\beta^*$ possesses local rationality at an information set $I^j_i$ if \({\cal D}(I^j_i|\beta^*)
u^i(\beta^{*i}_{I^j_i},\varrho^i_{I^j_i}(\beta^{*-I^j_i},\tilde\beta^*),\mu^*|I^j_i)\ge{\cal D}(I^j_i|\beta^*)$ $u^i(\beta^i_{I^j_i},\varrho^i_{I^j_i}(\beta^{*-I^j_i},\tilde\beta^*), \mu^*| I^j_i)\)
for every $\beta^i_{I^j_i}$ of player $i$. A behavioral strategy profile
$\tilde\beta^*$ together with $(\beta^*,\mu^*)$ possesses local rationality at an information set $I^j_i$ if
\(u^i(\tilde\beta^{*i}_{I^j_i},\varrho^i_{I^j_i}(\beta^{*-I^j_i},\tilde\beta^*),\mu^*|I^j_i)\ge u^i(\tilde\beta^i_{I^j_i},\varrho^i_{I^j_i}(\beta^{*-I^j_i},\tilde\beta^*),\mu^*|I^j_i)\)
for every $\tilde\beta^i_{I^j_i}$ of player $i$. }

{\setlength{\parskip}{1em}
\noindent\textit{\textbf{Local Sequential Rationality}}: An assessment $(\beta^*,\mu^*)$ together with $\tilde\beta^*$ possesses local sequential rationality if it meets the local rationality at every information set. A behavioral strategy profile $\tilde\beta^*$ together with $(\beta^*,\mu^*)$  possesses local sequential rationality if it meets the local rationality at every information set. }

A direct application of Theorem~\ref{ed1wsre} yields a polynomial system to serve as a necessary and sufficient condition for determining whether a consistent assessment  is a WSRE or not.
\begin{theorem}\label{nscwsrethm1}{\em $(\beta^*,\mu^*)\in\Psi$ is a WSRE if and only if there exists a vector $(\tilde\beta^*,\lambda^*,\tilde\lambda^*,\zeta^*,\tilde\zeta^*)$ together with $(\beta^*,\mu^*)$ satisfying the polynomial system,
\begin{equation}\label{nscwsre1}\setlength{\abovedisplayskip}{1.2pt}
\setlength{\belowdisplayskip}{1.2pt}
\begin{array}{l}
{\cal D}(I^j_i|\beta)u^i(a,\varrho^i_{I^j_i}(\beta^{-I^j_i},\tilde\beta),\mu|I^j_i)+\lambda^i_{I^j_i}(a)-\zeta^i_{I^j_i}=0,\;i\in N, j\in M_i, a\in A(I^j_i),\\
u^i(a,\varrho^i_{I^j_i}(\beta^{-I^j_i},\tilde\beta),\mu|I^j_i)+\tilde\lambda^i_{I^j_i}(a)-\tilde\zeta^i_{I^j_i}=0,\;i\in N, j\in M_i, a\in A(I^j_i),\\
\sum\limits_{a\in A(I^j_i)}\beta^i_{I^j_i}(a)=1,\;\sum\limits_{a\in A(I^j_i)}\tilde\beta^i_{I^j_i}(a)=1,\;i\in N, j\in M_i,\\
\beta^i_{I^j_i}(a)\lambda^i_{I^j_i}(a)=0,\;0\le\beta^i_{I^j_i}(a),\;0\le\lambda^i_{I^j_i}(a),\;i\in N, j\in M_i, a\in A(I^j_i),\\
\tilde\beta^i_{I^j_i}(a)\tilde\lambda^i_{I^j_i}(a)=0,\;0\le\tilde\beta^i_{I^j_i}(a),\;0\le\tilde\lambda^i_{I^j_i}(a),\;i\in N, j\in M_i, a\in A(I^j_i).
\end{array}
\end{equation}
}
\end{theorem}
\begin{proof} ($\Rightarrow$). Let $(\beta^*,\mu^*)\in\Psi$ be a WSRE. We define $\zeta^{*i}_{I^j_i}=\max\limits_{a\in A(I^j_i)}{\cal D}(I^j_i|\beta^*)u^i(a$, $\varrho^i_{I^j_i}(\beta^{*-I^j_i},\tilde\beta^*), \mu^*|I^j_i)$, $\tilde\zeta^{*i}_{I^j_i}=\max\limits_{a\in A(I^j_i)}u^i(a,\varrho^i_{I^j_i}(\beta^{*-I^j_i},\tilde\beta^*),\mu^*| I^j_i)$, 
$\lambda^{*i}_{I^j_i}(a)=\zeta^{*i}_{I^j_i}-{\cal D}(I^j_i|\beta^*)u^i(a,\varrho^i_{I^j_i}(\beta^{*-I^j_i},\tilde\beta^*)$, $\mu^*|I^j_i)$, and $\tilde\lambda^{*i}_{I^j_i}(a)= \tilde\zeta^{*i}_{I^j_i}-u^i(a,\varrho^i_{I^j_i}(\beta^{*-I^j_i},\tilde\beta^*),\mu^*| I^j_i)$. Suppose that there exists some $a'\in A(I^j_i)$ such that $\lambda^{*i}_{I^j_i}(a')>0$. Then there exists $a''\in A(I^j_i)$ such that  ${\cal D}(I^j_i|\beta^*)u^i(a'',\varrho^i_{I^j_i}(\beta^{*-I^j_i},\tilde\beta^*), \mu^*|I^j_i)>{\cal D}(I^j_i|\beta^*)u^i(a',\varrho^i_{I^j_i}(\beta^{*-I^j_i},\tilde\beta^*), \mu^*|I^j_i)$. Thus it follows from Theorem~\ref{ed1wsre} that $\beta^{*i}_{I^j_i}(a')=0$.  One can show in a similar way that $\tilde\beta^{*i}_{I^j_i}(a')=0$ if  $\tilde\lambda^{*i}_{I^j_i}(a')>0$. Therefore, $(\beta^*,\tilde\beta^*,\mu^*,\zeta^*,\tilde\zeta^*,\lambda^*,\tilde\lambda^*)$ satisfies the system~(\ref{nscwsre1}).

($\Leftarrow$). Let $(\beta^*,\tilde\beta^*,\mu^*,\lambda^*,\tilde\lambda^*,\zeta^*,\tilde\zeta^*)$ be a solution to the system~(\ref{nscwsre1}). Multiplying $\beta^{*i}_{I^j_i}(a)$ to the first group of equations and $\tilde\beta^{*i}_{I^j_i}(a)$ to the second group of equations in the system~(\ref{nscwsre1}) and taking the sum over $A(I^j_i)$, we get 
$\zeta^{*i}_{I^j_i}=\sum\limits_{a\in A(I^j_i)}\beta^{*i}_{I^j_i}(a){\cal D}(I^j_i|\beta^*)u^i(a,\varrho^i_{I^j_i}(\beta^{*-I^j_i},\tilde\beta^*), \mu^*|I^j_i)$ and $\tilde\zeta^{*i}_{I^j_i}=\sum\limits_{a\in A(I^j_i)}\tilde\beta^{*i}_{I^j_i}(a)u^i(a,\varrho^i_{I^j_i}(\beta^{*-I^j_i},\tilde\beta^*), \mu^*|I^j_i)$. 
As a result of $\lambda^*\ge 0$, we have $\zeta^{*i}_{I^j_i}=\max\limits_{a\in A(I^j_i)}{\cal D}(I^j_i|\beta^*)u^i(a$, $\varrho^i_{I^j_i}(\beta^{*-I^j_i},\tilde\beta^*), \mu^*|I^j_i)$ and $\tilde\zeta^{*i}_{I^j_i}=\max\limits_{a\in A(I^j_i)}u^i(a,\varrho^i_{I^j_i}(\beta^{*-I^j_i},\tilde\beta^*),\mu^*|I^j_i)$. Thus, (i). $\beta^{*i}_{I^j_i}(a')=0$ whenever ${\cal D}(I^j_i|\beta^*)u^i(a'',\varrho^i_{I^j_i}(\beta^{*-I^j_i},\tilde\beta^*), \mu^*|I^j_i)>{\cal D}(I^j_i|\beta^*)u^i(a',\varrho^i_{I^j_i}(\beta^{*-I^j_i},\tilde\beta^*), \mu^*|I^j_i)$
since $\lambda^{*i}_{I^j_i}(a')>0$, and  (ii). $\tilde\beta^{*i}_{I^j_i}(a')=0$ whenever $u^i(a'',\varrho^i_{I^j_i}(\beta^{*-I^j_i},\tilde\beta^*),\mu^*|I^j_i)>u^i(a',\varrho^i_{I^j_i}(\beta^{*-I^j_i},\tilde\beta^*),\mu^*|I^j_i)$ since $\tilde\lambda^{*i}_{I^j_i}(a')>0$. Therefore it follows from Theorem~\ref{ed1wsre} that $(\beta^*,\mu^*)$ is a WSRE. The proof is completed.
\end{proof}

To further enhance the applications of WSRE, one can exploit Theorem~\ref{nscwsrethm1} to develop path-following methods to compute WSREs. Such a task is carried out in Section~\ref{spwsre}.
\section{A Characterization of WSRE through $\varepsilon$-Perfect $\gamma$-WSRE}

To employ Theorem~\ref{ed1wsre} to find a WSRE, one needs to construct a convergent sequence of totally mixed assessments. Nevertheless, Theorem~\ref{ed1wsre} provides insufficient information on constructing such a sequence. 
To bridge this gap, we introduce the concept of $\varepsilon$-perfect $\gamma$-WSRE.
\begin{definition}[\bf An Equivalent Definition of WSRE through $\varepsilon$-Perfect $\gamma$-WSRE with Local Sequential Rationality]\label{edAwsre}
{\em For any given $\gamma>0$ and $\varepsilon>0$, a totally mixed assessment $(\beta(\gamma,\varepsilon),\mu(\gamma,\varepsilon))$ together with $\tilde\beta(\gamma,\varepsilon)$ constitutes an $\varepsilon$-perfect $\gamma$-WSRE if $(\beta(\gamma,\varepsilon),\mu(\gamma,\varepsilon),\tilde\beta(\gamma,\varepsilon))$  satisfies the properties:\newline
 (i). $\beta^i_{I^j_i}(\gamma,\varepsilon; a')\le \varepsilon$ for any $i\in N$, $j\in M_i$ and $a',a''\in A(I^j_i)$ with 
 \[\setlength{\abovedisplayskip}{1.2pt}
\setlength{\belowdisplayskip}{1.2pt}{\cal D}(I^j_i|\beta(\gamma,\varepsilon))(u^i(a'',\varrho^i_{I^j_i}(\beta^{-I^j_i}(\gamma,\varepsilon),\tilde\beta(\gamma,\varepsilon)),\mu(\gamma,\varepsilon)|I^j_i)
-u^i(a',\varrho^i_{I^j_i}(\beta^{-I^j_i}(\gamma,\varepsilon),\tilde\beta(\gamma,\varepsilon)),\mu(\gamma,\varepsilon)|I^j_i))>\gamma,\] 
(ii). $\tilde\beta^i_{I^j_i}(\gamma,\varepsilon; a')=0$ for any $i\in N$, $j\in M_i$ and $a',a''\in A(I^j_i)$ with \[\setlength{\abovedisplayskip}{1.2pt}
\setlength{\belowdisplayskip}{1.2pt} u^i(a'',\varrho^i_{I^j_i}(\beta^{-I^j_i}(\gamma,\varepsilon),\tilde\beta(\gamma,\varepsilon)),\mu(\gamma,\varepsilon)|I^j_i)
-u^i(a',\varrho^i_{I^j_i}(\beta^{-I^j_i}(\gamma,\varepsilon),\tilde\beta(\gamma,\varepsilon)),\mu(\gamma,\varepsilon)|I^j_i)>\gamma,\]
(iii). $\mu(\gamma,\varepsilon)=(\mu^i_{I^j_i}(\gamma,\varepsilon;h):i\in N,j\in M_i,h\in I^j_i)$ with $\mu^i_{I^j_i}(\gamma,\varepsilon;h)=\frac{\omega(h|\beta(\gamma,\varepsilon))}{\omega(I^j_i|\beta(\gamma,\varepsilon))}$ for $h\in I^j_i$. 

\noindent For any given $\gamma>0$, $(\beta(\gamma),\mu(\gamma),\tilde\beta(\gamma))$ is a perfect $\gamma$-WSRE if it is a limit point of a sequence $\{(\beta(\gamma,\varepsilon_k), \mu(\gamma,\varepsilon_k), \tilde\beta(\gamma,\varepsilon_k)), k=1,2,\ldots\}$ with $\varepsilon_k>0$ and $\lim\limits_{k\to\infty}\varepsilon_k=0$, where $(\beta(\gamma,\varepsilon_k), \mu(\gamma,\varepsilon_k), \tilde\beta(\gamma,\varepsilon_k))$ is an $\varepsilon_k$-perfect $\gamma$-WSRE for every $k$.
$(\beta^*,\mu^*)$ is a WSRE if $(\beta^*,\mu^*,\tilde\beta^*)$ is a limit point of a sequence $\{(\beta(\gamma_q),\mu(\gamma_q), \tilde\beta(\gamma_q)),q=1,2,\ldots\}$ with $\gamma_q>0$ and $\lim\limits_{q\to\infty}\gamma_q=0$, where $(\beta(\gamma_q),\mu(\gamma_q), \tilde\beta(\gamma_q))$ is a perfect $\gamma_q$-WSRE for every $q$.}
\end{definition}
\begin{theorem}\label{edwsrethmA}{\em Definition~\ref{edAwsre} and Definition~\ref{ed1wsre} of WSRE are equivalent.
}
\end{theorem}
\begin{proof} $(\Rightarrow)$. We denote by $(\beta^*,\mu^*)$ a WSRE according to Definition~\ref{edAwsre}. Then there exists a convergent sequence $\{(\beta(\gamma_q),\mu(\gamma_q), \tilde\beta(\gamma_q)),q=1,2,\ldots\}$ with $\gamma_q>0$ and $\lim\limits_{q\to\infty}\gamma_q=0$ such that  $(\beta^*,\mu^*,\tilde\beta^*)=\lim\limits_{q\to\infty}(\beta(\gamma_q),\mu(\gamma_q),\tilde\beta(\gamma_q))$, where $(\beta(\gamma_q),\mu(\gamma_q),\tilde\beta(\gamma_q))$ is a perfect $\gamma_q$-WSRE for every $q$ and the limit point of a convergent sequence $\{(\beta(\gamma_q,\varepsilon_k), \mu(\gamma_q,\varepsilon_k), \tilde\beta(\gamma_q,\varepsilon_k)), k=1,2,\ldots\}$  satisfying that $\varepsilon_k>0$, $\lim\limits_{k\to\infty}\varepsilon_k=0$, and $(\beta(\gamma_q,\varepsilon_k), \mu(\gamma_q,\varepsilon_k), \tilde\beta(\gamma_q,\varepsilon_k))$ is an $\varepsilon_k$-perfect $\gamma_q$-WSRE for every $k$. 
Let $\{\delta_q>0,q=1,2,\ldots\}$ be a convergent sequence with $\lim\limits_{q\to\infty}\delta_q=0$. For each $q$, we choose $(\beta(\gamma_q, \varepsilon_{k_q}),\mu(\gamma_q, \varepsilon_{k_q}), \tilde\beta(\gamma_q, \varepsilon_{k_q}))$ such that $\|(\beta(\gamma_q, \varepsilon_{k_q}),\mu(\gamma_q, \varepsilon_{k_q}), \tilde\beta(\gamma_q, \varepsilon_{k_q}))-(\beta(\gamma_q),\mu(\gamma_q),\tilde\beta(\gamma_q))\|<\delta_q$. Then, $\{(\beta(\gamma_q, \varepsilon_{k_q}),\mu(\gamma_q, \varepsilon_{k_q}), \tilde\beta(\gamma_q, \varepsilon_{k_q})), q=1,2,\ldots\}$ is a convergent sequence with $\lim\limits_{q\to\infty}(\beta(\gamma_q, \varepsilon_{k_q})$, $\mu(\gamma_q, \varepsilon_{k_q}), \tilde\beta(\gamma_q, \varepsilon_{k_q}))=(\beta^*,\mu^*,\tilde\beta^*)$.  For some $i\in N$ and $j\in M_i$,
suppose that there exist $a',a''\in A(I^j_i)$ such that  $\beta^{*i}_{I^j_i}(a')>0$ and $r_0={\cal D}(I^j_i|\beta^*)(u^i(a'', \varrho^i_{I^j_i}(\beta^{*-I^j_i},\tilde\beta^*),\mu^*|I^j_i)-u^i(a', \varrho^i_{I^j_i}(\beta^{*-I^j_i},\tilde\beta^*),\mu^*|I^j_i))>0$. Then, as a result of the continuity of ${\cal D}(I^j_i|\beta)$ and $u^i(a, \varrho^i_{I^j_i}(\beta^{-I^j_i},\tilde\beta),\mu|I^j_i)$, there exists a sufficiently large $Q_0$ such that 
 ${\cal D}(I^j_i|\beta(\gamma_q,\varepsilon_{k_q}))(u^i(a'', \varrho^i_{I^j_i}(\beta^{-I^j_i}(\gamma_q, \varepsilon_{k_q}),\tilde\beta(\gamma_q,\varepsilon_{k_q})),\mu(\gamma_q, \varepsilon_{k_q})|I^j_i)-u^i(a', \varrho^i_{I^j_i}(\beta^{-I^j_i}(\gamma_q, \varepsilon_{k_q}),\tilde\beta(\gamma_q,\varepsilon_{k_q})),\mu(\gamma_q, \varepsilon_{k_q})|I^j_i))\ge r_0/2$ for all $q\ge Q_0$. Since $\lim\limits_{q\to\infty}\gamma_q=0$, there exists a sufficiently large $Q_1\ge Q_0$ such that $\gamma_q<r_0/2$ for all $q\ge Q_1$. This implies $\beta^i_{I^j_i}(\gamma_q, \varepsilon_{k_q};a')\le\varepsilon_{k_q}$ for all $q\ge Q_1$. Thus, $0<\beta^{*i}_{I^j_i}(a')=\lim\limits_{q\to\infty}\beta^i_{I^j_i}(\gamma_q, \varepsilon_{k_q};a')=0$. A contradiction occurs. 
Therefore, $(\beta^*,\mu^*)$ is a WSRE according to Theorem~\ref{ed1wsre}.

$(\Rightarrow)$. We denote by $(\beta^*,\mu^*)$ a WSRE according to Definition~\ref{ed1wsre}. Then there exists $\tilde\beta^*$ together with $(\beta^*,\mu^*)$ satisfying the properties of Theorem~\ref{ed1wsre}. As a result of the consistency, there exists a convergent sequence, $\{(\beta^\ell,\mu^\ell),\ell=1,2,\ldots\}$, such that $\beta^*=\lim\limits_{\ell\to\infty}\beta^\ell$ and $\mu^*=\lim\limits_{\ell\to\infty}\mu^\ell$, where $\beta^\ell=(\beta^{\ell i}_{I^j_i}(a):i\in N,j\in M_i,a\in A(I^j_i))>0$ and $\mu^\ell=(\mu^{\ell i}_{I^j_i}(h):i\in N, j\in M_i, h\in I^j_i)=\mu(\beta^\ell)=(\mu^i_{I^j_i}(h|\beta^\ell):i\in N,j\in M_i,h\in I^j_i)$ with  $\mu^{\ell i}_{I^j_i}(h)=\mu^i_{I^j_i}(h|\beta^\ell)=\frac{\omega(h|\beta^\ell)}{\omega(I^j_i|\beta^\ell)}$.  For $i\in N$ and $j\in M_i$, let $A^0(I^j_i)=\{a\in A(I^j_i)|\beta^{*i}_{I^j_i}(a)=0\}$, $A^+(I^j_i)=\{a\in A(I^j_i)|{\cal D}(I^j_i|\beta^*)(\max\limits_{\tilde a\in A(I^j_i)}u^i(\tilde a,\varrho^i_{I^j_i}(\beta^{*-I^j_i},\tilde\beta^*),\mu^*|I^j_i)-u^i(a, \varrho^i_{I^j_i}(\beta^{*-I^j_i},\tilde\beta^*),\mu^*|I^j_i))>0\}$, and $\tilde A^+(I^j_i)=\{a\in A(I^j_i)|\max\limits_{\tilde a\in A(I^j_i)}u^i(\tilde a,\varrho^i_{I^j_i}(\beta^{*-I^j_i},\tilde\beta^*),\mu^*|I^j_i)-u^i(a, \varrho^i_{I^j_i}(\beta^{*-I^j_i},\tilde\beta^*),\mu^*|I^j_i)>0\}$. When $\mathop{\cup}\limits_{i\in N,\;j\in M_i}A^+(I^j_i)=\emptyset$, we have $\beta^*>0$ and consequently, $(\beta^*,\mu^*)$ is a WSRE according to Definition~\ref{edAwsre}. We assume that $\mathop{\cup}\limits_{i\in N,\;j\in M_i}A^+(I^j_i)\ne\emptyset$.
Let $
\gamma_0=
 \frac{1}{2}\min\limits_{i\in N,\;j\in M_i}\min\{\min\limits_{a\in A^+(I^j_i)}{\cal D}(I^j_i|\beta^*)(\max\limits_{\tilde a\in A(I^j_i)}u^i(\tilde a,\varrho^i_{I^j_i}(\beta^{*-I^j_i},\tilde\beta^*),\mu^*|I^j_i)-u^i(a, \varrho^i_{I^j_i}(\beta^{*-I^j_i},\tilde\beta^*),\mu^*|I^j_i)), \min\limits_{a\in \tilde A^+(I^j_i)}\max\limits_{\tilde a\in A(I^j_i)}u^i(\tilde a,\varrho^i_{I^j_i}(\beta^{*-I^j_i},\tilde\beta^*),\mu^*|I^j_i)-u^i(a, \varrho^i_{I^j_i}(\beta^{*-I^j_i},\tilde\beta^*),\mu^*|I^j_i)\}$.
We take $\{\varepsilon_\ell,\ell=1,2,\ldots\}$ to be a convergent sequence with $\varepsilon_\ell\ge\max\limits_{i\in N,\;j\in M_i}\max\limits_{a\in A^0(I^j_i)}\beta^{\ell i}_{I^j_i}(a)$ and $\lim\limits_{\ell\to\infty}\varepsilon_\ell=0$. 
Let $\beta(\gamma_0,\varepsilon_\ell)=(\beta^i_{I^j_i}(\gamma_0,\varepsilon_\ell; a): i\in N,j\in M_i,a\in A(I^j_i))$ with $\beta^i_{I^j_i}(\gamma_0,\varepsilon_\ell)=\beta^{\ell i}_{I^j_i}$, $\mu(\gamma_0,\varepsilon_\ell)=(\mu^i_{I^j_i}(\gamma_0,\varepsilon_\ell; h): i\in N,j\in M_i,h\in I^j_i)$ with $\mu^i_{I^j_i}(\gamma_0,\varepsilon_\ell; h)=\mu^{\ell i}_{I^j_i}(h)$, and $\tilde\beta(\gamma_0,\varepsilon_\ell)=\tilde\beta^*$. We know from Definition~\ref{ed1wsre} that $\beta^{*i}_{I^j_i}(a')=0$ for any $i\in N$, $j\in M_i$ and $a',a''\in A(I^j_i)$ with ${\cal D}(I^ji|\beta^*)(u^i(a'',\varrho^i_{I^j_i}(\beta^{*-I^j_i},\tilde\beta^*),\mu^*|I^j_i)-u^i(a',\varrho^i_{I^j_i}(\beta^{*-I^j_i},\tilde\beta^*),\mu^*|I^j_i))>0$ and $\tilde\beta^{*i}_{I^j_i}(a')=0$ for any $i\in N$, $j\in M_i$ and $a',a''\in A(I^j_i)$ with $u^i(a'',\varrho^i_{I^j_i}(\beta^{*-I^j_i},\tilde\beta^*),\mu^*|I^j_i)-u^i(a',\varrho^i_{I^j_i}(\beta^{*-I^j_i},\tilde\beta^*),\mu^*|I^j_i)>0$. Thus it follows from the continuity of ${\cal D}(I^j_i|\beta)$ and $u^i(a,\varrho^i_{I^j_i}(\beta^{-I^j_i},\tilde\beta),\mu|I^j_i)$   that there exists a sufficiently large $L_0$ such that, for any $\ell\ge L_0$, $(\beta(\gamma_0,\varepsilon_\ell),\mu(\gamma_0,\varepsilon_\ell), \tilde\beta(\gamma_0,\varepsilon_\ell))$ is an $\varepsilon_\ell$-perfect $\gamma_0$-WSRE  according to Definition~\ref{edAwsre}. Therefore, $(\beta^*,\mu^*)$ is a WSRE according to Definition~\ref{edAwsre}. The proof is completed.
\end{proof}
As a direct application of Definition~\ref{edAwsre}, we acquire a proof to the existence of a WSRE in an extensive-form game with perfect recall.
\begin{theorem} {\em An extensive-form game with perfect recall always has a WSRE.}
\end{theorem}
\begin{proof} Let $\{\gamma_q>0, q=1,2,\ldots\}$ be a convergent sequence with  $\lim\limits_{q\to\infty}\gamma_q=0$ and $\{\varepsilon_k>0,k=1,2,\ldots\}$ a convergent sequence with $\lim\limits_{k\to\infty}\varepsilon_k=0$. We denote $m_0=\sum\limits_{i\in N}\sum\limits_{j\in M_i}|A(I^j_i)|$.
Let $\delta(\varepsilon_k)=\frac{\varepsilon^2_k}{m_0}$ and $\triangle(\varepsilon_k)=\mathop{\times}\limits_{i\in N,\;j\in M_i}\triangle^i_{I^j_i}(\varepsilon_k)$, where $\triangle^i_{I^j_i}(\varepsilon_k)=\{\beta^i_{I^j_i}(\varepsilon_k)\in\triangle^i_{I^j_i}|\beta^i_{I^j_i}(\varepsilon_k;a)\ge\delta(\varepsilon_k),a\in A(I^j_i)\}$. For $i\in N$ and $j\in M_i$, we define a point-to-set mapping $F^i_{I^j_i}(\cdot | \gamma_q):\triangle(\varepsilon_k)\times\Xi\times\triangle\to\triangle^i_{I^j_i}(\varepsilon_k)\times\Xi^i_{I^j_i}\times\triangle^i_{I^j_i}$ by
{\small\[\setlength{\abovedisplayskip}{1.2pt}
\setlength{\belowdisplayskip}{1.2pt}
\begin{array}{rl}
 & F^i_{I^j_i}(\beta(\varepsilon_k),\mu(\varepsilon_k),\tilde\beta(\varepsilon_k)|\gamma_q)\\
 = & \left\{\begin{array}{l}
(\beta^{*i}_{I^j_i}(\varepsilon_k),\mu^{*i}_{I^j_i}(\varepsilon_k),\tilde\beta^{*i}_{I^j_i}(\varepsilon_k))\\
\in\triangle^i_{I^j_i}(\varepsilon)\times\Xi^i_{I^j_i}\times\triangle^i_{I^j_i}
\end{array}\left|\begin{array}{ll}
\text{$\beta^{*i}_{I^j_i}(\varepsilon_k; a')\le\varepsilon_k$ for any $a',a''\in A(I^j_i)$ with}\\

\text{${\cal D}(I^j_i|\beta(\varepsilon_k))(u^i(a'',\varrho^i_{I^j_i}(\beta^{-I^j_i}(\varepsilon_k), \tilde\beta(\varepsilon_k)), \mu(\varepsilon_k)|I^j_i)$}\\

\text{$-u^i(a',\varrho^i_{I^j_i}(\beta^{-I^j_i}(\varepsilon_k), \tilde\beta(\varepsilon_k)),\mu(\varepsilon_k)|I^j_i))> \gamma_q$.}\\

\text{$\tilde\beta^{*i}_{I^j_i}(\varepsilon_k; a')=0$ for any $a',a''\in A(I^j_i)$ with}\\

\text{$u^i(a'',\varrho^i_{I^j_i}(\beta^{-I^j_i}(\varepsilon_k), \tilde\beta(\varepsilon_k)),\mu(\varepsilon_k)|I^j_i)$}\\

\text{$-u^i(a',\varrho^i_{I^j_i}(\beta^{-I^j_i}(\varepsilon_k), \tilde\beta(\varepsilon_k)),\mu(\varepsilon_k)|I^j_i)> \gamma_q$}
\end{array}\right.
\right\},
\end{array}\]}where $\mu^*(\varepsilon_k)=(\mu^{*i}_{I^j_i}(\varepsilon_k; h):i\in N,j\in M_i,h\in I^j_i)$ with $\mu^{*i}_{I^j_i}(\varepsilon_k; h)=\mu^{*i}_{I^j_i}(h|\beta(\varepsilon_k))=\frac{\omega(h|\beta(\varepsilon_k))}{\omega(I^j_i|\beta(\varepsilon_k))}$.
Clearly, for any $(\beta(\varepsilon_k),\mu(\varepsilon_k),\tilde\beta(\varepsilon_k))\in\triangle(\varepsilon_k)\times\Xi\times\triangle$, the points in $F^i_{I^j_i}(\beta(\varepsilon_k),\mu(\varepsilon_k),\tilde\beta(\varepsilon_k)|\gamma_q)$ satisfy a finite number of linear inequalities. Thus, $F^i_{I^j_i}(\beta(\varepsilon_k),\mu(\varepsilon_k),\tilde\beta(\varepsilon_k)|\gamma_q)$ is a nonempty convex and compact set. The continuity of $u^i(a,\varrho^i_{I^j_i}(\beta^{-I^j_i}(\varepsilon_k), \tilde\beta(\varepsilon_k)),\mu(\varepsilon_k)|I^j_i)$  on $\triangle(\varepsilon_k)\times\triangle\times\Xi$ and ${\cal D}(I^j_i|\beta(\varepsilon_k))$ on $\triangle(\varepsilon_k)$ implies that $F^i_{I^j_i}(\beta(\varepsilon_k),\mu(\varepsilon_k),\tilde\beta(\varepsilon_k)|\gamma_q)$ is upper-semicontinuous on $\triangle(\varepsilon_k)\times\Xi\times\triangle$. Let $F(\beta(\varepsilon_k),\mu(\varepsilon_k),\tilde\beta(\varepsilon_k)|\gamma_q)=\mathop{\times}\limits_{i\in N,\;j\in M_i}F^i_{I^j_i}(\beta(\varepsilon_k),\mu(\varepsilon_k),\tilde\beta(\varepsilon_k)|\gamma_q)$. Then, $F(\cdot|\gamma_q):\triangle(\varepsilon_k)\times\Xi\times\triangle\to\triangle(\varepsilon_k)\times\Xi\times\triangle$ meets all the requirements of Kakutani's fixed point theorem. Thus there exists some $(\beta^*(\gamma_q, \varepsilon_k), \mu^*(\gamma_q, \varepsilon_k), \tilde\beta^*(\gamma_q, \varepsilon_k))\in \triangle(\varepsilon_k)\times\Xi\times\triangle$ such that $(\beta^*(\gamma_q, \varepsilon_k), \mu^*(\gamma_q, \varepsilon_k), \tilde\beta^*(\gamma_q, \varepsilon_k))\in F(\beta^*(\gamma_q, \varepsilon_k), \mu^*(\gamma_q, \varepsilon_k), \tilde\beta^*(\gamma_q, \varepsilon_k)|\gamma_q)$. Since $\{(\beta^*(\gamma_q, \varepsilon_k), \mu^*(\gamma_q, \varepsilon_k)$, $\tilde\beta^*(\gamma_q, \varepsilon_k)),k=1,2,\ldots\}$ is a bounded sequence, it has a convergent subsequence. For convenience, we still denote such a convergent subsequence by $\{(\beta^*(\gamma_q, \varepsilon_k), \mu^*(\gamma_q, \varepsilon_k), \tilde\beta^*(\gamma_q, \varepsilon_k)),k=1,2,\ldots\}$.  Let $(\beta^*(\gamma_q),\mu^*(\gamma_q),\tilde\beta^*(\gamma_q))=\lim\limits_{k\to\infty}(\beta^*(\gamma_q, \varepsilon_k), \mu^*(\gamma_q, \varepsilon_k), \tilde\beta^*(\gamma_q, \varepsilon_k))$. Since $\{(\beta^*(\gamma_q),\mu^*(\gamma_q)\ \tilde\beta^*(\gamma_q)),q=1,2,\ldots\}$ is a bounded sequence, it has a convergent subsequence. For convenience, we still denote such a convergent subsequence by $\{(\beta^*(\gamma_q),\mu^*(\gamma_q),\tilde\beta^*(\gamma_q)),q=1,2,\ldots\}$. Let $(\beta^*,\mu^*,\tilde\beta^*)=\lim\limits_{q\to\infty}(\beta^*(\gamma_q),\mu^*(\gamma_q),\tilde\beta^*(\gamma_q))$.  It follows from the above construction that $(\beta^*(\gamma_q, \varepsilon_k), \mu^*(\gamma_q, \varepsilon_k), \tilde\beta^*(\gamma_q, \varepsilon_k))$ is an $\varepsilon_k$-perfect $\gamma_q$-WSRE for every $k$. Therefore, $(\beta^*,\mu^*,\tilde\beta^*)$ is a WSRE according to Definition~\ref{edAwsre}. The proof is completed.
\end{proof}

We illustrate with two examples  how one can employ Definition~\ref{edAwsre} to analytically determine all WSREs in small extensive-form games in Section~\ref{examples}.

With Definition~\ref{edAwsre}, we have a difficulty to acquire a polynomial system as a necessary and sufficient condition to determine whether a given $(\beta,\mu,\tilde\beta)$ is an $\varepsilon$-perfect $\gamma$-WSRE or not. To resolve this difficulty, we need to present Definition~\ref{edAwsre} in an alternative way.
For any sufficiently small $\varepsilon>0$, let $\eta(\varepsilon)=(\eta^i_{I^j_i}(\varepsilon):i\in N,j\in M_i)$  be a vector with $\eta^i_{I^j_i}(\varepsilon)=(\eta^{i}_{I^j_i}(\varepsilon; a):a\in A(I^j_i))^\top$ such that $0<\eta^{i}_{I^j_i}(\varepsilon;a)\le \varepsilon$ and $\tau^i_{I^j_i}(\varepsilon)=\sum\limits_{a\in A(I^j_i)}\eta^{i}_{I^j_i}(\varepsilon; a)< 1$. 
We specify $\varpi(\beta,\eta(\varepsilon))=(\varpi(\beta^p_{I^\ell_p},\eta(\varepsilon)):p\in N,\ell\in M_p)$, where $\varpi(\beta^p_{I^\ell_p},\eta(\varepsilon))=(\varpi(\beta^p_{I^\ell_p}(a),\eta(\varepsilon)):a\in A(I^\ell_p))^\top$ with \[
\setlength{\abovedisplayskip}{1.2pt} 
\setlength{\belowdisplayskip}{1.2pt}
\varpi(\beta^p_{I^\ell_p}(a),\eta(\varepsilon))=(1-\tau^p_{I^\ell_p}(\varepsilon))\beta^p_{I^\ell_p}(a)+\eta^{p}_{I^\ell_p}(\varepsilon;a).\]
Given $\varpi(\beta,\eta(\varepsilon))$, one can readily verify that Definition~\ref{edAwsre} can be equivalently rewritten as follows.
\begin{definition}[\bf An Equivalent Definition of WSRE through $\varepsilon$-Perfect $\gamma$-WSRE with Local Sequential Rationality and Separation of Perturbation from Strategies]\label{edBwsre}
{\em For any given $\gamma>0$ and $\varepsilon>0$, a totally mixed assessment $(\varpi(\beta(\gamma),\eta(\varepsilon)),\mu(\gamma,\varepsilon))$ together with $\tilde\beta(\gamma,\varepsilon)$ constitutes an $\varepsilon$-perfect $\gamma$-WSRE if $(\varpi(\beta(\gamma),\eta(\varepsilon)),\mu(\gamma,\varepsilon), \tilde\beta(\gamma,\varepsilon))$  satisfies the properties:\newline
 (i). $\beta^i_{I^j_i}(\gamma; a')=0$ for any $i\in N$, $j\in M_i$ and $a',a''\in A(I^j_i)$ with 
 \[\setlength{\abovedisplayskip}{1.2pt}
\setlength{\belowdisplayskip}{1.2pt}\begin{array}{l}
{\cal D}(I^j_i|\varpi(\beta(\gamma),\eta(\varepsilon)))(u^i(a'',\varrho^i_{I^j_i}(\varpi(\beta^{-I^j_i}(\gamma),\eta(\varepsilon)),\tilde\beta(\gamma,\varepsilon)),\mu(\gamma,\varepsilon)|I^j_i)\\
\hspace{5cm}-u^i(a',\varrho^i_{I^j_i}(\varpi(\beta^{-I^j_i}(\gamma),\eta(\varepsilon)),\tilde\beta(\gamma,\varepsilon)),\mu(\gamma,\varepsilon)|I^j_i))>\gamma,\end{array}\] 
(ii). $\tilde\beta^i_{I^j_i}(a')=0$ for any $i\in N$, $j\in M_i$ and $a',a''\in A(I^j_i)$ with\[\setlength{\abovedisplayskip}{1.2pt}
\setlength{\belowdisplayskip}{1.2pt}\begin{array}{l}
u^i(a'',\varrho^i_{I^j_i}(\varpi(\beta^{-I^j_i}(\gamma),\eta(\varepsilon)),\tilde\beta(\gamma, \varepsilon)),\mu(\gamma,\varepsilon)|I^j_i)\\
\hspace{5cm}-u^i(a',\varrho^i_{I^j_i}(\varpi(\beta^{-I^j_i}(\gamma),\eta(\varepsilon)),\tilde\beta(\gamma, \varepsilon)),\mu(\gamma,\varepsilon)|I^j_i)>\gamma,\end{array}\] 
(iii). $\mu(\gamma,\varepsilon)=(\mu^i_{I^j_i}(\gamma,\varepsilon;h):i\in N,j\in M_i,h\in I^j_i)$ with $\mu^i_{I^j_i}(\gamma,\varepsilon;h)=\frac{\omega(h|\varpi(\beta(\gamma),\eta(\varepsilon)))}{\omega(I^j_i|\varpi(\beta(\gamma),\eta(\varepsilon)))}$ for $h\in I^j_i$. 

\noindent For any given $\gamma>0$,  $(\beta(\gamma),\mu(\gamma),\tilde\beta(\gamma))$ is a perfect $\gamma$-WSRE if $(\mu(\gamma),\tilde\beta(\gamma))$ is a limit point of a sequence $\{(\mu(\gamma,\varepsilon_k), \tilde\beta(\gamma,\varepsilon_k)), k=1,2,\ldots\}$ with $\varepsilon_k>0$ and $\lim\limits_{k\to\infty}\varepsilon_k=0$, where $(\varpi(\beta(\gamma),\eta(\varepsilon_k)), \mu(\gamma,\varepsilon_k)$, $\tilde\beta(\gamma, \varepsilon))$ is an $\varepsilon_k$-perfect $\gamma$-WSRE for every $k$. 
$(\beta^*,\mu^*)$ is a WSRE if $(\beta^*,\mu^*,\tilde\beta^*)$ is a limit point of a sequence $\{(\beta(\gamma_q),\mu(\gamma_q), \tilde\beta(\gamma_q)),q=1,2,\ldots\}$ with $\gamma_q>0$ and $\lim\limits_{q\to\infty}\gamma_q=0$, where $(\beta(\gamma_q),\mu(\gamma_q), \tilde\beta(\gamma_q))$ is a perfect $\gamma_q$-WSRE for every $q$.}
\end{definition}

We attain in a similar way to the proof of Theorem~\ref{nscwsrethm1} the following conclusion. One can compute an $\varepsilon$-perfect $\gamma$-WSRE by directly solving the system~(\ref{nscwsreesB}).
\begin{theorem}\label{nscwsrethmB} {\em $(\varpi(\beta(\gamma),\eta(\varepsilon)),\mu(\gamma,\varepsilon), \tilde\beta(\gamma,\varepsilon))$ is an $\varepsilon$-perfect $\gamma$-WSRE if and only if there exists $(\lambda(\gamma,\varepsilon), \tilde\lambda(\gamma,\varepsilon), \zeta(\gamma,\varepsilon),\tilde\zeta(\gamma,\varepsilon))$  together with $(\beta(\gamma),\eta(\varepsilon),\mu(\gamma,\varepsilon), \tilde\beta(\gamma,\varepsilon))$ satisfying the polynomial system,
{\small
\begin{equation}\label{nscwsreesB}\setlength{\abovedisplayskip}{1.2pt}
\setlength{\belowdisplayskip}{1.2pt}\
\begin{array}{l}
-\gamma\le {\cal D}(I^j_i|\varpi(\beta,\eta))u^i(a,\varrho^i_{I^j_i}(\varpi(\beta^{-I^j_i},\eta),\tilde\beta),\mu|I^j_i)+\lambda^i_{I^j_i}(a)-\zeta^i_{I^j_i}\le\gamma,\;i\in N, j\in M_i, a\in A(I^j_i),\\

-\gamma\le u^i(a,\varrho^i_{I^j_i}(\varpi(\beta^{-I^j_i},\eta),\tilde\beta),\mu|I^j_i)+\tilde\lambda^i_{I^j_i}(a)-\tilde\zeta^i_{I^j_i}\le\gamma,\;i\in N, j\in M_i, a\in A(I^j_i),\\

\sum\limits_{a\in A(I^j_i)}\beta^i_{I^j_i}(a)=1,\;
\sum\limits_{a\in A(I^j_i)}\tilde\beta^i_{I^j_i}(a)=1,\;i\in N, j\in M_i,\\

\beta^i_{I^j_i}(a)\lambda^i_{I^j_i}(a)=0, \;0\le\beta^i_{I^j_i}(a),\;0\le\lambda^i_{I^j_i}(a),\; i\in N, j\in M_i, a\in A(I^j_i),\\

\tilde\beta^i_{I^j_i}(a)\tilde\lambda^i_{I^j_i}(a)=0, \;0\le\tilde\beta^i_{I^j_i}(a),\;0\le\tilde\lambda^i_{I^j_i}(a),\; i\in N, j\in M_i, a\in A(I^j_i).
\end{array}\end{equation}
}
}
\end{theorem}

\section{Illustrative Examples\label{examples}}

This section presents two examples to illustrate how one can employ Definition~\ref{edAwsre} to analytically find all WSREs for small extensive-form games.
\begin{figure}[H]
    \centering
    \begin{minipage}{0.49\textwidth}
        \centering
        \includegraphics[width=0.80\textwidth, height=0.15\textheight]{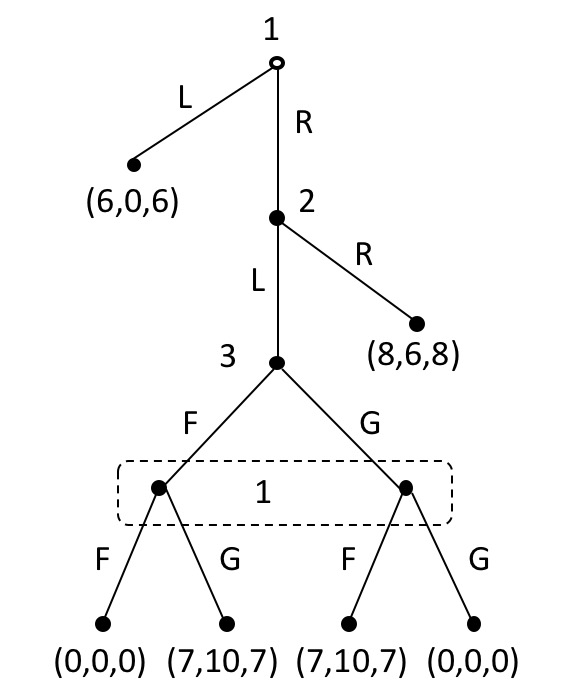}
        % first figure itself
                \caption{\label{TFigure1}\scriptsize An Extensive-Form Game}
\end{minipage}\hfill
    \begin{minipage}{0.49\textwidth}
        \centering
        \includegraphics[width=0.80\textwidth, height=0.15\textheight]{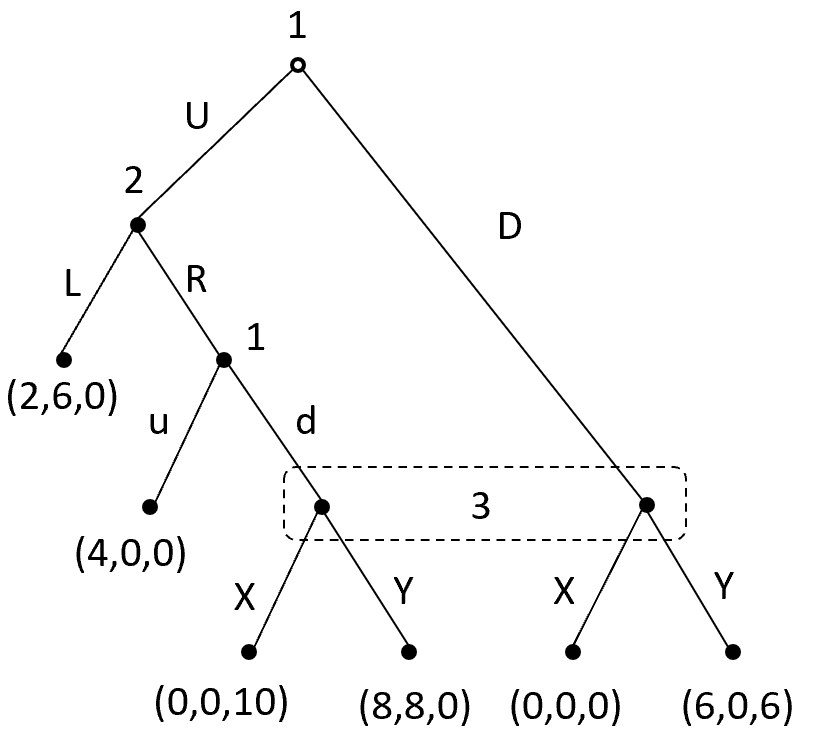}
        % first figure itself
\caption{\label{TFigure2}\scriptsize  An Extensive-Form Game}\end{minipage}
 \end{figure}

\begin{example} {\em
Consider the game in Fig.~\ref{TFigure1}. The information sets consist of $I^1_1=\{\emptyset\}$, $I^2_1=\{\langle R, L, F\rangle, \langle R, L, G\rangle\}$, $I^1_2=\{\langle R\rangle\}$, and $I^1_3=\{\langle R, L\rangle\}$. We denote by $(\beta(\gamma, \varepsilon), \mu(\gamma, \varepsilon),\tilde\beta(\gamma,\varepsilon))$
 an $\varepsilon$-perfect $\gamma$-WSRE. For convenience, we will omit $(\gamma,\varepsilon)$ from $(\beta(\gamma, \varepsilon), \mu(\gamma, \varepsilon),\tilde\beta(\gamma,\varepsilon))$  in the rest of this example.
 Each WSRE is presented in the form of $((\beta^{*1}_{I^1_1}(L), \beta^{*1}_{I^1_1}(R)), (\beta^{*1}_{I^2_1}(F),\beta^{*1}_{I^2_1}(G)), (\beta^{*2}_{I^1_2}(L),\beta^{*2}_{I^1_2}(R))$, $(\beta^{*3}_{I^1_3}(F),\beta^{*3}_{I^1_3}(G)))$.
 It follows from Fig.~\ref{TFigure1} that 
{\small\[\setlength{\abovedisplayskip}{1.2pt}
\setlength{\belowdisplayskip}{1.2pt}
\begin{array}{l}
{\cal D}(I^1_1|\beta)=1,\; {\cal D}(I^2_1|\beta)=\beta^1_{I^1_1}(R),\;{\cal D}(I^1_2|\beta)=1,\;{\cal D}(I^1_3|\beta)=1,\\

\mu^1_{I^1_1}(\emptyset)=1,\; \mu^1_{I^2_1}(\langle R, L,F\rangle)=\beta^3_{I^1_3}(F),\; \mu^1_{I^2_1}(\langle R, L,G\rangle)=\beta^3_{I^1_3}(G),\;

\mu^2_{I^1_2}(\langle R\rangle)=1,\; \mu^3_{I^1_3}(\langle R, L\rangle)=1,
\end{array}\]}
{\small\[\setlength{\abovedisplayskip}{1.2pt}
\setlength{\belowdisplayskip}{1.2pt}
\begin{array}{l}
u^1(L,\varrho^1_{I^1_1}(\beta^{-I^1_1},\tilde\beta),\mu| I^1_1)=6,\\

u^1(R,\varrho^1_{I^1_1}(\beta^{-I^1_1},\tilde\beta),\mu| I^1_1)=7\beta^2_{I^1_2}(L)(\beta^3_{I^1_3}(F)\tilde\beta^1_{I^2_1}(G)+ \beta^3_{I^1_3}(G)\tilde\beta^1_{I^2_1}(F))+8\beta^2_{I^1_2}(\gamma,\varepsilon; R),\\

u^1(F,\varrho^1_{I^2_1}(\beta^{-I^2_1},\tilde\beta),\mu| I^2_1)=7\beta^3_{I^1_3}(G),\;

u^1(G,\varrho^1_{I^2_1}(\beta^{-I^2_1},\tilde\beta),\mu| I^2_1)=7\beta^3_{I^1_3}(F),\\

u^2(L,\varrho^2_{I^1_2}(\beta^{-I^1_2},\tilde\beta),\mu| I^1_2)=10(\beta^3_{I^1_3}(F)\beta^1_{I^2_1}(G)+ \beta^3_{I^1_3}(G)\beta^1_{I^2_1}(F)),\;

u^2(R,\varrho^2_{I^1_2}(\beta^{-I^1_2},\tilde\beta),\mu| I^1_2)=6,\\

u^3(F,\varrho^3_{I^1_3}(\beta^{-I^1_3},\tilde\beta),\mu| I^1_3)= 7\beta^1_{I^2_1}(G),\;

u^3(G,\varrho^3_{I^1_3}(\beta^{-I^1_3},\tilde\beta),\mu| I^1_3)= 7\beta^1_{I^2_1}(F).
\end{array}\]}

\noindent {\bf Case (1)}. Suppose that  ${\cal D}(I^2_1|\beta)(u^1(F,\varrho^1_{I^2_1}(\beta^{-I^2_1},\tilde\beta),\mu| I^2_1)-u^1(G,\varrho^1_{I^2_1}(\beta^{-I^2_1},\tilde\beta),\mu| I^2_1))>\gamma$. 
Then, $\beta^1_{I^2_1}(G)\le\varepsilon$ and $\tilde\beta^1_{I^2_1}(G)=0$. Thus, ${\cal D}(I^1_3|\beta)(u^3(G,\varrho^3_{I^1_3}(\beta^{-I^1_3},\tilde\beta),\mu| I^1_3)-u^3(F,\varrho^3_{I^1_3}(\beta^{-I^1_3},\tilde\beta),\mu| I^1_3))>\gamma$ and consequently, $\beta^3_{I^1_3}(F)\le\varepsilon$. Therefore, 
${\cal D}(I^1_2|\beta)(u^2(L,\varrho^2_{I^1_2}(\beta^{-I^1_2},\tilde\beta),\mu| I^1_2)-u^2(R,\varrho^2_{I^1_2}(\beta^{-I^1_2},\tilde\beta),\mu| I^1_2))>\gamma$ and 
accordingly, $\beta^2_{I^1_2}(R)\le\varepsilon$. Hence, ${\cal D}(I^1_1|\beta)(u^1(R,\varrho^1_{I^1_1}(\beta^{-I^1_1},\tilde\beta),\mu| I^1_1)-u^1(L,\varrho^1_{I^1_1}(\beta^{-I^1_1},\tilde\beta),\mu| I^1_1))>\gamma$ and subsequently, $\beta^1_{I^1_1}(L)\le\varepsilon$. The game has a WSRE given by $(R, F, L, G)$ with $\mu^{*1}_{I^2_1}(\langle R, L, G\rangle)=1$.

\noindent {\bf Case (2)}. Suppose that  ${\cal D}(I^2_1|\beta)(u^1(G,\varrho^1_{I^2_1}(\beta^{-I^2_1},\tilde\beta),\mu| I^2_1)-u^1(F,\varrho^1_{I^2_1}(\beta^{-I^2_1},\tilde\beta),\mu| I^2_1))>\gamma$. 
Then, $\beta^1_{I^2_1}(F)\le\varepsilon$ and $\tilde\beta^1_{I^2_1}(F)=0$. Thus, ${\cal D}(I^1_3|\beta)(u^3(F,\varrho^3_{I^1_3}(\beta^{-I^1_3},\tilde\beta),\mu| I^1_3)-u^3(G,\varrho^3_{I^1_3}(\beta^{-I^1_3},\tilde\beta),\mu| I^1_3))>\gamma$ and consequently, $\beta^3_{I^1_3}(G)\le\varepsilon$. Therefore, 
${\cal D}(I^1_2|\beta)(u^2(L,\varrho^2_{I^1_2}(\beta^{-I^1_2},\tilde\beta),\mu| I^1_2)-u^2(R,\varrho^2_{I^1_2}(\beta^{-I^1_2},\tilde\beta),\mu| I^1_2))>\gamma$ and accordingly, $\beta^2_{I^1_2}(R)\le\varepsilon$. Hence, ${\cal D}(I^1_1|\beta)(u^1(R,\varrho^1_{I^1_1}(\beta^{-I^1_1},\tilde\beta),\mu| I^1_1)-u^1(L,\varrho^1_{I^1_1}(\beta^{-I^1_1},\tilde\beta),\mu| I^1_1))>\gamma$ and subsequently, $\beta^1_{I^1_1}(L)\le\varepsilon$.
The game has a WSRE given by $(R, G, L, F)$ with $\mu^{*1}_{I^2_1}(\langle R, L,F\rangle)=1$.

\noindent {\bf Case (3)}. Suppose that  ${\cal D}(I^2_1|\beta)|u^1(F,\varrho^1_{I^2_1}(\beta^{-I^2_1},\tilde\beta),\mu| I^2_1)-u^1(G,\varrho^1_{I^2_1}(\beta^{-I^2_1},\tilde\beta),\mu| I^2_1)|\le\gamma$. 
Then,  \begin{equation}\label{edAwsreeqA}\setlength{\abovedisplayskip}{1.2pt}
\setlength{\belowdisplayskip}{1.2pt}
\beta^1_{I^1_1}(R)|\beta^3_{I^1_3}(F)-\frac{1}{2}|\le\frac{\gamma}{14}.\end{equation}

\noindent {\bf (a)}. Assume that ${\cal D}(I^1_3|\beta)(u^3(F,\varrho^3_{I^1_3}(\beta^{-I^1_3},\tilde\beta),\mu| I^1_3)-u^3(G,\varrho^3_{I^1_3}(\beta^{-I^1_3},\tilde\beta),\mu| I^1_3))>\gamma$.
Then, $\beta^3_{I^1_3}(G)\le\varepsilon$. Thus, $u^1(G,\varrho^1_{I^2_1}(\beta^{-I^2_1},\tilde\beta),\mu| I^2_1)-u^1(F,\varrho^1_{I^2_1}(\beta^{-I^2_1},\tilde\beta),\mu| I^2_1)>\gamma$ and consequently, $\tilde\beta^1_{I^2_1}(F)=0$. Therefore,  ${\cal D}(I^1_1|\beta)(u^1(R,\varrho^1_{I^1_1}(\beta^{-I^1_1},\tilde\beta),\mu| I^1_1)-u^1(L,\varrho^1_{I^1_1}(\beta^{-I^1_1},\tilde\beta),\mu| I^1_1))>\gamma$ and accordingly, $\beta^1_{I^1_1}(L)\le\varepsilon$. Hence it follows from Eq.~(\ref{edAwsreeqA}) that 
$(1-\varepsilon)^2\le \frac{1}{2}+\frac{1}{14}\gamma$. A contradiction occurs. The assumption is excluded.

\noindent {\bf (b)}. Assume that ${\cal D}(I^1_3|\beta)(u^3(G,\varrho^3_{I^1_3}(\beta^{-I^1_3},\tilde\beta),\mu| I^1_3)-u^3(F,\varrho^3_{I^1_3}(\beta^{-I^1_3},\tilde\beta),\mu| I^1_3))>\gamma$.
Then,  $\beta^3_{I^1_3}(F)\le\varepsilon$ and $\frac{1}{2}-\frac{\gamma}{14}>\beta^1_{I^2_1}(\gamma,\varepsilon; G)$. Thus, $u^1(F,\varrho^1_{I^2_1}(\beta^{-I^2_1},\tilde\beta),\mu| I^2_1)-u^1(G,\varrho^1_{I^2_1}(\beta^{-I^2_1},\tilde\beta),\mu| I^2_1)>\gamma$ and consequently, $\tilde\beta^1_{I^2_1}(G)=0$. Moreover, Eq.~(\ref{edAwsreeqA}) and $\beta^3_{I^1_3}(F)\le\varepsilon$ together imply that either  ${\cal D}(I^1_1|\beta)(u^1(L,\varrho^1_{I^1_1}(\beta^{-I^1_1},\tilde\beta),\mu| I^1_1)-u^1(R,\varrho^1_{I^1_1}(\beta^{-I^1_1},\tilde\beta),\mu| I^1_1))>\gamma$ or ${\cal D}(I^1_1|\beta)|(u^1(L,\varrho^1_{I^1_1}(\beta^{-I^1_1},\tilde\beta),\mu| I^1_1)-u^1(R,\varrho^1_{I^1_1}(\beta^{-I^1_1},\tilde\beta),\mu| I^1_1))|\le\gamma$.

\noindent {\bf (i)}. Consider the scenario that ${\cal D}(I^1_1|\beta)(u^1(L,\varrho^1_{I^1_1}(\beta^{-I^1_1},\tilde\beta),\mu| I^1_1)-u^1(R,\varrho^1_{I^1_1}(\beta^{-I^1_1},\tilde\beta),\mu| I^1_1))>\gamma$. Then, $\frac{3}{4}-\frac{1}{8}\gamma>\beta^2_{I^1_2}(R)$ and $\beta^1_{I^1_1}(R)\le\varepsilon$. Thus either 
${\cal D}(I^1_2|\beta)(u^2(L,\varrho^2_{I^1_2}(\beta^{-I^1_2},\tilde\beta),\mu| I^1_2)-u^2(R,\varrho^2_{I^1_2}(\beta^{-I^1_2},\tilde\beta),\mu| I^1_2))>\gamma$ or ${\cal D}(I^1_2|\beta)|u^2(L,\varrho^2_{I^1_2}(\beta^{-I^1_2},\tilde\beta),\mu| I^1_2)-u^2(R,\varrho^2_{I^1_2}(\beta^{-I^1_2},\tilde\beta),\mu| I^1_2)|\le\gamma$. 

\noindent $\bullet$ Presume that ${\cal D}(I^1_2|\beta)(u^2(R,\varrho^2_{I^1_2}(\beta^{-I^1_2},\tilde\beta),\mu| I^1_2)-u^2(L,\varrho^2_{I^1_2}(\beta^{-I^1_2},\tilde\beta),\mu| I^1_2))>\gamma$. Then, $\beta^2_{I^1_2}(R)\le\varepsilon$ and  $\beta^3_{I^1_3}(F)+ (\beta^3_{I^1_3}(G)-\beta^3_{I^1_3}(F))\beta^1_{I^2_1}(F)>\frac{3}{5}-\frac{\gamma}{10}$. The game has a class of WSREs given by $(L, (\beta^{*1}_{I^2_1}(F), 1-\beta^{*1}_{I^2_1}(F)), L, G)$ with $\beta^{*1}_{I^2_1}(F)>\frac{3}{5}$ and $\mu^{*1}_{I^2_1}(\langle R, L,F\rangle)=0$. 

\noindent $\bullet$ Presume that ${\cal D}(I^1_2|\beta)|u^2(R,\varrho^2_{I^1_2}(\beta^{-I^1_2},\tilde\beta),\mu| I^1_2)-u^2(L,\varrho^2_{I^1_2}(\beta^{-I^1_2},\tilde\beta),\mu| I^1_2)|\le\gamma$. Then, $\frac{3}{5}-\frac{\gamma}{10}\le\beta^3_{I^1_3}(F)+ (\beta^3_{I^1_3}(G)-\beta^3_{I^1_3}(F))\beta^1_{I^2_1}(F)\le\frac{3}{5}+\frac{\gamma}{10}$. The game has a class of WSREs given by $(L, (\frac{3}{5}, \frac{2}{5}), (\beta^{*2}_{I^1_2}(L), 1-\beta^{*2}_{I^1_2}(L)), G)$ with $\beta^{*2}_{I^1_2}(L)>\frac{1}{4}$ and $\mu^{*1}_{I^2_1}(\langle R, L,F\rangle)=0$. 

\noindent {\bf (ii)}. Consider the scenario that ${\cal D}(I^1_1|\beta)|u^1(R,\varrho^1_{I^1_1}(\beta^{-I^1_1},\tilde\beta),\mu| I^1_1)-u^1(L,\varrho^1_{I^1_1}(\beta^{-I^1_1},\tilde\beta),\mu| I^1_1)|\le\gamma$. Then, $|\beta^2_{I^1_2}(R)-\frac{3}{4}|\le\frac{1}{8}\gamma$. Thus, ${\cal D}(I^1_2|\beta)|u^2(L,\varrho^2_{I^1_2}(\beta^{-I^1_2},\tilde\beta),\mu| I^1_2)-u^2(R,\varrho^2_{I^1_2}(\beta^{-I^1_2},\tilde\beta),\mu| I^1_2)|\le\gamma$ and accordingly, $\frac{3}{5}-\frac{\gamma}{10}\le\beta^3_{I^1_3}(F)+ (\beta^3_{I^1_3}(G)-\beta^3_{I^1_3}(F))\beta^1_{I^2_1}(F)\le\frac{3}{5}+\frac{\gamma}{10}$. The game has a WSRE given by $(L, (\frac{3}{5}, \frac{2}{5}), (\frac{1}{4}, \frac{3}{4}), G)$ with  $\mu^{*1}_{I^2_1}(\langle R, L,F\rangle)=0$. 

\noindent {\bf (c)}. Assume that ${\cal D}(I^1_3|\beta)|u^3(G,\varrho^3_{I^1_3}(\beta^{-I^1_3},\tilde\beta),\mu| I^1_3)-u^3(F,\varrho^3_{I^1_3}(\beta^{-I^1_3},\tilde\beta),\mu| I^1_3)|\le\gamma$.
Then, $|\beta^1_{I^2_1}(G)-\frac{1}{2}|\le \frac{1}{14}\gamma$. Thus,  
${\cal D}(I^1_2|\beta)(u^2(R,\varrho^2_{I^1_2}(\beta^{-I^1_2},\tilde\beta),\mu| I^1_2)-u^2(L,\varrho^2_{I^1_2}(\beta^{-I^1_2},\tilde\beta),\mu| I^1_2))>\gamma$ and 
 accordingly, $\beta^2_{I^1_2}(L)\le\varepsilon$.
Thus,  ${\cal D}(I^1_1|\beta)(u^1(R,\varrho^1_{I^1_1}(\beta^{-I^1_1},\tilde\beta),\mu| I^1_1)-u^1(L,\varrho^1_{I^1_1}(\beta^{-I^1_1},\tilde\beta),\mu| I^1_1))>\gamma$ and consequently, $\beta^1_{I^1_1}(L)\le\varepsilon$. 
The game has a WSRE given by $(R, (\frac{1}{2}, \frac{1}{2}), R, 
(\frac{1}{2}, \frac{1}{2}))$ with  $\mu^{*1}_{I^2_1}(\langle R, L,F\rangle)=\frac{1}{2}$.

Cases (1)-(3) together show that the game has five classes of WSREs given by\\
(1). $(R, G, L, F)$ with $\mu^{*1}_{I^2_1}(\langle R, L,F\rangle)=1$.\\
(2). $(R, F, L, G)$ with $\mu^{*1}_{I^2_1}(\langle R, L, G\rangle)=1$.\\
 (3). $(L, (\beta^{*1}_{I^2_1}(F), 1-\beta^{*1}_{I^2_1}(F)), L, G)$ with $\beta^{*1}_{I^2_1}(F)>\frac{3}{5}$ and $\mu^{*1}_{I^2_1}(\langle R, L,F\rangle)=0$. \\
 (4). $(L, (\frac{3}{5}, \frac{2}{5}), (\beta^{*2}_{I^1_2}(L), 1-\beta^{*2}_{I^1_2}(L)), G)$ with $\beta^{*2}_{I^1_2}(L)\ge\frac{1}{4}$ and $\mu^{*1}_{I^2_1}(\langle R, L,F\rangle)=0$. \\
(5). $(R, (\frac{1}{2}, \frac{1}{2}), R, 
(\frac{1}{2}, \frac{1}{2}))$ with  $\mu^{*1}_{I^2_1}(\langle R, L,F\rangle)=\frac{1}{2}$.
}
\end{example}

\begin{example} {\em Consider the game in Fig.~\ref{TFigure2}. The information sets consist of $I^1_1=\{\emptyset\}$, $I^2_1=\{\langle U, R\rangle\}$, $I^1_2=\{\langle U\rangle\}$, and $I^1_3=\{\langle U, R, d\rangle, \langle D\rangle\}$. We denote by $(\beta(\gamma, \varepsilon), \mu(\gamma, \varepsilon),\tilde\beta(\gamma,\varepsilon))$
 an $\varepsilon$-perfect $\gamma$-WSRE. For simplicity of our exposition, we drop $(\gamma,\varepsilon)$ from $(\beta(\gamma, \varepsilon), \mu(\gamma, \varepsilon),\tilde\beta(\gamma,\varepsilon))$ so that $(\beta(\gamma, \varepsilon), \mu(\gamma, \varepsilon),\tilde\beta(\gamma,\varepsilon))$ becomes $(\beta, \mu,\tilde\beta)$ in the rest of this example.
 Each WSRE is presented in the form of 
 $((\beta^{*1}_{I^1_1}(U),\beta^{*1}_{I^1_1}(D)), (\beta^{*1}_{I^2_1}(u),\beta^{*1}_{I^2_1}(d)), (\beta^{*2}_{I^1_2}(L),\beta^{*2}_{I^1_2}(R)), (\beta^{*3}_{I^1_3}(X),\beta^{*3}_{I^1_3}(Y)))$. 
It follows from Fig.~\ref{TFigure2} that
{\small\[\setlength{\abovedisplayskip}{1.2pt}
\setlength{\belowdisplayskip}{1.2pt}
\begin{array}{l}
{\cal D}(I^1_1|\beta)=1,\; {\cal D}(I^2_1|\beta)=\beta^1_{I^1_1}(U),\; {\cal D}(I^1_2|\beta)=1,\;  {\cal D}(I^1_3|\beta)=1,\;
\mu^1_{I^1_1}(\emptyset)=1,\; \mu^1_{I^2_1}(\langle U, R\rangle)=1,\;\mu^2_{I^1_2}(\langle U\rangle)=1,\\

\mu^3_{I^1_3}(\langle U, R,d\rangle)=\frac{\beta^1_{I^1_1}(U) \beta^2_{I^1_2}(R) \beta^1_{I^2_1}(d)}{\beta^1_{I^1_1}(U) \beta^2_{I^1_2}(R) \beta^1_{I^2_1}(d)+\beta^1_{I^1_1}(D)},\;

\mu^3_{I^1_3}(\langle D\rangle)=\frac{\beta^1_{I^1_1}(D)}{\beta^1_{I^1_1}(U) \beta^2_{I^1_2}(R) \beta^1_{I^2_1}(d)+\beta^1_{I^1_1}(D)},\end{array}\]}

{\small\[\setlength{\abovedisplayskip}{1.2pt}
\setlength{\belowdisplayskip}{1.2pt}
\begin{array}{l}
u^1(U,\varrho^1_{I^1_1}(\beta^{-I^1_1},\tilde\beta),\mu|I^1_1)=2(\beta^2_{I^1_2}(L)+2\beta^2_{I^1_2}(R)\tilde\beta^1_{I^2_1}(u)+4\beta^2_{I^1_2}(R)\tilde\beta^1_{I^2_1}(d)\beta^3_{I^1_3}(Y)),\\

u^1(D,\varrho^1_{I^1_1}(\beta^{-I^1_1},\tilde\beta),\mu|I^1_1)=6\beta^3_{I^1_3}(Y),\\

u^1(u,\varrho^1_{I^2_1}(\beta^{-I^2_1},\tilde\beta),\mu|I^2_1)
=4,\; u^1(d,\varrho^1_{I^2_1}(\beta^{-I^2_1},\tilde\beta),\mu|I^2_1)
=8\beta^3_{I^1_3}(Y),\\

u^2(L,\varrho^2_{I^1_2}(\beta^{-I^1_2},\tilde\beta),\mu|I^1_2)
=6,\; u^2(R,\varrho^2_{I^1_2}(\beta^{-I^1_2},\tilde\beta),\mu|I^1_2)
=8\beta^1_{I^2_1}(d)\beta^3_{I^1_3}(Y),\\

u^3(X,\varrho^3_{I^1_3}(\beta^{-I^1_3}, \tilde\beta),\mu|I^1_3)
=10\mu^3_{I^1_3}(\langle U, R,d\rangle),\;u^3(Y,\varrho^3_{I^1_3}(\beta^{-I^1_3}, \tilde\beta),\mu|I^1_3)
=6\mu^3_{I^1_3}(\langle D\rangle).
\end{array}\]}

\noindent {\bf Case (1)}. Suppose that ${\cal D}(I^1_3|\beta)(u^3(X,\varrho^3_{I^1_3}(\beta^{-I^1_3}, \tilde\beta),\mu|I^1_3)-u^3(Y,\varrho^3_{I^1_3}(\beta^{-I^1_3}, \tilde\beta),\mu|I^1_3))>\gamma$. Then, $\beta^3_{I^1_3}(Y)\le\varepsilon$ and $\frac{5}{8}-\frac{1}{16}\gamma>\mu^3_{I^1_3}(\langle D\rangle)$, which yields $\beta^1_{I^1_1}(D)<\frac{10-\gamma}{6+\gamma}\beta^1_{I^1_1}(U)\beta^2_{I^1_2}(R)\beta^1_{I^2_1}(d)$. Thus, $u^1(u,\varrho^1_{I^2_1}(\beta^{-I^2_1},\tilde\beta),\mu|I^2_1)-u^1(d,\varrho^1_{I^2_1}(\beta^{-I^2_1},\tilde\beta),\mu|I^2_1)>\gamma$ and ${\cal D}(I^1_2|\beta)(u^2(L,\varrho^2_{I^1_2}(\beta^{-I^1_2},\tilde\beta),\mu|I^1_2)-u^2(R$, $\varrho^2_{I^1_2}(\beta^{-I^1_2},\tilde\beta),\mu|I^1_2))>\gamma$. Consequently, $\tilde\beta^1_{I^2_1}(d)=0$ and $\beta^2_{I^1_2}(R)\le\varepsilon$. Therefore, ${\cal D}(I^1_1|\beta)(u^1(U$, $\varrho^1_{I^1_1}(\beta^{-I^1_1},\tilde\beta),\mu|I^1_1)-u^1(D,\varrho^1_{I^1_1}(\beta^{-I^1_1},\tilde\beta),\mu|I^1_1))>\gamma$ and accordingly,  $\beta^1_{I^1_1}(D)\le\varepsilon$, which
 implies ${\cal D}(I^2_1|\beta)(u^1(u,\varrho^1_{I^2_1}(\beta^{-I^2_1},\tilde\beta),\mu|I^2_1)-u^1(u,\varrho^1_{I^2_1}(\beta^{-I^2_1},\tilde\beta),\mu|I^2_1))>\gamma$. Hence,  $\beta^1_{I^2_1}(d)\le\varepsilon$. The game has a WSRE given by $(U, u, L, X)$ with $\mu^{*3}_{I^1_3}(\langle D\rangle)<\frac{5}{8}$.

\noindent {\bf Case (2)}. Suppose that ${\cal D}(I^1_3|\beta)(u^3(Y,\varrho^3_{I^1_3}(\beta^{-I^1_3}, \tilde\beta),\mu|I^1_3)-u^3(X,\varrho^3_{I^1_3}(\beta^{-I^1_3}, \tilde\beta),\mu|I^1_3))>\gamma$. Then, $\beta^3_{I^1_3}(X)\le\varepsilon$ and $\mu^3_{I^1_3}(\langle D\rangle)>\frac{5}{8}+\frac{1}{16}\gamma$, which yields \begin{equation}
\label{edAwsrewsreeqB}\setlength{\abovedisplayskip}{1.2pt}
\setlength{\belowdisplayskip}{1.2pt}\beta^1_{I^1_1}(D)>\frac{10+\gamma}{6-\gamma}\beta^1_{I^1_1}(U)\beta^2_{I^1_2}(R)\beta^1_{I^2_1}(d).\end{equation} Thus, $u^1(d,\varrho^1_{I^2_1}(\beta^{-I^2_1},\tilde\beta),\mu|I^2_1)-u^1(u,\varrho^1_{I^2_1}(\beta^{-I^2_1},\tilde\beta),\mu|I^2_1)>\gamma$ and consequently, $\tilde\beta^1_{I^2_1}(u)=0$.

\noindent {\bf (a)}. Assume that ${\cal D}(I^1_2|\beta)(u^2(L,\varrho^2_{I^1_2}(\beta^{-I^1_2},\tilde\beta),\mu|I^1_2)-u^2(R,\varrho^2_{I^1_2}(\beta^{-I^1_2},\tilde\beta),\mu|I^1_2))>\gamma$. Then, $\beta^2_{I^1_2}(R)\le\varepsilon$  and $\frac{3}{4}-\frac{1}{8}\gamma>\beta^1_{I^2_1}(d)\beta^3_{I^1_3}(Y)$. Thus,  ${\cal D}(I^1_1|\beta)(u^1(D,\varrho^1_{I^1_1}(\beta^{-I^1_1},\tilde\beta),\mu|I^1_1)-u^1(U,\varrho^1_{I^1_1}(\beta^{-I^1_1},\tilde\beta),\mu|I^1_1))>\gamma$ and accordingly,  $\beta^1_{I^1_1}(U)\le\varepsilon$, which implies ${\cal D}(I^2_1|\beta)|u^1(u,\beta^{-I^2_1}$, $\mu|I^2_1)-u^1(d,\varrho^1_{I^2_1}(\beta^{-I^2_1},\tilde\beta),\mu|I^2_1)|\le\gamma$.
 The game has a class of WSREs given by $(D, (1-\beta^{*1}_{I^2_1}(d),\beta^{*1}_{I^2_1}(d)), L, Y)$ with $\beta^{*1}_{I^2_1}(d)<\frac{3}{4}$ and $\mu^{*3}_{I^1_3}(\langle D\rangle)=1$.

\noindent {\bf (b)}. Assume that ${\cal D}(I^1_2|\beta)(u^2(R,\varrho^2_{I^1_2}(\beta^{-I^1_2},\tilde\beta),\mu|I^1_2)-u^2(L,\varrho^2_{I^1_2}(\beta^{-I^1_2},\tilde\beta),\mu|I^1_2))>\gamma$. Then, $\beta^2_{I^1_2}(L)\le\varepsilon$  and $\beta^1_{I^2_1}(d)\beta^3_{I^1_3}(Y)>\frac{3}{4}+\frac{1}{8}\gamma$. Thus,  ${\cal D}(I^1_1|\beta)(u^1(U,\varrho^1_{I^1_1}(\beta^{-I^1_1},\tilde\beta),\mu|I^1_1)-u^1(D,\varrho^1_{I^1_1}(\beta^{-I^1_1},\tilde\beta),\mu|I^1_1))>\gamma$ and accordingly,  $\beta^1_{I^1_1}(D)\le\varepsilon$. Therefore, ${\cal D}(I^2_1|\beta)(u^1(d,\beta^{-I^2_1}$, $\mu|I^2_1)-u^1(u,\varrho^1_{I^2_1}(\beta^{-I^2_1},\tilde\beta),\mu|I^2_1))>\gamma$ and consequently, $\beta^1_{I^2_1}(u)\le\varepsilon$. These results together  with Eq.~(\ref{edAwsrewsreeqB}) bring us that $\varepsilon\ge\beta^1_{I^1_1}(D)>\frac{10+\gamma}{6-\gamma}\beta^1_{I^1_1}(U)\beta^2_{I^1_2}(R)\beta^1_{I^2_1}(d)\ge \frac{10+\gamma}{6-\gamma}(1-\varepsilon)(1-\varepsilon)(1-\varepsilon)$. A contradiction occurs and
the assumption cannot be sustained.

\noindent {\bf (c)}. Assume that ${\cal D}(I^1_2|\beta)|u^2(R,\varrho^2_{I^1_2}(\beta^{-I^1_2},\tilde\beta),\mu|I^1_2)-u^2(L,\varrho^2_{I^1_2}(\beta^{-I^1_2},\tilde\beta),\mu|I^1_2)|\le\gamma$. Then, $|\beta^1_{I^2_1}(d)\beta^3_{I^1_3}(Y)-\frac{3}{4}|\le\frac{1}{8}\gamma$, which yields 
${\cal D}(I^2_1|\beta)|u^1(d,\varrho^1_{I^2_1}(\beta^{-I^2_1},\tilde\beta),\mu|I^2_1)-u^1(u,\varrho^1_{I^2_1}(\beta^{-I^2_1},\tilde\beta),\mu|I^2_1)|\le\gamma$. Thus,  $\beta^1_{I^1_1}(U)|\frac{1}{2}- \beta^3_{I^1_3}(Y)|\le\frac{1}{8}\gamma$. Therefore either
${\cal D}(I^1_1|\beta)(u^1(D,\varrho^1_{I^1_1}(\beta^{-I^1_1},\tilde\beta),\mu|I^1_1)-u^1(U,\varrho^1_{I^1_1}(\beta^{-I^1_1},\tilde\beta),\mu|I^1_1))>\gamma$ or ${\cal D}(I^1_1|\beta)|u^1(D,\varrho^1_{I^1_1}(\beta^{-I^1_1},\tilde\beta),\mu|I^1_1)-u^1(U,\varrho^1_{I^1_1}(\beta^{-I^1_1},\tilde\beta),\mu|I^1_1)|\le\gamma$.

\noindent {\bf (i)}. Consider the scenario that ${\cal D}(I^1_1|\beta)(u^1(D,\varrho^1_{I^1_1}(\beta^{-I^1_1},\tilde\beta),\mu|I^1_1)-u^1(U,\varrho^1_{I^1_1}(\beta^{-I^1_1},\tilde\beta),\mu|I^1_1))>\gamma$. Then,  $\beta^2_{I^1_2}(R)<\frac{2}{3}-\frac{1-\beta^3_{I^1_3}(Y)+\frac{3}{2}\gamma}{3(4\beta^3_{I^1_3}(Y)-1)}$. The game has a class of WSREs given by $(D, (\frac{1}{4},\frac{3}{4}), (1-\beta^{*2}_{I^1_2}(R), \beta^{*2}_{I^1_2}(R)), Y)$ with $\beta^{*2}_{I^1_2}(R)<\frac{2}{3}$ and $\mu^{*3}_{I^1_3}(\langle D\rangle)=1$.

\noindent {\bf (ii)}. Consider the scenario that ${\cal D}(I^1_1|\beta)|u^1(D,\varrho^1_{I^1_1}(\beta^{-I^1_1},\tilde\beta),\mu|I^1_1)-u^1(U,\varrho^1_{I^1_1}(\beta^{-I^1_1},\tilde\beta),\mu|I^1_1)|\le\gamma$. Then,  $|\beta^2_{I^1_2}(R)(4\beta^3_{I^1_3}(Y)-1)-3\beta^3_{I^1_3}(Y)+1|\le\frac{1}{2}\gamma$. The game has a class of WSREs given by $(D, (\frac{1}{4},\frac{3}{4}), (\frac{1}{3}, \frac{2}{3}), Y)$ with $\mu^{*3}_{I^1_3}(\langle D\rangle)=1$.

\noindent {\bf Case (3)}. Suppose that ${\cal D}(I^1_3|\beta)|u^3(Y,\varrho^3_{I^1_3}(\beta^{-I^1_3}, \tilde\beta),\mu|I^1_3)-u^3(X,\varrho^3_{I^1_3}(\beta^{-I^1_3}, \tilde\beta),\mu|I^1_3)|\le\gamma$. Then,  $|\mu^3_{I^1_3}(\langle D\rangle)-\frac{5}{8}|\le\frac{1}{16}\gamma$, that is, {\small
\begin{equation}\label{edAwsrewsreeqC}\setlength{\abovedisplayskip}{1.2pt}
\setlength{\belowdisplayskip}{1.2pt}
\frac{10-\gamma}{6+\gamma}\beta^1_{I^1_1}(U) \beta^2_{I^1_2}(R) \beta^1_{I^2_1}(d)
\le\beta^1_{I^1_1}(D)\le \frac{10+\gamma}{6-\gamma}\beta^1_{I^1_1}(U) \beta^2_{I^1_2}(R) \beta^1_{I^2_1}(d).
\end{equation}}
Thus either ${\cal D}(I^1_1|\beta)(u^1(U,\varrho^1_{I^1_1}(\beta^{-I^1_1},\tilde\beta),\mu|I^1_1)-u^1(D,\varrho^1_{I^1_1}(\beta^{-I^1_1},\tilde\beta),\mu|I^1_1))>\gamma$ or ${\cal D}(I^1_1|\beta)|u^1(U$, $\varrho^1_{I^1_1}(\beta^{-I^1_1},\tilde\beta),\mu|I^1_1)-u^1(D,\varrho^1_{I^1_1}(\beta^{-I^1_1},\tilde\beta),\mu|I^1_1)|\le\gamma$.

\noindent {\bf (a)}. Assume that ${\cal D}(I^1_2|\beta)(u^2(L,\varrho^2_{I^1_2}(\beta^{-I^1_2},\tilde\beta),\mu|I^1_2)-u^2(R,\varrho^2_{I^1_2}(\beta^{-I^1_2},\tilde\beta),\mu|I^1_2))>\gamma$. Then, $\beta^2_{I^1_2}(R)\le\varepsilon$  and $\frac{3}{4}-\frac{1}{8}\gamma>\beta^1_{I^2_1}(d)\beta^3_{I^1_3}(Y)$. 

\noindent {\bf (i)}. Consider the scenario that ${\cal D}(I^1_1|\beta)(u^1(U,\varrho^1_{I^1_1}(\beta^{-I^1_1},\tilde\beta),\mu|I^1_1)-u^1(D,\varrho^1_{I^1_1}(\beta^{-I^1_1},\tilde\beta),\mu|I^1_1))>\gamma$.  Then, $\beta^3_{I^1_3}(Y)<\frac{\beta^2_{I^1_2}(L)+2\beta^2_{I^1_2}(R)\tilde\beta^1_{I^2_1}(u)-\frac{1}{2}\gamma}{3-4\beta^2_{I^1_2}(R)\tilde\beta^1_{I^2_1}(d)}$. Thus, ${\cal D}(I^2_1|\beta)(u^1(u,\beta^{-I^2_1}, \mu|I^2_1)-u^1(d,\varrho^1_{I^2_1}(\beta^{-I^2_1},\tilde\beta),\mu|I^2_1))>\gamma$ and consequently, $\beta^1_{I^2_1}(d)\le\varepsilon$. 
 The game has a class of WSREs given by $(U, u, L, (1-\beta^{*3}_{I^2_3}(Y), \beta^{*3}_{I^2_3}(Y)))$ with $\beta^{*3}_{I^2_3}(Y)<\frac{1}{3}$ and $\mu^{*3}_{I^1_3}(\langle D\rangle)=\frac{5}{8}$.
 
 \noindent {\bf (ii)}. Consider the scenario that ${\cal D}(I^1_1|\beta)|u^1(U,\varrho^1_{I^1_1}(\beta^{-I^1_1},\tilde\beta),\mu|I^1_1)-u^1(D,\varrho^1_{I^1_1}(\beta^{-I^1_1},\tilde\beta),\mu|I^1_1)|\le\gamma$.  Then, $|\beta^2_{I^1_2}(L)+2\beta^2_{I^1_2}(R)\tilde\beta^1_{I^2_1}(u)+(4\beta^2_{I^1_2}(R)\tilde\beta^1_{I^2_1}(d)-3)\beta^3_{I^1_3}(Y)|\le\frac{1}{2}\gamma$. Thus, ${\cal D}(I^2_1|\beta)(u^1(u,\beta^{-I^2_1}, \mu|I^2_1)-u^1(d,\varrho^1_{I^2_1}(\beta^{-I^2_1},\tilde\beta),\mu|I^2_1))>\gamma$ and consequently, $\beta^1_{I^2_1}(d)\le\varepsilon$. 
 The game has a class of WSREs given by $(U, u, L, (\frac{2}{3}, \frac{1}{3}))$ with $\mu^{*3}_{I^1_3}(\langle D\rangle)=\frac{5}{8}$.

\noindent {\bf (b)}. Assume that ${\cal D}(I^1_2|\beta)(u^2(R,\varrho^2_{I^1_2}(\beta^{-I^1_2},\tilde\beta),\mu|I^1_2)-u^2(L,\varrho^2_{I^1_2}(\beta^{-I^1_2},\tilde\beta),\mu|I^1_2))>\gamma$. Then, $\beta^2_{I^1_2}(L)\le\varepsilon$  and $\beta^1_{I^2_1}(d)\beta^3_{I^1_3}(Y)>\frac{3}{4}+\frac{1}{8}\gamma$. Thus, ${\cal D}(I^2_1|\beta)(u^1(d,\varrho^1_{I^2_1}(\beta^{-I^2_1},\tilde\beta),\mu|I^2_1)-u^1(u,\varrho^1_{I^2_1}(\beta^{-I^2_1},\tilde\beta),\mu|I^2_1))>\gamma$ and consequently, $\tilde\beta^1_{I^2_1}(u)=0$. 
These results together  with Eq.~(\ref{edAwsrewsreeqB}) ensure us that  ${\cal D}(I^1_1|\beta)|u^1(U$, $\varrho^1_{I^1_1}(\beta^{-I^1_1},\tilde\beta),\mu|I^1_1)-u^1(D,\varrho^1_{I^1_1}(\beta^{-I^1_1},\tilde\beta),\mu|I^1_1)|\le\gamma$. Thus, $\beta^2_{I^1_2}(L)+(4\beta^2_{I^1_2}(R)-3)\beta^3_{I^1_3}(Y)\le\frac{1}{2}\gamma$.
A contradiction occurs and
the assumption cannot be sustained.

\noindent {\bf (c)}. Assume that ${\cal D}(I^1_2|\beta)|u^2(R,\varrho^2_{I^1_2}(\beta^{-I^1_2},\tilde\beta),\mu|I^1_2)-u^2(L,\varrho^2_{I^1_2}(\beta^{-I^1_2},\tilde\beta),\mu|I^1_2)|\le\gamma$. Then, $|\beta^1_{I^2_1}(d)\beta^3_{I^1_3}(Y)-\frac{3}{4}|\le\frac{1}{8}\gamma$, which leads to $\frac{3}{4}-\frac{1}{8}\gamma\le \beta^3_{I^1_3}(Y)$ and
either ${\cal D}(I^2_1|\beta)(u^1(d,\varrho^1_{I^2_1}(\beta^{-I^2_1},\tilde\beta),\mu|I^2_1)-u^1(u$, $\varrho^1_{I^2_1}(\beta^{-I^2_1},\tilde\beta),\mu|I^2_1))>\gamma$ or
${\cal D}(I^2_1|\beta)|u^1(d,\varrho^1_{I^2_1}(\beta^{-I^2_1},\tilde\beta),\mu|I^2_1)-u^1(u,\varrho^1_{I^2_1}(\beta^{-I^2_1},\tilde\beta),\mu|I^2_1)|\le\gamma$. 

\noindent {\bf (i)}. Consider the scenario that ${\cal D}(I^2_1|\beta)(u^1(d,\varrho^1_{I^2_1}(\beta^{-I^2_1},\tilde\beta),\mu|I^2_1)-u^1(u,\varrho^1_{I^2_1}(\beta^{-I^2_1},\tilde\beta),\mu|I^2_1))>\gamma$. Then,  $\beta^1_{I^2_1}(u)\le\varepsilon$. Thus,   $u^1(d,\varrho^1_{I^2_1}(\beta^{-I^2_1},\tilde\beta),\mu|I^2_1)-u^1(u,\varrho^1_{I^2_1}(\beta^{-I^2_1},\tilde\beta),\mu|I^2_1)>\gamma$ and consequently, 
 $\tilde \beta^1_{I^2_1}(u)=0$. 

\noindent $\bullet$ Consider the situation that ${\cal D}(I^1_1|\beta)(u^1(U,\varrho^1_{I^1_1}(\beta^{-I^1_1},\tilde\beta),\mu|I^1_1)-u^1(D,\varrho^1_{I^1_1}(\beta^{-I^1_1},\tilde\beta),\mu|I^1_1))>\gamma$. Then, $ \beta^1_{I^1_1}(D)\le\varepsilon$. Thus it follows from Eq.~(\ref{edAwsrewsreeqC}) that ${\cal D}(I^1_1|\beta)(u^1(D,\varrho^1_{I^1_1}(\beta^{-I^1_1},\tilde\beta),\mu|I^1_1)-u^1(U$, $\varrho^1_{I^1_1}(\beta^{-I^1_1},\tilde\beta),\mu|I^1_1))>\gamma$. A contradiction occurs and the the situation is excluded.

\noindent $\bullet$ Consider the situation that ${\cal D}(I^1_1|\beta)|u^1(U,\varrho^1_{I^1_1}(\beta^{-I^1_1},\tilde\beta),\mu|I^1_1)-u^1(D,\varrho^1_{I^1_1}(\beta^{-I^1_1},\tilde\beta),\mu|I^1_1)|\le\gamma$. Then, 
$\frac{3\beta^3_{I^1_3}(Y)-1}{4\beta^3_{I^1_3}(Y)-1}-\frac{1}{2}\gamma\le\beta^2_{I^1_2}(R)\le\frac{3\beta^3_{I^1_3}(Y)-1}{4\beta^3_{I^1_3}(Y)-1}+\frac{1}{2}\gamma$.
The game has a  WSRE  given by $((\frac{24}{49}, \frac{25}{49}), d, (\frac{3}{8}, \frac{5}{8}), (\frac{2}{4},\frac{3}{4}))$ with $\mu^{*3}_{I^1_3}(\langle D\rangle)=\frac{5}{8}$

\noindent {\bf (ii)}. Consider the scenario that ${\cal D}(I^2_1|\beta)|u^1(d,\varrho^1_{I^2_1}(\beta^{-I^2_1},\tilde\beta),\mu|I^2_1)-u^1(u,\varrho^1_{I^2_1}(\beta^{-I^2_1},\tilde\beta),\mu|I^2_1)|\le\gamma$.  Then,  $\beta^1_{I^1_1}(U)\le\frac{\gamma}{8(\beta^3_{I^1_3}(Y)-\frac{1}{2})}$, which together with Eq.~(\ref{edAwsrewsreeqC}) yields
that $\beta^1_{I^1_1}(D)\le \frac{\gamma}{8(\beta^3_{I^1_3}(Y)-\frac{1}{2})}$. 
Thus, $1=\beta^1_{I^1_1}(U)+\beta^1_{I^1_1}(D)\le \frac{\gamma}{4(\beta^3_{I^1_3}(Y)-\frac{1}{2})}\le  \frac{\gamma}{4(\frac{3}{4}-\frac{1}{8}\gamma-\frac{1}{2})}=\frac{\gamma}{1-\frac{1}{2}\gamma}<1$. A contradiction occurs and the scenario is excluded.

The cases (1)-(3) together tell us that the game has five classes of WSREs given by
\begin{enumerate}
\item $(U, u, L, X)$ with $\mu^{*3}_{I^1_3}(\langle D\rangle)<\frac{5}{8}$;
\item  $(D, (1-\beta^{*1}_{I^2_1}(d),\beta^{*1}_{I^2_1}(d)), L, Y)$ with $\beta^{*1}_{I^2_1}(d)<\frac{3}{4}$ and $\mu^{*3}_{I^1_3}(\langle D\rangle)=1$;
\item $(D, (\frac{1}{4},\frac{3}{4}), (1-\beta^{*2}_{I^1_2}(R), \beta^{*2}_{I^1_2}(R)), Y)$ with $\beta^{*2}_{I^1_2}(R)\le\frac{2}{3}$ and $\mu^{*3}_{I^1_3}(\langle D\rangle)=1$; 
\item $(U, u, L, (1-\beta^{*3}_{I^1_3}(Y), \beta^{*3}_{I^1_3}(Y)))$ with $\beta^{*3}_{I^1_3}(Y)\le\frac{1}{3}$ and $\mu^{*3}_{I^1_3}(\langle D\rangle)=\frac{5}{8}$; and
\item $((\frac{24}{49}, \frac{25}{49}), d, (\frac{3}{8}, \frac{5}{8}), (\frac{2}{4},\frac{3}{4}))$ with $\mu^{*3}_{I^1_3}(\langle D\rangle)=\frac{5}{8}$.
\end{enumerate}
}
\end{example}

\section{\large Differentiable Path-Following Methods to Compute WSREs\label{spwsre}}

This section exploits Theorem~\ref{nscwsrethm1} to devise differentiable path-following methods to compute WSREs.\footnote{A general framework for establishing such a differentiable path-following method can be described as follows.
Step 1: Constitute with an extra variable $t\in (0,1]$ an artificial extensive-form game $\Gamma(t)$ in which each player at each of his information sets solves a convex optimization problem. The artificial game should continuously deform from a trivial game to the target game as $t$ descends from one to zero. $\Gamma(1)$ should have a unique equilibrium, which can be easily computed, and every convergent sequence of equilibria of $\Gamma(t_k)$, $k=1,2,\ldots$, with $\lim\limits_{k\to\infty}t_k=0$ should yield a desired equilibrium at its limit.
Step 2: Apply the optimality conditions to the convex optimization problems in the artificial game to acquire from the equilibrium condition an equilibrium system.
Step 3: Verify that the closure of the set of solutions of the equilibrium system contains a path-connected component that intersects both the levels of $t=1$ and $t=0$.  
Step 4: Ensure through an application of the Transversality Theorem in Eaves and Schmedders~\cite{Eaves and Schmedders (1999)} the existence of a smooth path that starts from the unique equilibrium at $t=1$ and approaches a desired equilibrium as $t\to 0$.
Step 5: Adapt a standard predictor-corrector method for numerically tracing the smooth path to a desired equilibrium.
} 
Let $\eta^0=(\eta^{0i}_{I^j_i}:i\in N,j\in M_i)$  be a given vector with $\eta^{0i}_{I^j_i}=(\eta^{0i}_{I^j_i}(a):a\in A(I^j_i))^\top$ such that $0<\eta^{0i}_{I^j_i}(a)$ and $\tau^i_{I^j_i}(\eta^0)=\sum\limits_{a\in A(I^j_i)}\eta^{0i}_{I^j_i}(a)< 1$. For $t\in [0,1]$,
let $\varpi(\beta,t)=(\varpi(\beta^q_{I^l_q},t):q\in N,l\in M_q)$ with $\varpi(\beta^q_{I^l_q},t)=(\varpi(\beta^q_{I^l_q}(a),t):a\in A(I^l_q))^\top$, where $\varpi(\beta^q_{I^l_q}(a),t)=(1-t(1-t)\tau^q_{I^l_q}(\eta^0))\beta^q_{I^l_q}(a)+t(1-t)\eta^{0q}_{I^l_q}(a)$. We denote by $\beta^0$ and $\tilde\beta^0$ two given totally mixed behavioral strategy profiles. 
When $\beta>0$,  we define $u^i(\beta|I^j_i)=\frac{u^i(\beta\land I^j_i)}{\omega(I^j_i|\beta)}$ and $u^i(a,\beta^{-I^j_i}|I^j_i)=\frac{u^i((a,\beta^{-I^j_i})\land I^j_i)}{\omega(I^j_i|\beta)}$.   Then, for $\beta>0$, it holds that $u^i(\beta,\mu(\beta)|I^j_i)=\frac{u^i(\beta\land I^j_i)}{\omega(I^j_i|\beta)}=u^i(\beta|I^j_i)$ and $u^i(a,\beta^{-I^j_i},\mu(\beta)|I^j_i)=\frac{u^i((a,\beta^{-I^j_i})\land I^j_i)}{\omega(I^j_i|\beta)}=u^i(a,\beta^{-I^j_i}|I^j_i)$. 
Given these notations, we establish barrier and penalty smooth paths to a WSRE.

\subsection{A Logarithmic-Barrier Smooth Path to a WSRE\label{logbwsre}}

For  $t\in (0,1]$, we constitute with $\varpi(\beta, t)$ a logarithmic-barrier extensive-form game $\Gamma_L(t)$ in which player $i$ at information set $I^j_i$ solves against a given pair  $(\hat\beta,\hat{\tilde\beta})$ the strictly convex optimization problem, {\footnotesize
\begin{equation}\setlength{\abovedisplayskip}{0pt}\setlength{\belowdisplayskip}{0pt}\label{logbop1wsre}
\begin{array}{rl}
\max\limits_{\beta^i_{I^j_i},\;\tilde\beta^i_{I^j_i}} & (1-t)\sum\limits_{a\in A(I^j_i)}(\beta^i_{I^j_i}(a){\cal D}(I^j_i|\varpi(\hat\beta,t)) +\tilde\beta^i_{I^j_i}(a))u^i(a,\varrho^i_{I^j_i}(\varpi(\hat\beta^{-I^j_i},t),\hat{\tilde\beta})| I^j_i)\\
& 
+t\sum\limits_{a\in A(I^j_i)}(\beta^{0i}_{I^j_i}(a)\ln\beta^i_{I^j_i}(a) +\tilde\beta^{0i}_{I^j_i}(a)\ln\tilde\beta^i_{I^j_i}(a))\\
\text{s.t.} & \sum\limits_{a\in A(I^j_i)}\beta^i_{I^j_i}(a)=1,\;\sum\limits_{a\in A(I^j_i)}\tilde\beta^i_{I^j_i}(a)=1.
\end{array}
\end{equation}}Applying the optimality conditions to the problem~(\ref{logbop1wsre}), we acquire from the equilibrium condition of $(\hat\beta,\hat{\tilde\beta})=(\beta,\tilde\beta)$ the equilibrium system of $\Gamma_L(t)$, {\footnotesize
\begin{equation}\setlength{\abovedisplayskip}{0pt}\setlength{\belowdisplayskip}{0pt}\label{logbes1wsre}
\begin{array}{l}
 (1-t){\cal D}(I^j_i|\varpi(\beta,t)) u^i(a,\varrho^i_{I^j_i}(\varpi(\beta^{-I^j_i},t),\tilde\beta)| I^j_i)
+t\beta^{0i}_{I^j_i}(a)/\beta^i_{I^j_i}(a)-\zeta^i_{I^j_i}=0,\;i\in N,j\in M_i,a\in A(I^j_i),\\

(1-t)u^i(a,\varrho^i_{I^j_i}(\varpi(\beta^{-I^j_i},t),\tilde\beta)| I^j_i)
+t\tilde\beta^{0i}_{I^j_i}(a)/\tilde\beta^i_{I^j_i}(a)-\tilde\zeta^i_{I^j_i}=0,\;
i\in N,j\in M_i,a\in A(I^j_i),\\

 \sum\limits_{a\in A(I^j_i)}\beta^i_{I^j_i}(a)=1,\;\sum\limits_{a\in A(I^j_i)}\tilde\beta^i_{I^j_i}(a)=1,\;i\in N, j\in M_i;\;

 0<\beta^i_{I^j_i},\;0<\tilde\beta^i_{I^j_i},\;i\in N, j\in M_i, a\in A(I^j_i).
\end{array}
\end{equation}}Multiplying $\beta^i_{I^j_i}(a)$ and $\tilde\beta^i_{I^j_i}(a)$ to equations in the first and second groups of the system~(\ref{logbes1wsre}), respectively, we come to the system, {\footnotesize
\begin{equation}\setlength{\abovedisplayskip}{0pt}\setlength{\belowdisplayskip}{0pt}\label{logbes2wsre}
\begin{array}{l}
 (1-t)\beta^i_{I^j_i}(a){\cal D}(I^j_i|\varpi(\beta,t)) u^i(a,\varrho^i_{I^j_i}(\varpi(\beta^{-I^j_i},t),\tilde\beta)| I^j_i)
+t\beta^{0i}_{I^j_i}(a)-\beta^i_{I^j_i}(a)\zeta^i_{I^j_i}=0,\\
\hspace{11.6cm}\;i\in N,j\in M_i,a\in A(I^j_i),\\

(1-t)\tilde\beta^i_{I^j_i}(a)u^i(a,\varrho^i_{I^j_i}(\varpi(\beta^{-I^j_i},t),\tilde\beta)| I^j_i)
+t\tilde\beta^{0i}_{I^j_i}(a)-\tilde\beta^i_{I^j_i}(a)\tilde\zeta^i_{I^j_i}=0,\;
i\in N,j\in M_i,a\in A(I^j_i),\\

 \sum\limits_{a\in A(I^j_i)}\beta^i_{I^j_i}(a)=1,\;\sum\limits_{a\in A(I^j_i)}\tilde\beta^i_{I^j_i}(a)=1,\;i\in N, j\in M_i;\;
 
 0<\beta^i_{I^j_i},\;0<\tilde\beta^i_{I^j_i},\;i\in N, j\in M_i, a\in A(I^j_i).
\end{array}
\end{equation}}Taking the sum of equations in the first group and the sum of equations in the second group
of the system~(\ref{logbes2wsre}) over $A(I^j_i)$, respectively, we get from a division operation the system, {\footnotesize
\begin{equation}\setlength{\abovedisplayskip}{0pt}\setlength{\belowdisplayskip}{0pt}\label{logbes3wsre}
\begin{array}{l}
\zeta^i_{I^j_i}= \frac{1}{\sum\limits_{a'\in A(I^j_i)}\beta^i_{I^j_i}(a')}( (1-t){\cal D}(I^j_i|\varpi(\beta,t))\sum\limits_{a'\in A(I^j_i)} \beta^i_{I^j_i}(a')u^i(a',\varrho^i_{I^j_i}(\varpi(\beta^{-I^j_i},t),\tilde\beta)| I^j_i)+t),\\

\tilde\zeta^i_{I^j_i}=\frac{1}{\sum\limits_{a'\in A(I^j_i)}\tilde\beta^i_{I^j_i}(a')}((1-t)\sum\limits_{a'\in A(I^j_i)}\tilde\beta^i_{I^j_i}(a')u^i(a',\varrho^i_{I^j_i}(\varpi(\beta^{-I^j_i},t),\tilde\beta)| I^j_i)
+t),\;i\in N,j\in M_i.
\end{array}
\end{equation}}We denote by $a^0_{I^j_i}$ a given reference action in $A(I^j_i)$.
Substituting $\zeta^i_{I^j_i}$ and $\tilde\zeta^i_{I^j_i}$ of the system~(\ref{logbes3wsre}) into the system~(\ref{logbes2wsre}), we arrive at the system with fewer variables, {\footnotesize
\begin{equation}\setlength{\abovedisplayskip}{0pt}\setlength{\belowdisplayskip}{0pt}\label{logbes4wsre}
\begin{array}{l}
 (1-t)\beta^i_{I^j_i}(a){\cal D}(I^j_i|\varpi(\beta,t))\sum\limits_{a'\in A(I^j_i)}\beta^i_{I^j_i}(a') (u^i(a,\varrho^i_{I^j_i}(\varpi(\beta^{-I^j_i},t),\tilde\beta)| I^j_i)-u^i(a',\varrho^i_{I^j_i}(\varpi(\beta^{-I^j_i},t),\tilde\beta)| I^j_i))\\
 \hspace{5cm}
+t(\beta^{0i}_{I^j_i}(a)\sum\limits_{a'\in A(I^j_i)}\beta^i_{I^j_i}(a')-\beta^i_{I^j_i}(a))=0,\;i\in N,j\in M_i,a\in A(I^j_i)\backslash\{a^0_{I^j_i}\},\\

(1-t)\tilde\beta^i_{I^j_i}(a)\sum\limits_{a'\in A(I^j_i)}\tilde\beta^i_{I^j_i}(a') (u^i(a,\varrho^i_{I^j_i}(\varpi(\beta^{-I^j_i},t),\tilde\beta)| I^j_i)-u^i(a',\varrho^i_{I^j_i}(\varpi(\beta^{-I^j_i},t),\tilde\beta)| I^j_i))\\
\hspace{5cm}+t(\tilde\beta^{0i}_{I^j_i}(a)\sum\limits_{a'\in A(I^j_i)}\tilde\beta^i_{I^j_i}(a')-\tilde\beta^i_{I^j_i}(a))=0,\;i\in N,j\in M_i,a\in A(I^j_i)\backslash\{a^0_{I^j_i}\},\\

 \sum\limits_{a\in A(I^j_i)}\beta^i_{I^j_i}(a)=1,\;\sum\limits_{a\in A(I^j_i)}\tilde\beta^i_{I^j_i}(a)=1,\;i\in N, j\in M_i;\;
 
 0<\beta^i_{I^j_i},\;0<\tilde\beta^i_{I^j_i},\;i\in N, j\in M_i, a\in A(I^j_i).
\end{array}
\end{equation}}Since $u^i(a,\varrho^i_{I^j_i}(\varpi(\beta^{-I^j_i},t),\tilde\beta)| I^j_i)=u^i((a,\varrho^i_{I^j_i}(\varpi(\beta^{-I^j_i},t),\tilde\beta))\land I^j_i)/\omega(I^j_i|\varpi(\beta, t))$,  we secure from the multiplications of $\omega(I^j_i|\varpi(\beta, t))$ to equations in the first and second groups of the system~(\ref{logbes4wsre}) the system, {\footnotesize
\begin{equation}\setlength{\abovedisplayskip}{0pt}\setlength{\belowdisplayskip}{0pt}\label{logbes5wsre}
\begin{array}{l}
 (1-t)\beta^i_{I^j_i}(a){\cal D}(I^j_i|\varpi(\beta,t))\sum\limits_{a'\in A(I^j_i)}\beta^i_{I^j_i}(a') (u^i((a,\varrho^i_{I^j_i}(\varpi(\beta^{-I^j_i},t),\tilde\beta))\land I^j_i)\\
 \hspace{0.7cm}-u^i((a',\varrho^i_{I^j_i}(\varpi(\beta^{-I^j_i},t),\tilde\beta))\land I^j_i))
+t\omega(I^j_i|\varpi(\beta, t))(\beta^{0i}_{I^j_i}(a)\sum\limits_{a'\in A(I^j_i)}\beta^i_{I^j_i}(a')-\beta^i_{I^j_i}(a))=0,\\
\hspace{10.4cm}i\in N,j\in M_i,a\in A(I^j_i)\backslash\{a^0_{I^j_i}\},\\

(1-t)\tilde\beta^i_{I^j_i}(a)\sum\limits_{a'\in A(I^j_i)}\tilde\beta^i_{I^j_i}(a') (u^i((a,\varrho^i_{I^j_i}(\varpi(\beta^{-I^j_i},t),\tilde\beta))\land I^j_i)-u^i((a',\varrho^i_{I^j_i}(\varpi(\beta^{-I^j_i},t),\tilde\beta))\land I^j_i))\\
 \hspace{2.7cm}
+t\omega(I^j_i|\varpi(\beta, t))(\tilde\beta^{0i}_{I^j_i}(a)\sum\limits_{a'\in A(I^j_i)}\tilde\beta^i_{I^j_i}(a')-\tilde\beta^i_{I^j_i}(a))=0,\;i\in N,j\in M_i,a\in A(I^j_i)\backslash\{a^0_{I^j_i}\},\\

 \sum\limits_{a\in A(I^j_i)}\beta^i_{I^j_i}(a)=1,\;\sum\limits_{a\in A(I^j_i)}\tilde\beta^i_{I^j_i}(a)=1,\;i\in N, j\in M_i;\;
 
 0<\beta^i_{I^j_i},\;0<\tilde\beta^i_{I^j_i},\;i\in N, j\in M_i, a\in A(I^j_i).
\end{array}
\end{equation}}When $t=1$, the system~(\ref{logbes5wsre}) has a unique solution given by $(\beta^*(1),\tilde\beta^*(1))$ with $\beta^{*i}_{I^j_i}(1; a)=\beta^{0i}_{I^j_i}(a)$ and $\tilde\beta^{*i}_{I^j_i}(1; a)=\tilde\beta^{0i}_{I^j_i}(a)$ for $i\in N, j\in M_i, a\in A(I^j_i)$.
Let $\widetilde{\mathscr{E}}_L$ be the set of all $(\beta,\tilde\beta,t)$ satisfying the system~(\ref{logbes5wsre}) with $t>0$ and $\mathscr{E}_L$ the closure of $\widetilde{\mathscr{E}}_L$. As a corollary of Theorem~\ref{nscwsrethm1}, one can draw the following conclusion.
\begin{corollary}\label{logbwsrethm1} {\em Every $(\beta^*,\tilde\beta^*,0)\in \mathscr{E}_L$ yields a WSRE.}
\end{corollary}

To continue the developments, we need a fixed point theorem from Mas-Colell~\cite{Mas-Colell (1974)}.

\begin{theorem}[Mas-Colell's fixed point theorem] \label{fpthm} {\em Let $C$ be a nonempty, compact and convex subset of $\mathbb{R}^m$ and $F: C\times [0,1]\rightarrow C$ an upper-hemicontinuous mapping. Then the set $H=\{(z,t)\in C\times [0,1]\;|\;z\in F(z,t)\}$ contains a connected set $H^c$ such that $(C\times\{1\})\cap H^c\neq\emptyset$ and $(C\times\{0\})\cap H^c\neq\emptyset$.}\end{theorem}

As a corollary of Theorem~\ref{fpthm}, we come to the following result. 
\begin{corollary} \label{ccthm1} {\em $\mathscr{E}_L$ contains a unique connected component $\mathscr{E}_L^c$ such that $\mathscr{E}_L^c\cap(\triangle\times\triangle\times\{0\})\ne\emptyset$ and $\mathscr{E}_A^c\cap(\triangle\times\triangle\times\{1\})\ne\emptyset$.}
\end{corollary}

Let $m_0=\sum\limits_{i\in N}\sum\limits_{j\in M_i}|A(I^j_i)|$. 
We denote by $\widetilde{\mathscr{G}}_0$ the set of all $(\beta,\tilde\beta,t)$ satisfying the system~(\ref{logbes5wsre}) with $0<t\le 1$ and by $\mathscr{G}_0$ the closure of $\widetilde{\mathscr{G}}_0$.
Let $g_0(\beta,\tilde\beta,t)$ denote the left-hand sides of equations in the system~(\ref{logbes5wsre}). Subtracting a perturbation term of $t(1-t)\alpha\in\mathbb{R}^{2m_0}$ from $g_0(\beta,\tilde\beta,t)$, we arrive at the system,
\(g_0(\beta,\tilde\beta,t)-t(1-t)\alpha=0.\)
Let $g(\beta,\tilde\beta,t;\alpha)=g_0(\beta,\tilde\beta,t)-t(1-t)\alpha$. For any given $\alpha\in\mathbb{R}^{2m_0}$, we denote $g_\alpha(\beta,\tilde\beta,t)=g(\beta,\tilde\beta,t;\alpha)$.
Let $\widetilde{\mathscr{G}}_\alpha=\{(\beta,\tilde\beta,t)|g_\alpha(\beta,\tilde\beta,t)=0\text{ with }0< t\le 1\}$ and $\mathscr{G}_\alpha$ be the closure of $\widetilde {\mathscr{G}}_\alpha$.
One can obtain that $\mathscr{G}_\alpha$ is a compact set and $g(\beta,\tilde\beta,t;\alpha)$ is continuously differentiable on $\mathbb{R}^{2m_0}\times (0,1)\times \mathbb{R}^{2m_0}$ with its Jacobian matrix given by $D_{\alpha}g(\beta,\tilde\beta,t;\alpha)=t(1-t)I_{2m_0}$, where $I_{2m_0}$ is an identity matrix. Clearly, as $0<t<1$, $D_{\alpha}g(\beta,\tilde\beta,t;\alpha)$ is nonsingular. Furthermore, when $t=1$, the Jacobian matrix, $D_{(\beta,\;\tilde\beta)}g_0(\beta,\tilde\beta,1)$, is nonsingular. When $\|\alpha\|>0$ is sufficiently small, Corollary~\ref{ccthm1} together with the continuity of $g(\beta,\tilde\beta,t;\alpha)$  tells us that there is a unique connected component in $\mathscr{G}_\alpha$ intersecting both $\mathbb{R}^{2m_0}\times\{1\}$ and $\mathbb{R}^{2m_0}\times\{0\}$.
These results together with the well-known transversality theorem in Eaves and Schmedders~\cite{Eaves and Schmedders (1999)} and implicit function theorem lead us to the following conclusion.
\begin{theorem}\label{efgthm6}{\em For generic choice of $\alpha$ with sufficiently small $\|\alpha\|>0$, there exists a smooth path  $P_{\alpha}\subseteq\mathscr{G}_{\alpha}$ that starts from the unique solution $(\beta^*(1),\tilde\beta^*(1),1)$ on the level of $t=1$ and ends at a WSRE of $\Gamma$ on the target level of $t=0$.}
\end{theorem}

One can adapt a standard predictor-corrector method in Eaves and Schmedders~\cite{Eaves and Schmedders (1999)} for numerically tracing the smooth paths specified by the system~(\ref{logbes5wsre}) to find a WSRE. The smooth paths specified by the system~(\ref{logbes5wsre}) for the games in Figs.~\ref{TFigure1}-\ref{TFigure2} are illustrated in Figs.~\ref{TFigure1MLogB1}-\ref{TFigure2MLogB2}.
% and the system~(\ref{Slogbes2wsre})
\begin{figure}[H]
    \centering
    \begin{minipage}{0.49\textwidth}
        \centering
        \includegraphics[width=0.80\textwidth, height=0.15\textheight]{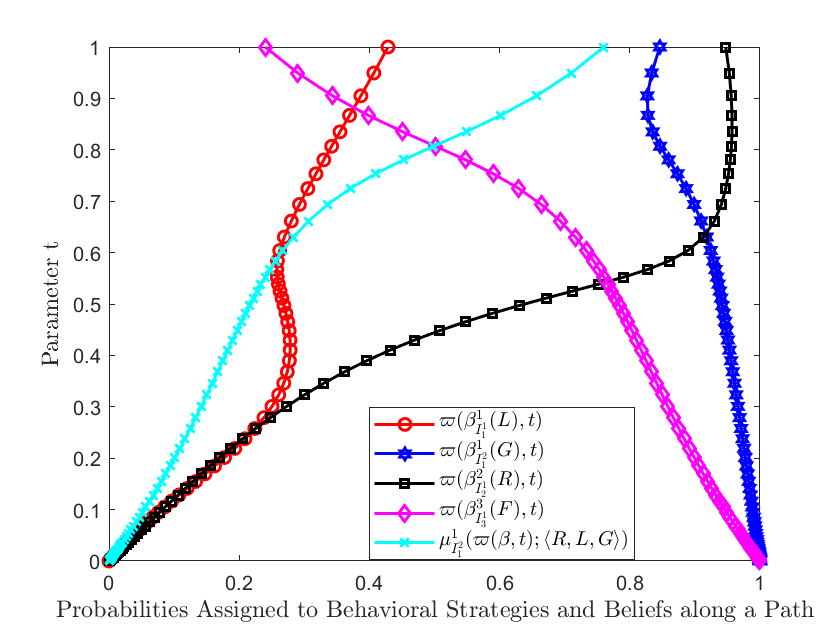}
        % first figure itself
               \caption{\label{TFigure1MLogB1}\scriptsize The Smooth Path of $\varpi(\beta,t)$ Specified by the System~(\ref{logbes5wsre}) for the Game in Fig.~\ref{TFigure1}}
\end{minipage}\hfill
    \begin{minipage}{0.49\textwidth}
        \centering
        \includegraphics[width=0.80\textwidth, height=0.15\textheight]{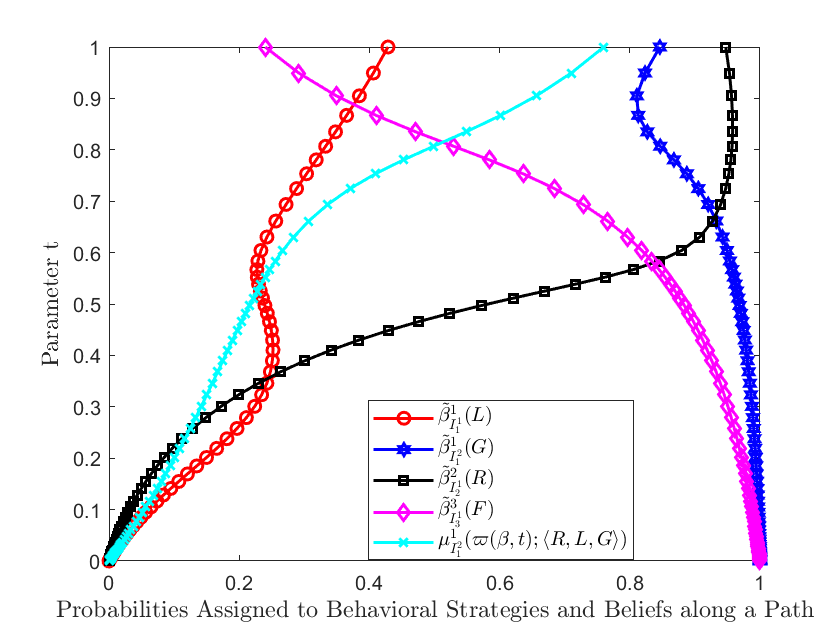}
        % first figure itself
        \caption{\label{TFigure1MLogB2}\scriptsize The Smooth Path of $\tilde\beta$ Specified by the System~(\ref{logbes5wsre}) for the Game in Fig.~\ref{TFigure1}}
\end{minipage}
 \end{figure}

\begin{figure}[H]
    \centering
    \begin{minipage}{0.49\textwidth}
        \centering
        \includegraphics[width=0.80\textwidth, height=0.15\textheight]{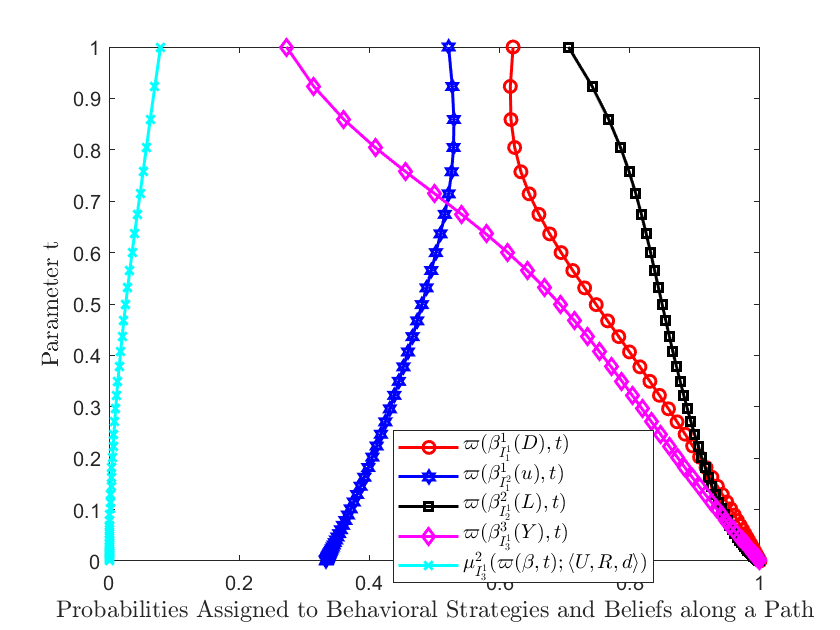}
        % first figure itself
               \caption{\label{TFigure2MLogB1}\scriptsize  The Smooth Path of $\varpi(\beta, t)$ Specified by the System~(\ref{logbes5wsre}) for the Game in Fig.~\ref{TFigure2}}
\end{minipage}\hfill
    \begin{minipage}{0.49\textwidth}
        \centering
        \includegraphics[width=0.80\textwidth, height=0.15\textheight]{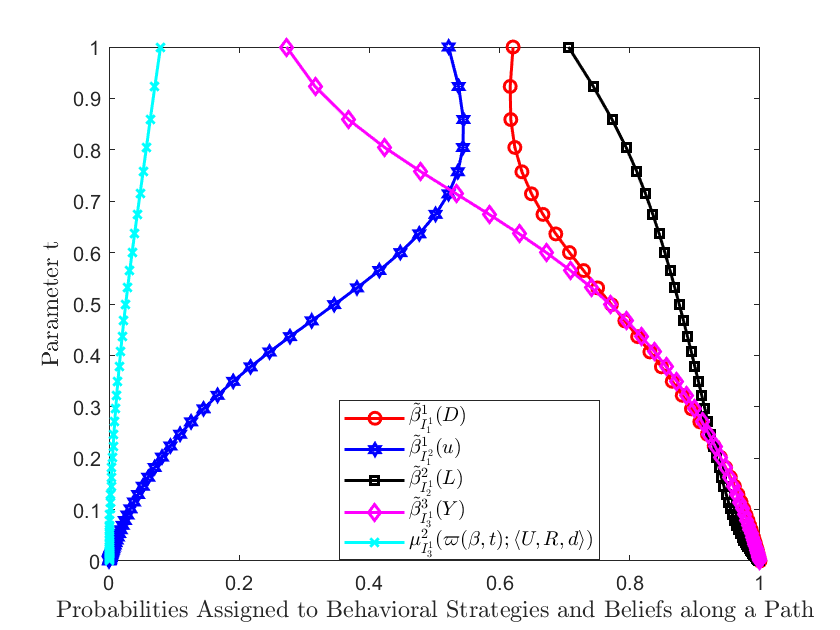}
        % first figure itself
        \caption{\label{TFigure2MLogB2}\scriptsize  The Smooth Path of $\tilde\beta$ Specified by the System~(\ref{logbes5wsre}) for the Game in Fig.~\ref{TFigure2}}
\end{minipage}
\end{figure}

\subsection{An Entropy-Barrier Smooth Path to a WSRE}

To derive a more efficient method, we propose an alternative scheme in this subsection. For $t\in (0,1]$, we constitute with $\varpi(\beta, t)$ an entropy-barrier extensive-form game $\Gamma_E(t)$ in which player $i$ at information set $I^j_i$ solves against a given pair  $(\hat\beta,\hat{\tilde\beta})$ the strictly convex optimization problem, {\footnotesize
\begin{equation}\setlength{\abovedisplayskip}{0pt}\setlength{\belowdisplayskip}{0pt}\label{entbop1wsre}
\begin{array}{rl}
\max\limits_{\beta^i_{I^j_i},\;\tilde\beta^i_{I^j_i}} & (1-t)\sum\limits_{a\in A(I^j_i)}(\beta^i_{I^j_i}(a){\cal D}(I^j_i|\varpi(\hat\beta,t)) +\tilde\beta^i_{I^j_i}(a))u^i(a,\varrho^i_{I^j_i}(\varpi(\hat\beta^{-I^j_i},t),\hat{\tilde\beta})| I^j_i)\\
& 
-t\sum\limits_{a\in A(I^j_i)}(\beta^i_{I^j_i}(a)(\ln(\beta^i_{I^j_i}(a)/\beta^{0i}_{I^j_i}(a))-1)+\tilde\beta^i_{I^j_i}(a)(\ln(\tilde\beta^i_{I^j_i}(a)/\tilde\beta^{0i}_{I^j_i}(a)) -1))\\
\text{s.t.} & \sum\limits_{a\in A(I^j_i)}\beta^i_{I^j_i}(a)=1,\;\sum\limits_{a\in A(I^j_i)}\tilde\beta^i_{I^j_i}(a)=1.
\end{array}
\end{equation}}Applying the optimality conditions to the problem~(\ref{entbop1wsre}), we acquire from the equilibrium condition of $(\hat\beta,\hat{\tilde\beta})=(\beta,\tilde\beta)$ the equilibrium system of $\Gamma_E(t)$, {\footnotesize
\begin{equation}\setlength{\abovedisplayskip}{0pt}\setlength{\belowdisplayskip}{0pt}\label{entbes1wsre}
\begin{array}{l}
 (1-t){\cal D}(I^j_i|\varpi(\beta,t)) u^i(a,\varrho^i_{I^j_i}(\varpi(\beta^{-I^j_i},t),\tilde\beta)| I^j_i)
-t\ln(\beta^i_{I^j_i}(a)/\beta^{0i}_{I^j_i}(a))-\zeta^i_{I^j_i}=0,\;i\in N,j\in M_i,a\in A(I^j_i),\\

(1-t)u^i(a,\varrho^i_{I^j_i}(\varpi(\beta^{-I^j_i},t),\tilde\beta)| I^j_i)
-t\ln(\tilde\beta^i_{I^j_i}(a)/\tilde\beta^{0i}_{I^j_i}(a))-\tilde\zeta^i_{I^j_i}=0,\;
i\in N,j\in M_i,a\in A(I^j_i),\\

 \sum\limits_{a\in A(I^j_i)}\beta^i_{I^j_i}(a)=1,\;\sum\limits_{a\in A(I^j_i)}\tilde\beta^i_{I^j_i}(a)=1,\;i\in N, j\in M_i;\;

 0<\beta^i_{I^j_i},\;0<\tilde\beta^i_{I^j_i},\;i\in N, j\in M_i, a\in A(I^j_i).
\end{array}
\end{equation}}

\noindent We define \[\setlength{\abovedisplayskip}{0pt}\setlength{\belowdisplayskip}{0pt}
d(x)=\left\{\begin{array}{ll}
\exp(1-1/x) & \text{if $0<x$,}\\
 0 & \text{otherwise.}
 \end{array}\right.\]
Let $v=(v^i_{I^j_i}(a):i\in N, j\in M_i, a\in A(I^j_i))$ and $\tilde v=(\tilde v^i_{I^j_i}(a):i\in N, j\in M_i, a\in A(I^j_i))$. Substituting 
$\beta^i_{I^j_i}(v;a)=d(v^i_{I^i_j}(a))$ and $\tilde\beta^i_{I^j_i}(\tilde v;a)=d(
\tilde v^i_{I^i_j}(a))$ into the system~(\ref{entbes1wsre}) for $\beta^i_{I^j_i}(a)$ and $\tilde\beta^i_{I^j_i}(a)$, we come to an equivalent system, {\footnotesize
\begin{equation}\setlength{\abovedisplayskip}{0pt}\setlength{\belowdisplayskip}{0pt}\label{entbes2wsre}
\begin{array}{l}
 (1-t){\cal D}(I^j_i|\varpi(\beta(v),t)) u^i(a,\varrho^i_{I^j_i}(\varpi(\beta^{-I^j_i}(v),t),\tilde\beta(\tilde v))| I^j_i)
+t/v^i_{I^j_i}(a)-t(1-\ln\beta^{0i}_{I^j_i}(a))-\zeta^i_{I^j_i}=0,\\
\hspace{12.2cm}i\in N,j\in M_i,a\in A(I^j_i),\\

(1-t)u^i(a,\varrho^i_{I^j_i}(\varpi(\beta^{-I^j_i}(v),t),\tilde\beta(\tilde v))| I^j_i)
+t/\tilde v^i_{I^j_i}(a)-t(1-\ln\tilde\beta^{0i}_{I^j_i}(a))-\tilde\zeta^i_{I^j_i}=0,\;
i\in N,j\in M_i,a\in A(I^j_i),\\

 \sum\limits_{a\in A(I^j_i)}\beta^i_{I^j_i}(v; a)=1,\;\sum\limits_{a\in A(I^j_i)}\tilde\beta^i_{I^j_i}(\tilde v; a)=1,\;i\in N, j\in M_i;\;

 0<v^i_{I^j_i},\;0<\tilde v^i_{I^j_i},\;i\in N, j\in M_i, a\in A(I^j_i).
\end{array}
\end{equation}}

\noindent Multiplying $v^i_{I^j_i}(a)$ to the first group of equations and $\tilde v^i_{I^j_i}(a)$ to the second group of equations in the system~(\ref{entbes2wsre}), we arrive at the system,
 {\footnotesize
\begin{equation}\setlength{\abovedisplayskip}{0pt}\setlength{\belowdisplayskip}{0pt}\label{entbes3wsre}
\begin{array}{l}
 (1-t){\cal D}(I^j_i|\varpi(\beta(v),t)) v^i_{I^j_i}(a)u^i(a,\varrho^i_{I^j_i}(\varpi(\beta^{-I^j_i}(v),t),\tilde\beta(\tilde v))| I^j_i)
+t-tv^i_{I^j_i}(a)(1-\ln\beta^{0i}_{I^j_i}(a))\\
\hspace{9.6cm}-\zeta^i_{I^j_i}v^i_{I^j_i}(a)=0,\;i\in N,j\in M_i,a\in A(I^j_i),\\

(1-t)\tilde v^i_{I^j_i}(a)u^i(a,\varrho^i_{I^j_i}(\varpi(\beta^{-I^j_i}(v),t),\tilde\beta(\tilde v))| I^j_i)
+t-t\tilde v^i_{I^j_i}(a)(1-\ln\tilde\beta^{0i}_{I^j_i}(a))-\tilde\zeta^i_{I^j_i}\tilde v^i_{I^j_i}(a)=0,\\
\hspace{12.2cm}i\in N,j\in M_i,a\in A(I^j_i),\\

 \sum\limits_{a\in A(I^j_i)}\beta^i_{I^j_i}(v; a)=1,\;\sum\limits_{a\in A(I^j_i)}\tilde\beta^i_{I^j_i}(\tilde v; a)=1,\;i\in N, j\in M_i;\;

 0<v^i_{I^j_i},\;0<\tilde v^i_{I^j_i},\;i\in N, j\in M_i, a\in A(I^j_i).
\end{array}
\end{equation}}

\noindent Taking the sum of equations in the first group and the sum of equations in the second group of the system~(\ref{entbes3wsre}) over $A(I^j_i)$, we get after a simplification the system,
 {\footnotesize
\begin{equation}\setlength{\abovedisplayskip}{0pt}\setlength{\belowdisplayskip}{0pt}\label{entbes4wsre}
\begin{array}{rl}
\zeta^i_{I^j_i}= & \frac{1}{\sum\limits_{a'\in A(I^j_i)}v^i_{I^j_i}(a')}((1-t){\cal D}(I^j_i|\varpi(\beta(v),t))\sum\limits_{a'\in A(I^j_i)}v^i_{I^j_i}(a')u^i(a',\varrho^i_{I^j_i}(\varpi(\beta^{-I^j_i}(v),t),\tilde\beta(\tilde v))| I^j_i)\\
& +t|A(I^j_i)|-t\sum\limits_{a'\in A(I^j_i)}v^i_{I^j_i}(a')(1-\ln\beta^{0i}_{I^j_i}(a'))),\;i\in N,j\in M_i,\\

\tilde\zeta^i_{I^j_i} = & \frac{1}{\sum\limits_{a'\in A(I^j_i)}\tilde v^i_{I^j_i}(a')}((1-t)\sum\limits_{a'\in A(I^j_i)}\tilde v^i_{I^j_i}(a')u^i(a',\varrho^i_{I^j_i}(\varpi(\beta^{-I^j_i}(v),t),\tilde\beta(\tilde v))| I^j_i)
\\ 
& +t|A(I^j_i)|-t\sum\limits_{a'\in A(I^j_i)}\tilde v^i_{I^j_i}(a')(1-\ln\tilde\beta^{0i}_{I^j_i}(a'))),\;i\in N,j\in M_i.
\end{array}
\end{equation}}

\noindent 
Let $a^0_{I^j_i}\in A(I^i_j)$ be a given reference action. Substituting $\zeta^i_{I^j_i}$ and $\tilde\zeta^i_{I^j_i}$ of the system~(\ref{entbes4wsre}) into the system~(\ref{entbes3wsre}), we harvest after a simplifiation an equivalent system with much fewer variables,
{\footnotesize
\begin{equation}\setlength{\abovedisplayskip}{0pt}\setlength{\belowdisplayskip}{0pt}\label{entbes5wsre}
\begin{array}{l}
 (1-t){\cal D}(I^j_i|\varpi(\beta(v),t)) v^i_{I^j_i}(a)\sum\limits_{a'\in A(I^j_i)}v^i_{I^j_i}(a')(u^i(a,\varrho^i_{I^j_i}(\varpi(\beta^{-I^j_i}(v),t),\tilde\beta(\tilde v))| I^j_i)\\
 \hspace{3cm}-u^i(a',\varrho^i_{I^j_i}(\varpi(\beta^{-I^j_i}(v),t),\tilde\beta(\tilde v))| I^j_i))
+t\sum\limits_{a'\in A(I^j_i)}(v^i_{I^j_i}(a')-v^i_{I^j_i}(a)) \\
\hspace{4cm}-tv^i_{I^j_i}(a)\sum\limits_{a'\in A(I^j_i)}v^i_{I^j_i}(a')\ln(\beta^{0i}_{I^j_i}(a')/\beta^{0i}_{I^j_i}(a))=0,\;i\in N,j\in M_i,a\in A(I^j_i)\backslash\{a^0_{I^j_i}\},\\

 (1-t)\tilde v^i_{I^j_i}(a)\sum\limits_{a'\in A(I^j_i)}\tilde v^i_{I^j_i}(a')(u^i(a,\varrho^i_{I^j_i}(\varpi(\beta^{-I^j_i}(v),t),\tilde\beta(\tilde v))| I^j_i)\\
 \hspace{3cm}-u^i(a',\varrho^i_{I^j_i}(\varpi(\beta^{-I^j_i}(v),t),\tilde\beta(\tilde v))| I^j_i))
+t\sum\limits_{a'\in A(I^j_i)}(\tilde v^i_{I^j_i}(a')-\tilde v^i_{I^j_i}(a)) \\
\hspace{4cm}-t\tilde v^i_{I^j_i}(a)\sum\limits_{a'\in A(I^j_i)}\tilde v^i_{I^j_i}(a')\ln(\tilde\beta^{0i}_{I^j_i}(a')/\tilde\beta^{0i}_{I^j_i}(a))=0,\;i\in N,j\in M_i,a\in A(I^j_i)\backslash\{a^0_{I^j_i}\},\\

 \sum\limits_{a\in A(I^j_i)}\beta^i_{I^j_i}(v; a)=1,\;\sum\limits_{a\in A(I^j_i)}\tilde\beta^i_{I^j_i}(\tilde v; a)=1,\;i\in N, j\in M_i;\;

 0<v^i_{I^j_i},\;0<\tilde v^i_{I^j_i},\;i\in N, j\in M_i, a\in A(I^j_i).
\end{array}
\end{equation}}

\noindent
 When $t=1$, the system~(\ref{entbes5wsre}) has a unique solution given by $(v^*(1), \tilde v^*(1))$ with $v^{*i}_{I^j_i}(1;a)=(1-\ln\beta^{0i}_{I^j_i}(a))^{-1}$ and $\tilde v^{*i}_{I^j_i}(1;a)=(1-\ln\tilde\beta^{0i}_{I^j_i}(a))^{-1}$.

After subtracting a perturbation term from the system~(\ref{entbes5wsre}),  one can show as in the previous subsection that the perturbed system determines a unique smooth path that starts from $(v^*(1),\tilde v^*(1), 1)$ and approaches a WSRE.
One can adapt a standard predictor-corrector method in Eaves and Schmedders~\cite{Eaves and Schmedders (1999)} for numerically tracing the smooth path specified by the system~(\ref{entbes5wsre}) to find a WSRE. The smooth paths specified by the system~(\ref{entbes5wsre}) for the games in Figs.~\ref{TFigure1}-\ref{TFigure2} are illustrated in Figs.~\ref{TFigure1MEntB1}-\ref{TFigure2MEntB2}.

\begin{figure}[H]
    \centering
    \begin{minipage}{0.49\textwidth}
        \centering
        \includegraphics[width=0.80\textwidth, height=0.15\textheight]{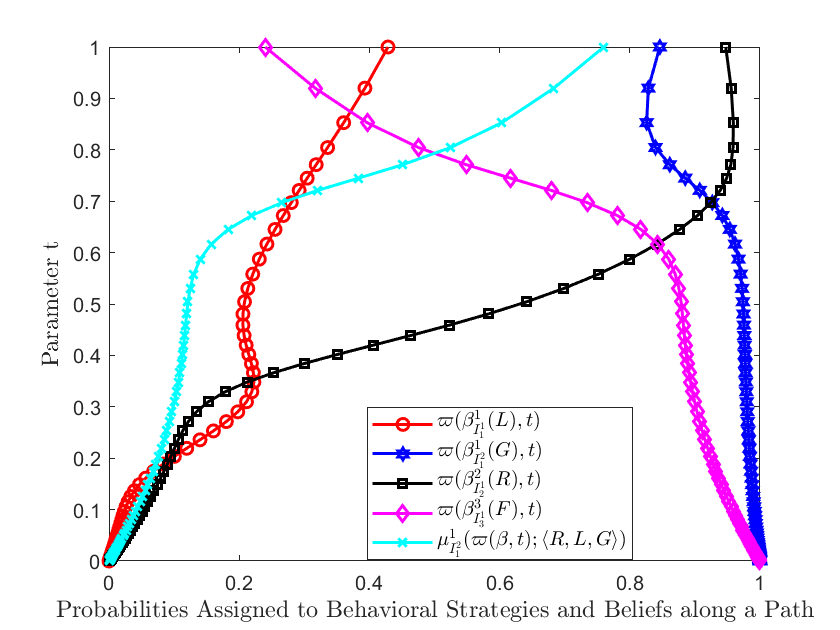}
        % first figure itself
               \caption{\label{TFigure1MEntB1}\scriptsize The Smooth Path of $\varpi(\beta,t)$ Specified by the System~(\ref{entbes5wsre}) for the Game in Fig.~\ref{TFigure1}}
\end{minipage}\hfill
    \begin{minipage}{0.49\textwidth}
        \centering
        \includegraphics[width=0.80\textwidth, height=0.15\textheight]{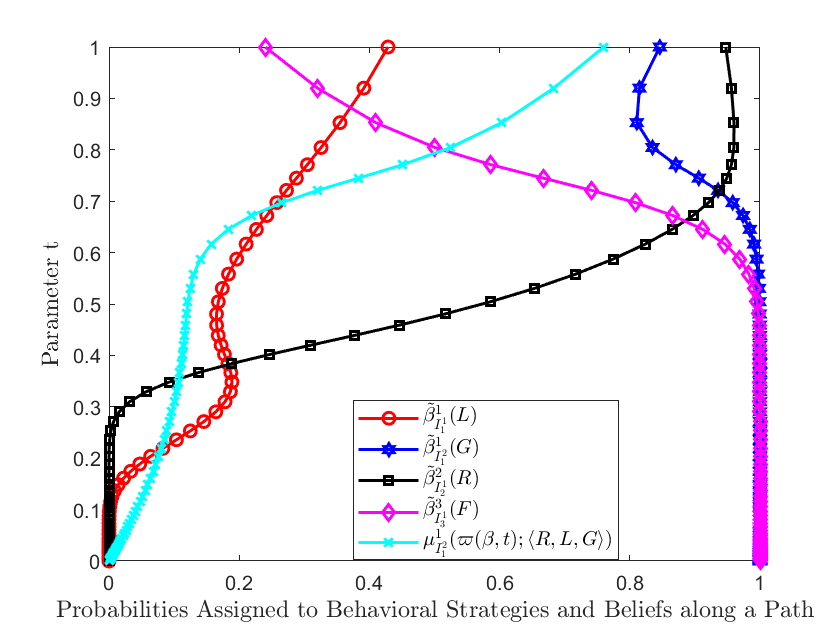}
        % first figure itself
        \caption{\label{TFigure1MEntB2}\scriptsize The Smooth Path of $\tilde\beta$ Specified by the System~(\ref{entbes5wsre}) for the Game in Fig.~\ref{TFigure1}}
\end{minipage}
 \end{figure}

\begin{figure}[H]
    \centering
    \begin{minipage}{0.49\textwidth}
        \centering
        \includegraphics[width=0.80\textwidth, height=0.15\textheight]{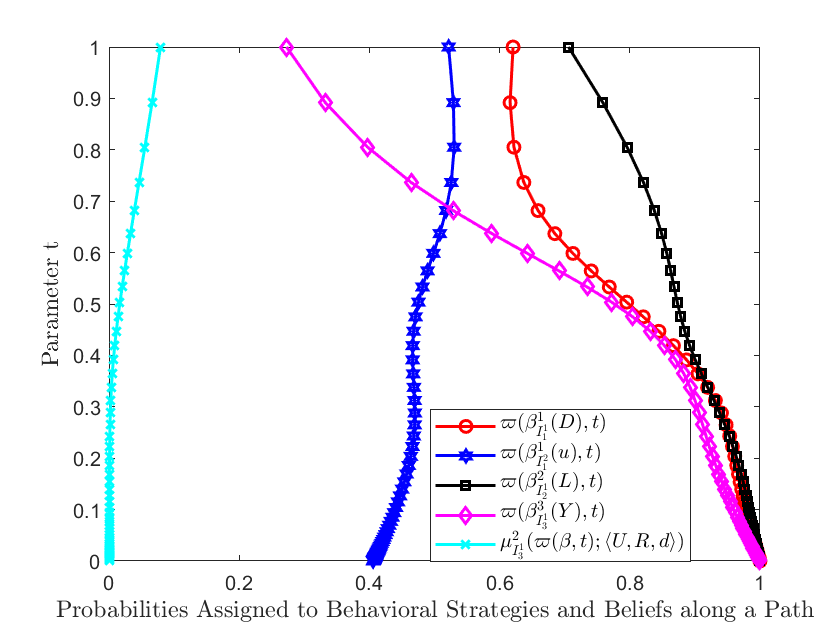}
        % first figure itself
               \caption{\label{TFigure2MEntB1}\scriptsize The Smooth Path of $\varpi(\beta,t)$ Specified by the System~(\ref{entbes5wsre}) for the Game in Fig.~\ref{TFigure2}}
\end{minipage}\hfill
    \begin{minipage}{0.49\textwidth}
        \centering
        \includegraphics[width=0.80\textwidth, height=0.15\textheight]{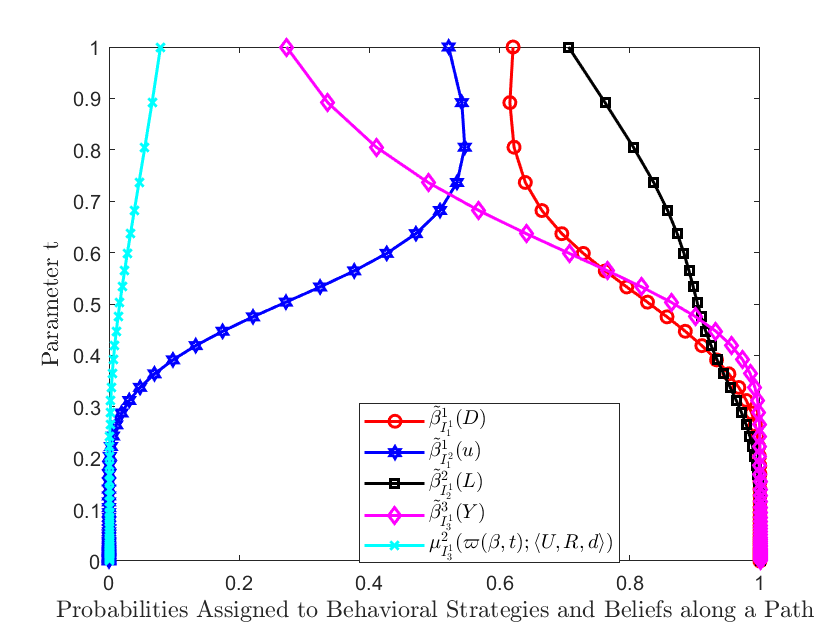}
        % first figure itself
        \caption{\label{TFigure2MEntB2}\scriptsize The Smooth Path of $\tilde\beta$ Specified by the System~(\ref{entbes5wsre}) for the Game in Fig.~\ref{TFigure2}}
\end{minipage}
 \end{figure}

\subsection{A Square-Root-Barrier Smooth Path to a WSRE}

As an alternative scheme, for $t\in (0,1]$, we constitute with $\varpi(\beta, t)$ a square-root-barrier extensive-form game $\Gamma_S(t)$ in which player $i$ at information set $I^j_i$ solves against a given pair  $(\hat\beta,\hat{\tilde\beta})$ the strictly convex optimization problem, {\footnotesize
\begin{equation}\setlength{\abovedisplayskip}{0pt}\setlength{\belowdisplayskip}{0pt}\label{srtbop1wsre}
\begin{array}{rl}
\max\limits_{\beta^i_{I^j_i},\;\tilde\beta^i_{I^j_i}} & (1-t)\sum\limits_{a\in A(I^j_i)}(\beta^i_{I^j_i}(a){\cal D}(I^j_i|\varpi(\hat\beta,t)) +\tilde\beta^i_{I^j_i}(a))u^i(a,\varrho^i_{I^j_i}(\varpi(\hat\beta^{-I^j_i},t),\hat{\tilde\beta})| I^j_i)\\
& 
+t\sum\limits_{a\in A(I^j_i)}(\sqrt{\beta^{0i}_{I^j_i}(a)\beta^i_{I^j_i}(a)} +\sqrt{\tilde\beta^{0i}_{I^j_i}(a)\tilde\beta^i_{I^j_i}(a)})\\
\text{s.t.} & \sum\limits_{a\in A(I^j_i)}\beta^i_{I^j_i}(a)=1,\;\sum\limits_{a\in A(I^j_i)}\tilde\beta^i_{I^j_i}(a)=1.
\end{array}
\end{equation}}Applying the optimality conditions to the problem~(\ref{srtbop1wsre}), we acquire from the equilibrium condition of $(\hat\beta,\hat{\tilde\beta})=(\beta,\tilde\beta)$ the equilibrium system of $\Gamma_S(t)$, {\footnotesize
\begin{equation}\setlength{\abovedisplayskip}{0pt}\setlength{\belowdisplayskip}{0pt}\label{srtbes1wsre}
\begin{array}{l}
 (1-t){\cal D}(I^j_i|\varpi(\beta,t)) u^i(a,\varrho^i_{I^j_i}(\varpi(\beta^{-I^j_i},t),\tilde\beta)| I^j_i)
+t\sqrt{\beta^{0i}_{I^j_i}(a)}/\sqrt{\beta^i_{I^j_i}(a)}-\zeta^i_{I^j_i}=0,\;i\in N,j\in M_i,a\in A(I^j_i),\\

(1-t)u^i(a,\varrho^i_{I^j_i}(\varpi(\beta^{-I^j_i},t),\tilde\beta)| I^j_i)
+t\sqrt{\tilde\beta^{0i}_{I^j_i}(a)}/\sqrt{\tilde\beta^i_{I^j_i}(a)}-\tilde\zeta^i_{I^j_i}=0,\;
i\in N,j\in M_i,a\in A(I^j_i),\\

 \sum\limits_{a\in A(I^j_i)}\beta^i_{I^j_i}(a)=1,\;\sum\limits_{a\in A(I^j_i)}\tilde\beta^i_{I^j_i}(a)=1,\;i\in N, j\in M_i;\;

 0<\beta^i_{I^j_i},\;0<\tilde\beta^i_{I^j_i},\;i\in N, j\in M_i, a\in A(I^j_i).
\end{array}
\end{equation}}

\noindent We define $c(x)=x^2$. 
Let $v=(v^i_{I^j_i}(a):i\in N, j\in M_i,a\in A(I^j_i))$ and $\tilde v=(\tilde v^i_{I^j_i}(a):i\in N, j\in M_i,a\in A(I^j_i))$. Substituting $\beta^i_{I^j_i}(v; a)=c(v^i_{I^j_i}(a))$ and $\tilde\beta^i_{I^j_i}(
\tilde v; a)=c(\tilde v^i_{I^j_i}(a))$ into the system~(\ref{srtbes1wsre}) for $\beta^i_{I^j_i}(a)$ and $\tilde\beta^i_{I^j_i}(a)$, we come to the system, {\footnotesize
\begin{equation}\setlength{\abovedisplayskip}{0pt}\setlength{\belowdisplayskip}{0pt}\label{srtbes1Awsre}
\begin{array}{l}
 (1-t){\cal D}(I^j_i|\varpi(\beta(v),t)) u^i(a,\varrho^i_{I^j_i}(\varpi(\beta^{-I^j_i}(v),t),\tilde\beta(\tilde v))| I^j_i)
+t\sqrt{\beta^{0i}_{I^j_i}(a)}/v^i_{I^j_i}(a)-\zeta^i_{I^j_i}=0,\\
\hspace{12cm}i\in N,j\in M_i,a\in A(I^j_i),\\

(1-t)u^i(a,\varrho^i_{I^j_i}(\varpi(\beta^{-I^j_i}(v),t),\tilde\beta(\tilde v))| I^j_i)
+t\sqrt{\tilde\beta^{0i}_{I^j_i}(a)}/\tilde v^i_{I^j_i}(a)-\tilde\zeta^i_{I^j_i}=0,\;
i\in N,j\in M_i,a\in A(I^j_i),\\

 \sum\limits_{a\in A(I^j_i)}\beta^i_{I^j_i}(v; a)=1,\;\sum\limits_{a\in A(I^j_i)}\tilde\beta^i_{I^j_i}(\tilde v; a)=1,\;i\in N, j\in M_i;\;

 0<\beta^i_{I^j_i},\;0<\tilde\beta^i_{I^j_i},\;i\in N, j\in M_i, a\in A(I^j_i).
\end{array}
\end{equation}}

\noindent 
Multiplying $v^i_{I^j_i}(a)$ and $\tilde v^i_{I^j_i}(a)$ to equations in the first and second groups of the system~(\ref{srtbes1Awsre}), respectively, we come to the system, {\footnotesize
\begin{equation}\setlength{\abovedisplayskip}{0pt}\setlength{\belowdisplayskip}{0pt}\label{srtbes2wsre}
\begin{array}{l}
 (1-t){\cal D}(I^j_i|\varpi(\beta,t)) v^i_{I^j_i}(a)u^i(a,\varrho^i_{I^j_i}(\varpi(\beta^{-I^j_i}(v),t),\tilde\beta(\tilde v))| I^j_i)
+t\sqrt{\beta^{0i}_{I^j_i}(a)}-v^i_{I^j_i}(a)\zeta^i_{I^j_i}=0,\\
\hspace{11.6cm}\;i\in N,j\in M_i,a\in A(I^j_i),\\

(1-t)\tilde v^i_{I^j_i}(a)u^i(a,\varrho^i_{I^j_i}(\varpi(\beta^{-I^j_i}(v),t),\tilde\beta(\tilde v))| I^j_i)
+t\sqrt{\tilde\beta^{0i}_{I^j_i}(a)}-\tilde v^i_{I^j_i}(a)\tilde\zeta^i_{I^j_i}=0,\;
i\in N,j\in M_i,a\in A(I^j_i),\\

 \sum\limits_{a\in A(I^j_i)}\beta^i_{I^j_i}(v; a)=1,\;\sum\limits_{a\in A(I^j_i)}\tilde\beta^i_{I^j_i}(\tilde v; a)=1,\;i\in N, j\in M_i;\;
 
 0<\beta^i_{I^j_i},\;0<\tilde\beta^i_{I^j_i},\;i\in N, j\in M_i, a\in A(I^j_i).
\end{array}
\end{equation}}

\noindent
Taking the sum of equations in the first group and the sum of equations in the second group
of the system~(\ref{srtbes2wsre}) over $A(I^j_i)$, respectively, we get from a division operation the system, {\footnotesize
\begin{equation}\setlength{\abovedisplayskip}{0pt}\setlength{\belowdisplayskip}{0pt}\label{srtbes3wsre}
\begin{array}{l}
\zeta^i_{I^j_i}= \frac{1}{\sum\limits_{a'\in A(I^j_i)} v^i_{I^j_i}(a')}( (1-t){\cal D}(I^j_i|\varpi(\beta(v),t))\sum\limits_{a'\in A(I^j_i)} v^i_{I^j_i}(a')u^i(a',\varrho^i_{I^j_i}(\varpi(\beta^{-I^j_i}(v),t),\tilde\beta(\tilde v))| I^j_i)+t\sum\limits_{a'\in A(I^j_i)}\sqrt{\beta^{0i}_{I^j_i}(a')}),\\

\tilde\zeta^i_{I^j_i}=\frac{1}{\sum\limits_{a'\in A(I^j_i)}\tilde\beta^i_{I^j_i}(a')}((1-t)\sum\limits_{a'\in A(I^j_i)}\tilde\beta^i_{I^j_i}(a')u^i(a',\varrho^i_{I^j_i}(\varpi(\beta^{-I^j_i},t),\tilde\beta)| I^j_i)
+t\sum\limits_{a'\in A(I^j_i)}\sqrt{\tilde\beta^{0i}_{I^j_i}(a')}),\;i\in N,j\in M_i.
\end{array}
\end{equation}}We denote by $a^0_{I^j_i}$ a given reference action in $A(I^j_i)$.
Substituting $\zeta^i_{I^j_i}$ and $\tilde\zeta^i_{I^j_i}$ of the system~(\ref{srtbes3wsre}) into the system~(\ref{srtbes2wsre}), we arrive at the system, {\footnotesize
\begin{equation}\setlength{\abovedisplayskip}{0pt}\setlength{\belowdisplayskip}{0pt}\label{srtbes4wsre}
\begin{array}{l}
 (1-t){\cal D}(I^j_i|\varpi(\beta(v),t))v^i_{I^j_i}(a)\sum\limits_{a'\in A(I^j_i)}v^i_{I^j_i}(a') (u^i(a,\varrho^i_{I^j_i}(\varpi(\beta^{-I^j_i}(v),t),\tilde\beta(\tilde v))| I^j_i)\\
\hspace{2cm} -u^i(a',\varrho^i_{I^j_i}(\varpi(\beta^{-I^j_i}(v),t),\tilde\beta(\tilde v))| I^j_i))
 
+t\sum\limits_{a'\in A(I^j_i)}(v^i_{I^j_i}(a')\sqrt{\beta^{0i}_{I^j_i}(a)}-v^i_{I^j_i}(a)\sqrt{\beta^{0i}_{I^j_i}(a')})=0,\\

\hspace{11cm} i\in N,j\in M_i,a\in A(I^j_i)\backslash\{a^0_{I^j_i}\},\\

(1-t)\tilde v^i_{I^j_i}(a)\sum\limits_{a'\in A(I^j_i)}\tilde v^i_{I^j_i}(a') (u^i(a,\varrho^i_{I^j_i}(\varpi(\beta^{-I^j_i}(v),t),\tilde\beta(\tilde v))| I^j_i)-u^i(a',\varrho^i_{I^j_i}(\varpi(\beta^{-I^j_i}(v),t),\tilde\beta(\tilde v))| I^j_i))\\
\hspace{3cm}+t\sum\limits_{a'\in A(I^j_i)}(\tilde v^i_{I^j_i}(a')\sqrt{\tilde\beta^{0i}_{I^j_i}(a)}-\tilde v^i_{I^j_i}(a)\sqrt{\tilde\beta^{0i}_{I^j_i}(a')})=0,\;i\in N,j\in M_i,a\in A(I^j_i)\backslash\{a^0_{I^j_i}\},\\

 \sum\limits_{a\in A(I^j_i)}\beta^i_{I^j_i}(v; a)=1,\;\sum\limits_{a\in A(I^j_i)}\tilde\beta^i_{I^j_i}(\tilde v; a)=1,\;i\in N, j\in M_i;\;
 
 0<\beta^i_{I^j_i},\;0<\tilde\beta^i_{I^j_i},\;i\in N, j\in M_i, a\in A(I^j_i).
\end{array}
\end{equation}}

\noindent When $t=1$, the system~(\ref{srtbes4wsre}) has a unique solution given by $(v^*(1),\tilde v^*(1))$ with $v^{*i}_{I^j_i}(a)=\sqrt{\beta^{0i}_{I^j_i}(a)}$ and $\tilde v^{*i}_{I^j_i}(a)=\sqrt{\tilde\beta^{0i}_{I^j_i}(a)}$.

After subtracting a perturbation term from the system~(\ref{srtbes4wsre}),  one can show as in the previous subsection that the perturbed system determines a unique smooth path that starts from $(v^*(1),\tilde v^*(1), 1)$ and approaches a WSRE.
One can adapt a standard predictor-corrector method in Eaves and Schmedders~\cite{Eaves and Schmedders (1999)} for numerically tracing the smooth path specified by the system~(\ref{srtbes4wsre}) to find a WSRE. The smooth paths specified by the system~(\ref{srtbes4wsre}) for the games in Figs.~\ref{TFigure1}-\ref{TFigure2} are illustrated in Figs.~\ref{TFigure1MSrtB1}-\ref{TFigure2MSrtB2}.

\begin{figure}[H]
    \centering
    \begin{minipage}{0.49\textwidth}
        \centering
        \includegraphics[width=0.80\textwidth, height=0.15\textheight]{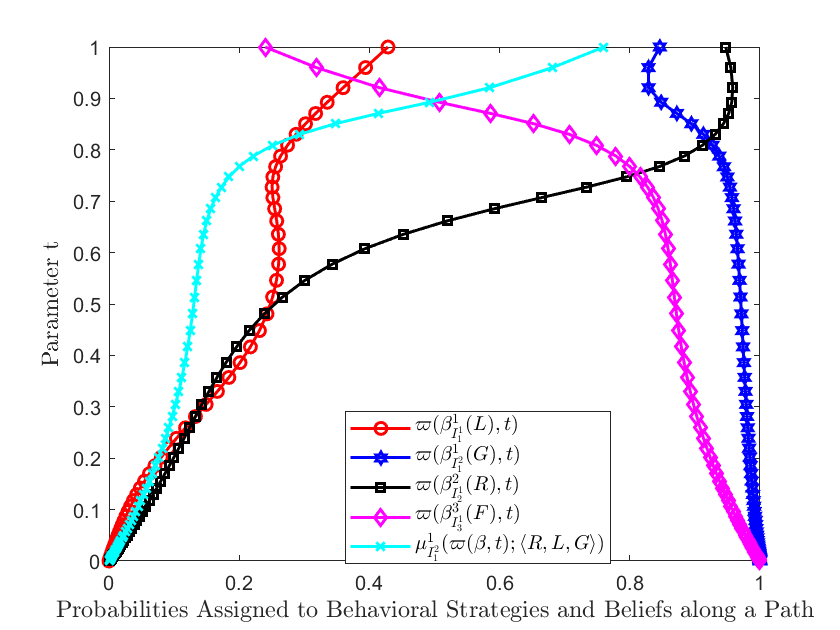}
        % first figure itself
               \caption{\label{TFigure1MSrtB1}\scriptsize The Smooth Path of $\varpi(\beta,t)$ Specified by the System~(\ref{srtbes4wsre}) for the Game in Fig.~\ref{TFigure1}}
\end{minipage}\hfill
    \begin{minipage}{0.49\textwidth}
        \centering
        \includegraphics[width=0.80\textwidth, height=0.15\textheight]{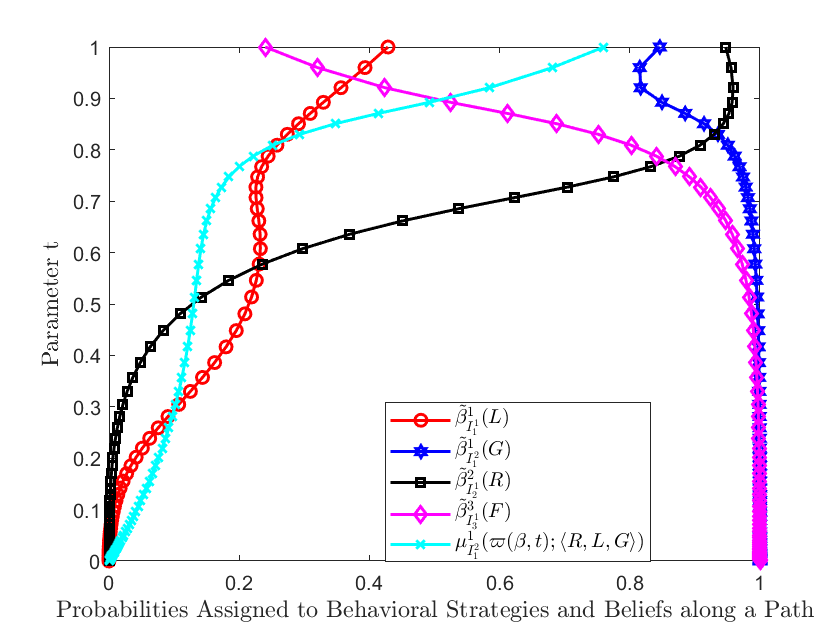}
        % first figure itself
        \caption{\label{TFigure1MSrtB2}\scriptsize The Smooth Path of $\tilde\beta$ Specified by the System~(\ref{srtbes4wsre}) for the Game in Fig.~\ref{TFigure1}}
\end{minipage}
 \end{figure}

\begin{figure}[H]
    \centering
    \begin{minipage}{0.49\textwidth}
        \centering
        \includegraphics[width=0.80\textwidth, height=0.15\textheight]{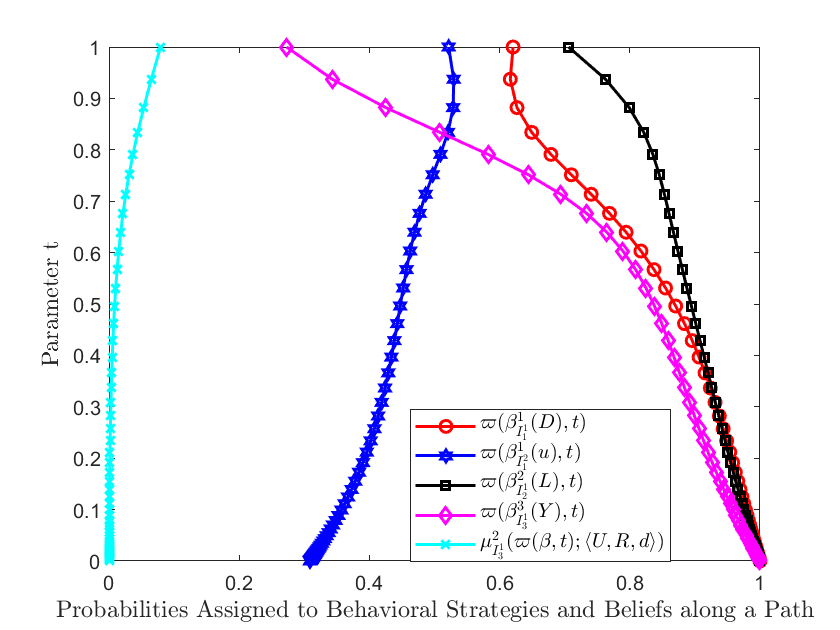}
        % first figure itself
               \caption{\label{TFigure2MSrtB1}\scriptsize The Smooth Path of $\varpi(\beta,t)$ Specified by the System~(\ref{srtbes4wsre}) for the Game in Fig.~\ref{TFigure2}}
\end{minipage}\hfill
    \begin{minipage}{0.49\textwidth}
        \centering
        \includegraphics[width=0.80\textwidth, height=0.15\textheight]{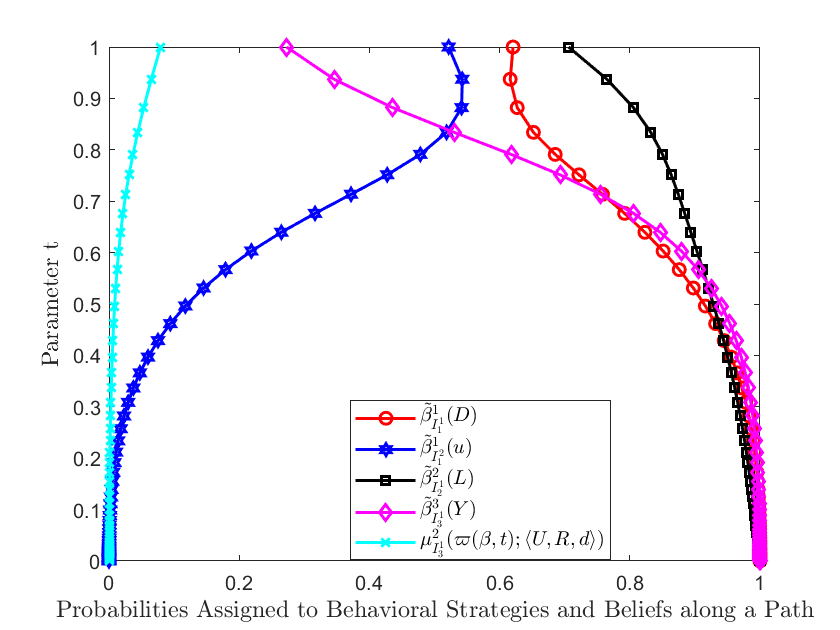}
        % first figure itself
        \caption{\label{TFigure2MSrtB2}\scriptsize The Smooth Path of $\tilde\beta$ Specified by the System~(\ref{srtbes4wsre}) for the Game in Fig.~\ref{TFigure2}}
\end{minipage}
 \end{figure}

\subsection{A Convex-Quadratic-Penalty Smooth Path to a WSRE\label{cqupwsre}}
 
For  $t\in (0,1]$, we constitute with $\varpi(\beta,t)$ a convex-quadratic-penalty extensive-form game $\Gamma_Q(t)$ in which player $i$ at information set $I^j_i$ solves against a given pair  $(\hat\beta,\hat{\tilde\beta})$ the strictly convex optimization problem, {\footnotesize
\begin{equation}\setlength{\abovedisplayskip}{0pt}\setlength{\belowdisplayskip}{0pt}\label{cqupop1wsre}
\begin{array}{rl}
\max\limits_{\beta^i_{I^j_i},\;\tilde\beta^i_{I^j_i}} & (1-t)\sum\limits_{a\in A(I^j_i)}(\beta^i_{I^j_i}(a){\cal D}(I^j_i|\varpi(\hat\beta,t)) +\tilde\beta^i_{I^j_i}(a))u^i(a,\varrho^i_{I^j_i}(\varpi(\hat\beta^{-I^j_i},t),\hat{\tilde\beta})| I^j_i)\\
& 
-\frac{1}{2}t\sum\limits_{a\in A(I^j_i)}((\beta^i_{I^j_i}(a)-\beta^{0i}_{I^j_i}(a))^2 +(\tilde\beta^i_{I^j_i}(a)-\tilde\beta^{0i}_{I^j_i}(a))^2)\\
\text{s.t.} & \sum\limits_{a\in A(I^j_i)}\beta^i_{I^j_i}(a)=1,\;\sum\limits_{a\in A(I^j_i)}\tilde\beta^i_{I^j_i}(a)=1,\;0\le\beta^i_{I^j_i}(a),\;0\le\tilde\beta^i_{I^j_i}(a),\;a\in A(I^j_i).
\end{array}
\end{equation}}Applying the optimality conditions to the problem~(\ref{cqupop1wsre}), we acquire from the equilibrium condition of $(\hat\beta,\hat{\tilde\beta})=(\beta,\tilde\beta)$ the equilibrium system of $\Gamma_Q(t)$, {\footnotesize
\begin{equation}\setlength{\abovedisplayskip}{0pt}\setlength{\belowdisplayskip}{0pt}\label{cqupes1wsre}
\begin{array}{l}
 (1-t){\cal D}(I^j_i|\varpi(\beta,t)) u^i(a,\varrho^i_{I^j_i}(\varpi(\beta^{-I^j_i},t),\tilde\beta)| I^j_i)+\lambda^i_{I^j_i}(a)
-t(\beta^i_{I^j_i}(a)-\beta^{0i}_{I^j_i}(a))-\zeta^i_{I^j_i}=0,\\
\hspace{11.6cm}i\in N,j\in M_i,a\in A(I^j_i),\\

(1-t)u^i(a,\varrho^i_{I^j_i}(\varpi(\beta^{-I^j_i},t),\tilde\beta)| I^j_i)+\tilde\lambda^i_{I^j_i}(a)
-t(\tilde\beta^i_{I^j_i}(a)-\tilde\beta^{0i}_{I^j_i}(a))-\tilde\zeta^i_{I^j_i}=0,\;i\in N,j\in M_i,a\in A(I^j_i),\\

 \sum\limits_{a\in A(I^j_i)}\beta^i_{I^j_i}(a)=1,\;\sum\limits_{a\in A(I^j_i)}\tilde\beta^i_{I^j_i}(a)=1,\;i\in N,j\in M_i,\\
 
\beta^i_{I^j_i}(a)\lambda^i_{I^j_i}(a) =0,\;\tilde\beta^i_{I^j_i}(a)\tilde\lambda^i_{I^j_i}(a) =0,\;0\le\beta^i_{I^j_i}(a),\;0\le\lambda^i_{I^j_i}(a),\;0\le\tilde\beta^i_{I^j_i}(a),\;0\le\tilde\lambda^i_{I^j_i}(a),\\
\hspace{11.6cm}i\in N,j\in M_i,a\in A(I^j_i).
\end{array}
\end{equation}}Given a reference action $a^0_{I^j_i}\in A(I^j_i)$, 
we obtain from simplification and subtraction to the system~(\ref{cqupes1wsre})   the system, {\footnotesize
\begin{equation}\setlength{\abovedisplayskip}{0pt}\setlength{\belowdisplayskip}{0pt}\label{cqupes2wsre}
\begin{array}{l}
 (1-t){\cal D}(I^j_i|\varpi(\beta,t))( u^i(a,\varrho^i_{I^j_i}(\varpi(\beta^{-I^j_i},t),\tilde\beta)| I^j_i)-u^i(a^0_{I^j_i},\varrho^i_{I^j_i}(\varpi(\beta^{-I^j_i},t),\tilde\beta)| I^j_i))+\lambda^i_{I^j_i}(a)-\lambda^i_{I^j_i}(a^0_{I^j_i})\\
 \hspace{3.2cm}
-t(\beta^i_{I^j_i}(a)-\beta^i_{I^j_i}(a^0_{I^j_i})-(\beta^{0i}_{I^j_i}(a)-\beta^{0i}_{I^j_i}(a^0_{I^j_i})))=0,\;i\in N,j\in M_i,a\in A(I^j_i)\backslash\{a^0_{I^j_i}\},\\

(1-t)(u^i(a,\varrho^i_{I^j_i}(\varpi(\beta^{-I^j_i},t),\tilde\beta)| I^j_i)-u^i(a^0_{I^j_i},\varrho^i_{I^j_i}(\varpi(\beta^{-I^j_i},t),\tilde\beta)| I^j_i))+\tilde\lambda^i_{I^j_i}(a)-\tilde\lambda^i_{I^j_i}(a^0_{I^j_i})\\

\hspace{1.6cm}-t(\tilde\beta^i_{I^j_i}(a)-\tilde\beta^i_{I^j_i}(a^0_{I^j_i})-(\tilde\beta^{0i}_{I^j_i}(a)-\tilde\beta^{0i}_{I^j_i}(a^0_{I^j_i})))=0,\;
i\in N,j\in M_i,a\in A(I^j_i)\backslash\{a^0_{I^j_i}\},\\

 \sum\limits_{a\in A(I^j_i)}\beta^i_{I^j_i}(a)=1,\;\sum\limits_{a\in A(I^j_i)}\tilde\beta^i_{I^j_i}(a)=1,\;i\in N,j\in M_i,\\
 
\beta^i_{I^j_i}(a)\lambda^i_{I^j_i}(a) =0,\;\tilde\beta^i_{I^j_i}(a)\tilde\lambda^i_{I^j_i}(a) =0,\;0\le\beta^i_{I^j_i}(a),\;0\le\lambda^i_{I^j_i}(a),\;0\le\tilde\beta^i_{I^j_i}(a),\;0\le\tilde\lambda^i_{I^j_i}(a),\\
\hspace{11.6cm}i\in N,j\in M_i,a\in A(I^j_i).
\end{array}
\end{equation}}A multiplication of $\omega(I^j_i|\varpi(\beta,t))$ to equations in the first and second groups of the system~(\ref{cqupes2wsre}) yields the system, {\footnotesize
\begin{equation}\setlength{\abovedisplayskip}{0pt}\setlength{\belowdisplayskip}{0pt}\label{cqupes3wsre}
\begin{array}{l}
 (1-t){\cal D}(I^j_i|\varpi(\beta,t))(u^i((a,\varrho^i_{I^j_i}(\varpi(\beta^{-I^j_i},t),\tilde\beta))\land I^j_i)-u^i((a^0_{I^j_i},\varrho^i_{I^j_i}(\varpi(\beta^{-I^j_i},t),\tilde\beta))\land I^j_i))\\
 \hspace{0.6cm}+\omega(I^j_i|\varpi(\beta, t))(\lambda^i_{I^j_i}(a)-\lambda^i_{I^j_i}(a^0_{I^j_i})
-t(\beta^i_{I^j_i}(a)-\beta^i_{I^j_i}(a^0_{I^j_i})-(\beta^{0i}_{I^j_i}(a)-\beta^{0i}_{I^j_i}(a^0_{I^j_i}))))=0,\\
\hspace{10.2cm}i\in N,j\in M_i,a\in A(I^j_i)\backslash\{a^0_{I^j_i}\},\\

(1-t)(u^i((a,\varrho^i_{I^j_i}(\varpi(\beta^{-I^j_i},t),\tilde\beta))\land I^j_i)-u^i((a^0_{I^j_i},\varrho^i_{I^j_i}(\varpi(\beta^{-I^j_i},t),\tilde\beta))\land I^j_i))\\

\hspace{0.6cm}+\omega(I^j_i|\varpi(\beta, t))(\tilde\lambda^i_{I^j_i}(a)-\tilde\lambda^i_{I^j_i}(a^0_{I^j_i})-t(\tilde\beta^i_{I^j_i}(a)-\tilde\beta^i_{I^j_i}(a^0_{I^j_i})-(\tilde\beta^{0i}_{I^j_i}(a)-\tilde\beta^{0i}_{I^j_i}(a^0_{I^j_i}))))=0,\\

\hspace{10.2cm}
i\in N,j\in M_i,a\in A(I^j_i)\backslash\{a^0_{I^j_i}\},\\

 \sum\limits_{a\in A(I^j_i)}\beta^i_{I^j_i}(a)=1,\;\sum\limits_{a\in A(I^j_i)}\tilde\beta^i_{I^j_i}(a)=1,\;i\in N,j\in M_i,\\
 
\beta^i_{I^j_i}(a)\lambda^i_{I^j_i}(a) =0,\;\tilde\beta^i_{I^j_i}(a)\tilde\lambda^i_{I^j_i}(a) =0,\;0\le\beta^i_{I^j_i}(a),\;0\le\lambda^i_{I^j_i}(a),\;0\le\tilde\beta^i_{I^j_i}(a),\;0\le\tilde\lambda^i_{I^j_i}(a),\\
\hspace{11.6cm}i\in N,j\in M_i,a\in A(I^j_i).
\end{array}
\end{equation}}Let $\phi_1(z)=(\frac{z+\sqrt{z^2}}{2})^2$ and $\phi_2(z)=(\frac{z-\sqrt{z^2}}{2})^2$. Then, $\phi_1(z)\phi_2(z)=0$. Let $v=(v^i_{I^j_i}(a):i\in N, j\in M_i, a\in A(I^j_i))$ and $\tilde v=(\tilde v^i_{I^j_i}(a):i\in N, j\in M_i, a\in A(I^j_i))$. Let $\beta^i_{I^j_i}(v;a)=\phi_1(v^i_{I^j_i}(a))$, $\lambda^i_{I^j_i}(v;a)=\phi_2(v^i_{I^j_i}(a))$,  $\tilde\beta^i_{I^j_i}(\tilde v;a)=\phi_1(\tilde z^i_{I^j_i}(a))$,  and $\tilde\lambda^i_{I^j_i}(\tilde v;a)=\phi_2(\tilde z^i_{I^j_i}(a))$. Substituting $\beta^i_{I^j_i}(v;a)$, $\lambda^i_{I^j_i}(v;a)$, $\tilde\beta^i_{I^j_i}(\tilde v;a)$, and $\tilde\lambda^i_{I^j_i}(\tilde v;a)$ into the system~(\ref{cqupes3wsre}) for  $\beta^i_{I^j_i}(a)$, $\lambda^i_{I^j_i}(a)$, $\tilde\beta^i_{I^j_i}(a)$, and $\tilde\lambda^i_{I^j_i}(a)$, we get an equivalent system with much fewer variables, {\footnotesize
\begin{equation}\setlength{\abovedisplayskip}{0pt}\setlength{\belowdisplayskip}{0pt}\label{cqupes4wsre}
\begin{array}{l}
 (1-t){\cal D}(I^j_i|\varpi(\beta(v),t))(u^i((a,\varrho^i_{I^j_i}(\varpi(\beta^{-I^j_i}(v),t),\tilde\beta(\tilde v)))\land I^j_i)\\
 \hspace{0.6cm}-u^i((a^0_{I^j_i},\varrho^i_{I^j_i}(\varpi(\beta^{-I^j_i}(v),t),\tilde\beta(\tilde v)))\land I^j_i))+\omega(I^j_i|\varpi(\beta(v), t))(\lambda^i_{I^j_i}(v;a)-\lambda^i_{I^j_i}(v; a^0_{I^j_i})
\\
\hspace{0.8cm}-t(\beta^i_{I^j_i}(v; a)-\beta^i_{I^j_i}(v; a^0_{I^j_i})-(\beta^{0i}_{I^j_i}(a)-\beta^{0i}_{I^j_i}(a^0_{I^j_i}))))=0,\;i\in N,j\in M_i,a\in A(I^j_i)\backslash\{a^0_{I^j_i}\},\\

(1-t)(u^i((a,\varrho^i_{I^j_i}(\varpi(\beta^{-I^j_i}(v),t),\tilde\beta(\tilde v)))\land I^j_i)-u^i((a^0_{I^j_i},\varrho^i_{I^j_i}(\varpi(\beta^{-I^j_i}(v),t),\tilde\beta(\tilde v)))\land I^j_i))\\

\hspace{3.6cm}+\omega(I^j_i|\varpi(\beta(v), t))(\tilde\lambda^i_{I^j_i}(\tilde v; a)-\tilde\lambda^i_{I^j_i}(\tilde v; a^0_{I^j_i})-t(\tilde\beta^i_{I^j_i}(\tilde v; a)-\tilde\beta^i_{I^j_i}(\tilde v; a^0_{I^j_i})\\

\hspace{5.2cm}
-(\tilde\beta^{0i}_{I^j_i}(a)-\tilde\beta^{0i}_{I^j_i}(a^0_{I^j_i}))))=0,\;i\in N,j\in M_i,a\in A(I^j_i)\backslash\{a^0_{I^j_i}\},\\

 \sum\limits_{a\in A(I^j_i)}\beta^i_{I^j_i}(v; a)=1,\;\sum\limits_{a\in A(I^j_i)}\tilde\beta^i_{I^j_i}(\tilde v; a)=1,\;i\in N,j\in M_i.
\end{array}
\end{equation}}When $t=1$, the system~(\ref{cqupes4wsre}) has a unique solution given by $(v^*(1),\tilde v^*(1))$ with $v^{*i}_{I^j_i}(1; a)=\sqrt{\beta^{0i}_{I^j_i}(1; a)}$ and $\tilde v^{*i}_{I^j_i}(1; a)=\sqrt{\tilde\beta^{0i}_{I^j_i}(1; a)}$. After subtracting a perturbation term from the system~(\ref{cqupes4wsre}),  one can show as in the previous subsection that the perturbed system determines a unique smooth path that starts from $(v^*(1),\tilde v^*(1), 1)$ and approaches a WSRE.
One can adapt a standard predictor-corrector method in Eaves and Schmedders~\cite{Eaves and Schmedders (1999)} for numerically tracing the smooth path specified by the system~(\ref{cqupes4wsre}) to a WSRE. The smooth paths specified by the system~(\ref{cqupes4wsre}) for the games in Figs.~\ref{TFigure1}-\ref{TFigure2} are illustrated in Figs.~\ref{TFigure1MCquP1}-\ref{TFigure2MCquP2}.

\begin{figure}[H]
    \centering
    \begin{minipage}{0.49\textwidth}
        \centering
        \includegraphics[width=0.80\textwidth, height=0.15\textheight]{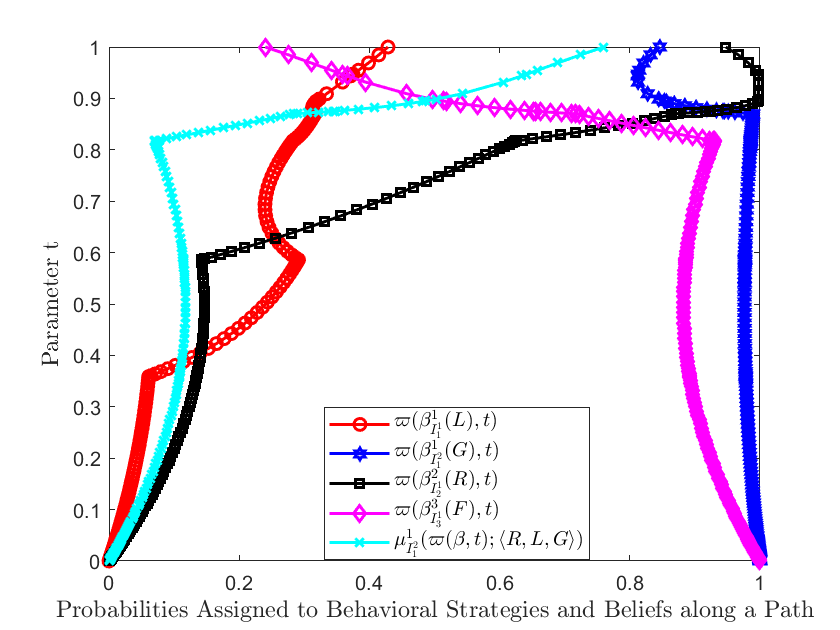}
        % first figure itself
               \caption{\label{TFigure1MCquP1}\scriptsize The Smooth Path of $\varpi(\beta,t)$ Specified by the System~(\ref{cqupes4wsre}) for the Game in Fig.~\ref{TFigure1}}
\end{minipage}\hfill
    \begin{minipage}{0.49\textwidth}
        \centering
        \includegraphics[width=0.80\textwidth, height=0.15\textheight]{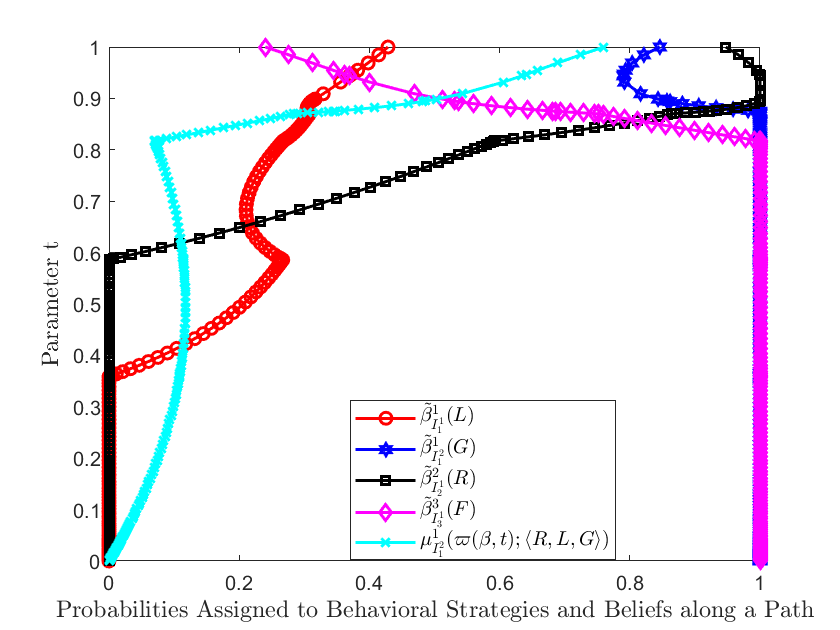}
        % first figure itself
        \caption{\label{TFigure1MCquP2}\scriptsize The Smooth Path of $\tilde\beta$ Specified by the System~(\ref{cqupes4wsre}) for the Game in Fig.~\ref{TFigure1}}
\end{minipage}
 \end{figure}
 
 \begin{figure}[H]
    \centering
    \begin{minipage}{0.49\textwidth}
        \centering
        \includegraphics[width=0.80\textwidth, height=0.15\textheight]{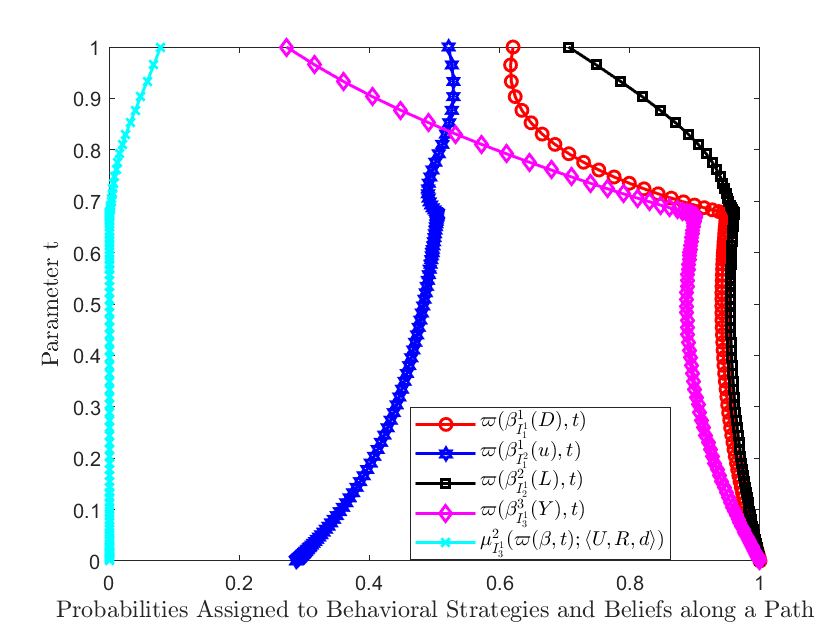}
        % first figure itself
               \caption{\label{TFigure2MCquP1}\scriptsize The Smooth Path of $\varpi(\beta,t)$ Specified by the System~(\ref{cqupes4wsre}) for the Game in Fig.~\ref{TFigure2}}
\end{minipage}\hfill
    \begin{minipage}{0.49\textwidth}
        \centering
        \includegraphics[width=0.80\textwidth, height=0.15\textheight]{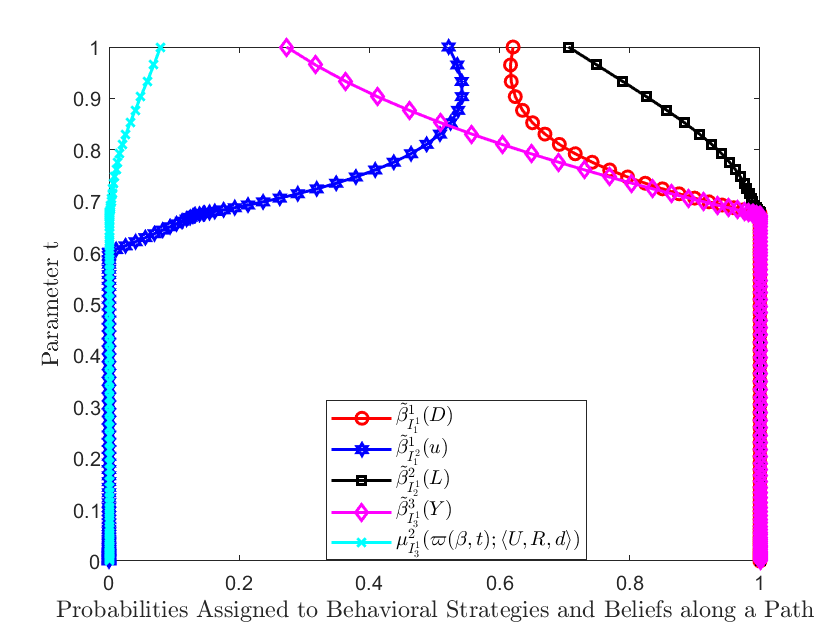}
        % first figure itself
        \caption{\label{TFigure2MCquP2}\scriptsize The Smooth Path of $\tilde\beta$ Specified by the System~(\ref{cqupes4wsre}) for the Game in Fig.~\ref{TFigure2}}
\end{minipage}
 \end{figure}

\section{Numerical Performances}

% the systems~(\ref{logbes5wsre}), (\ref{Slogbes2wsre}), ~(\ref{entbes5wsre}), (\ref{Sentbes2wsre}), (\ref{srtbes4wsre}), (\ref{Ssrtbes2wsre}), and (\ref{cqupes4wsre}),

We have adapted a standard predictor-corrector method as outlined in Eaves and Schmedders~\cite{Eaves and Schmedders (1999)} for numerically tracing the smooth paths to WSREs. The predictor-corrector method has been coded in MATLAB. The parameter values in the method are set as follows: the predictor step size $=$ $0.1*10^{0.2\ln{t}}$ and accuracy of a starting point for a corrector step $=$  $0.1*10^{0.5\ln{t}}$. The method terminates as soon as the criterion of $t<10^{-5}$ is met. The computation was carried out on Windows Server for Intel(R) Xeon(R) Gold 6426Y  2.50GHz  (2 processors) RAM 512GB. 

To demonstrate the efficiency of the smooth paths specified by the systems~(\ref{logbes5wsre}), (\ref{entbes5wsre}),  (\ref{srtbes4wsre}), and (\ref{cqupes4wsre}), we have utilized these methods to calculate WSREs for three categories of randomly generated extensive-form games: Type A, Type B, and Type C. In Type A games, both players $1$ and $2$ possess one information set. For player $i$ with $i\ge 3$, the number of information sets is $\prod_{k=1}^{i-2}|A(I_{k})|$. In Type B games, player $1$ has one information set. The number of information sets for player $i$ with $i\ge 2$ is $|A(I_{i-1})|$. A Type C game shares a similar structure with a Type A game, with the distinction that Type C games consist of multiple layers, denoted as $L$. Within each layer, actions are sequentially taken by players from player $1$ to player $n$. In our numerical experiments, every player has an equal number of actions in all of their information sets. Each payoff along a terminal history is an integer uniformly selected from $-10$ to $10$ and is set to zero with a random probability ranging from $0$ to $50\%$. Within each game defined by a specific choice of $(n,m_i,|A(I_i)|)$ or $(n,m_i,|A(I_i)|, L)$, 10 examples were generated and solved. We denote by LogBM, EntBM, SrtBM, and CquPM the methods with the systems~(\ref{logbes5wsre}), (\ref{entbes5wsre}),  (\ref{srtbes4wsre}), and (\ref{cqupes4wsre}), respectively. 
The efficiency of the methods is evaluated based on the computational time (in seconds) and the number of iterations. Tables~\ref{Table0}-\ref{Table2} summarize the median, minimum, and maximum values for the computational time and the number of iterations. A computation is considered a failure if either the computational time exceeds $4.32\times10^4$ seconds or the number of iterations exceeds $10^5$. In such cases, the result is marked with a dash (``-").

The numerical results in Tables~\ref{Table0}–\ref{Table2} reveal that LogBM consistently outperforms other methods across all game types (A, B, C), achieving the fastest computational times, fewest iterations, and zero failures. While SrtBM serves as a viable intermediate alternative (faster than EntBM but slower than LogBM in Type A/C games), it exhibits instability in large-scale Type B games. EntBM shows stability in all cases but is less efficient than LogBM, whereas CquPM fails frequently in Type B/C games and needs significantly more computational time and iterations, making it unsuitable for complex problems. In conclusion, LogBM is the recommended default choice, with SrtBM as a fallback, and CquPM should be avoided for large-scale applications.

% A computation is deemed successful if the condition $t<10^{-4}$ is met. Conversely, failure occurs when the iteration count or computational time surpasses the predefined threshold, in which case the outcome is recorded as $\infty$ for analytical purposes. However, in tabular representations, we use a dash (`-') to distinguish it from actual infinity.}
	
\begin{table}[H]
\linespread{1} 
\scriptsize
\centering
\caption{{\footnotesize Numerical Performances of LogBM, EntBM, SrtBM, and CquPM for Type A Extensive-Form Games}}\label{Table0}
\begin{tabular*}{\hsize}{@{}@{\extracolsep{\fill}}cc|cccc|cccc@{}}
\hline
\multicolumn{2}{c}{} & \multicolumn{4}{c}{Computational Time} & \multicolumn{4}{c}{Number of Iterations}\\
 $n,m_i,|A(I_i)|$ & & LogBM & EntBM & SrtBM & CquPM & LogBM & EntBM & SrtBM & CquPM\\
\hline
  $3,(1,1,2),(2,10,10)$ & med & 53.06 & 84.58 & 66.93 & 333.14 & 118 & 207 & 157 & 706 \\
 & min & 37.43 & 65.26 & 53.57 & 222.07 & 85 & 158 & 118 & 486 \\
 & max & 86.16 & 115.90 & 93.56 & 636.19 & 186 & 279 & 209 & 1130 \\
 $3,(1,1,2),(2,15,15)$ & med & 175.64 & 205.57 & 154.56 & 693.28 & 180 & 240 & 173 & 841 \\
 & min & 77.16 & 151.92 & 114.03 & 495.06 & 94 & 186 & 137 & 582\\
 & max & 798.77 & 474.93 & 241.52 & 891.77 & 686 & 401 & 274 & 1070 \\
 $3,(1,1,2),(2,20,20)$ & med & 292.22 & 457.49 & 385.93 & 1257.22 & 211 & 362 & 284 & 886 \\
 & min & 143.07 & 299.31 & 211.60 & 993.96 & 101 & 240 & 158 & 777 \\
 & max & 656.72 & 584.95 & 914.24 & 1395.24 & 457 & 469 & 666 & 967 \\
 $3,(1,1,2),(2,25,25)$ & med & 323.55 & 652.28 & 427.48 & 2307.66 & 150 & 324 & 207 & 1086 \\
 & min & 231.54 & 463.64 & 307.45 & 1983.01 & 99 & 231 & 148 & 877 \\
 & max & 462.48 & 1108.45 & 740.93 & 3351.86 & 183 & 400 & 237 & 1590\\
\hline
\end{tabular*}
\end{table}

\begin{table}[H]
\linespread{1} 
\scriptsize
\centering
\caption{{\footnotesize Numerical Performances of LogBM, EntBM, SrtBM, and CquPM for Type B Extensive-Form Games}}\label{Table1}
\begin{tabular*}{\hsize}{@{}@{\extracolsep{\fill}}cc|cccc|cccc@{}}
\hline
\multicolumn{2}{c}{} & \multicolumn{4}{c}{Computational Time} & \multicolumn{4}{c}{Number of Iterations}\\
 $n,m_i,|A(I_i)|$ & & LogBM & EntBM & SrtBM & CquPM & LogBM & EntBM & SrtBM & CquPM\\
\hline
 $4,(1,2,2,5),(2,2,5,3)$ & med & 76.54 & 129.42 & 131.25 & - & 151 & 274 & 232 & - \\
 & min & 42.57 & 65.67 & 62.28 & 244.44 & 120 & 189 & 152 & 542\\
 & max & 116.02 & 375.44 & 295.12 & - & 184 & 511 & 336 & -\\
 $5,(1,2,2,2,5),(2,2,2,5,3)$ & med & 391.05 & 404.36 & 373.67 & - & 279 & 314 & 251 & -\\
 & min & 171.06 & 190.93 & 175.81 & 839.66 & 145 & 163 & 150 & 561\\
 & max & 1191.49 & 863.83 & - & - & 798 & 581 & - & -\\
 $6,(1,2,2,2,2,5),(2,2,2,2,5,3)$ & med & 1515.99 & 1753.26 & 2479.75 & - & 355 & 350 & 546 & -\\
 & min & 744.63 & 740.49 & 705.97 & 2048.01 & 174 & 207 & 175 & 548\\
 & max & 3085.93 & 4377.84 & - & - & 695 & 592 & - & -\\
 $7,(1,2,2,2,2,2,5),(2,2,2,2,2,5,3)$ & med & 4782.82 & 4273.03 & 5109.53 & - & 352 & 345 & 392 & -\\
 & min & 2387.95 & 2503.54 & 2318.14 & 9695.66 & 204 & 209 & 197 & 519\\
 & max & 9709.25 & 8450.95 & - & - & 778 & 652 & - & -\\
\hline
\end{tabular*}
\end{table}

\begin{table}[H]
\linespread{1} 
\scriptsize
\centering
\caption{{\footnotesize Numerical Performances of LogBM, EntBM, SrtBM, and CquPM for Type C Extensive-Form Games}}\label{Table2}
\begin{tabular*}{\hsize}{@{}@{\extracolsep{\fill}}cc|cccc|cccc@{}}
\hline
\multicolumn{2}{c}{} & \multicolumn{4}{c}{Computational Time} & \multicolumn{4}{c}{Number of Iterations}\\
 $n,m_i,|A(I_i)|,L$ & & LogBM & EntBM & SrtBM & CquPM & LogBM & EntBM & SrtBM & CquPM\\
\hline
 $2,(11,21),(2,2),3$ & med & 212.46 & 308.95 & 231.82 & 778.31 & 252 & 399 & 289 & 848 \\
 & min & 146.34 & 159.70 & 156.14 & 504.22 & 167 & 197 & 180 & 545\\
 & max & 357.12 & 1065.14 & 501.61 & 1252.75 & 424 & 1525 & 720 & 1328\\
 $2,(43,85),(2,2),4$ & med & 2614.53 & 3423.78 & 3280.31 & 12456.63 & 330 & 456 & 445 & 1496\\
 & min & 1221.62 & 1788.93 & 1852.56 & 8498.56 & 182 & 175 & 270 & 1038\\
 & max & 3993.11 & 6012.14 & 4556.95 & - & 510 & 980 & 727 & -\\
 $2,(4,10),(3,3),2$ & med & 1080.65 & 1790.18 & 1115.36 & 5944.64 & 250 & 442 & 256 & 1254\\
 & min & 615.88 & 755.93 & 640.12 & 3249.42 & 179 & 218 & 196 & 644\\
 & max & 1598.80 & 6579.91 & 1642.46 & - & 348 & 1876 & 318 & -\\
 $3,(5,9,18),(2,2,2),2$ & med & 1183.71 & 1421.02 & 1210.07 & 5875.99 & 297 & 350 & 308 & 1134\\
 & min & 717.43 & 823.51 & 779.75 & 2790.52 & 199 & 257 & 216 & 795\\
 & max & 2180.57 & 2837.88 & 2592.24 & 17997.70 & 658 & 697 & 768 & 2341\\
\hline
\end{tabular*}
\end{table}

\section{Conclusion}

In this paper, we have presented a characterization of WSRE through $\varepsilon$-perfect $\gamma$-WSRE with local sequential rationality, which is achieved by introducing an extra behavioral strategy profile. This characterization offers an effective tool for analytically identifying all WSREs, particularly in the case of small games. As a result of this characterization, we have attained a polynomial system serving as a necessary and sufficient condition for determining whether an assessment is an $\varepsilon$-perfect $\gamma$-WSRE. Furthermore, we have demonstrated important applications of our characterization by proposing differentiable path-following methods for computing WSREs. Comprehensive numerical results further confirm the efficiency of the methods. One of the future research directions could be the development of an equilibrium concept weaker than WSRE through a relaxation of the requirement of KW-consistency of beliefs.

\begin{spacing}{0.618}

\end{spacing}


\begin{thebibliography}{00}

\setlength{\itemsep}{6pt plus 0.3ex}

\bibitem{Cao and Dang (2024)} Y. Cao \& C. Dang (2024). A characterization of Nash equilibrium in behavioral strategies through local sequential rationality and self-independent beliefs. arXiv:2504.00529v2 [econ.TH]

\bibitem{Cao et al. (2022)} Y. Cao, C. Dang \& Z. Xiao (2022). A differentiable path-following method to compute subgame perfect equilibria in stationary strategies in robust stochastic games and its applications. European Journal of Operational Research 298(3): 1032-1050.

%\bibitem{Cao-Chen-Dang (2023)} Y. Cao, Y. Chen, \& C. Dang (2023). A differentiable path-following method with compact formulation to compute proper equilibria, INFORMS J. Comp. (Published on-Line).

\bibitem{Chen and Dang (2021)} Y. Chen \& C. Dang (2021). A differentiable homotopy method to compute perfect equilibria, Math. Prog. 185: 77-109.

%\bibitem{Cao and Dang (2022)} Y. Cao \& C. Dang (2022). A variant of Harsanyi's tracing procedures to select a perfect equilibrium in normal form games. Games Econ. Behav. 134: 127-150.


%\bibitem{Dang (1991)} C. Dang (1991). The $D_1$-triangulation of $\mathbb{R}^n$ for simplicial algorithms for computing solutions of nonlinear equations. Mathematics of Operations Research 16(1): 148-161.

\bibitem{Eaves (1972)} B.C. Eaves (1972). Homotopies for the computation of fixed
points, Math. Prog. 3: 1-22.

\bibitem{Eaves and Schmedders (1999)} B.C. Eaves \& K. Schmedders (1999). General equilibrium models and homotopy methods. Journal of Economic Dynamics \& Control 23: 1249-1279.

\bibitem{van den Elzen and Talman (1999)} A.H. van den Elzen \&  A.J.J. Talman (1999). An algorithmic
approach towards the tracing procedure for bimatrix games, Games
Econ. Behav. 28: 130-145.

\bibitem{Fudenberg and Tirole (1991)} D. Fudenberg \& J. Tirole (1991). Perfect Bayesian equilibrium and sequential equilibrium, J. Econ. Theory 53(2): 236-260.

\bibitem{Garcia and Zangwill (1981)} C.B. Garcia \&  W.I. Zangwill (1981). Pathways to Solutions,
Fixed Points, and Equilibria, Series in Comp. Math.,
Prentice-Hall, New Jersey.

\bibitem{Govindan and Wilson (2002)} S. Govindan \&  R. Wilson (2002). Structure theorems for game trees, PNAS 99: 9077-9080.

\bibitem{Govindan and Wilson (2003)} S. Govindan \&  R. Wilson (2003). A global Newton method to
compute Nash equilibria, J. Econ.Theory 110: 65-86.

\bibitem{Govindan and Wilson (2009)} S. Govindan \&  R. Wilson (2009). On forward induction, Econometrica 77(1):1-28.

\bibitem{Harsanyi (1975)} J.C. Harsanyi (1975). The tracing procedure: a Bayesian approach to defining a solution for n-person noncooperative games, Int. J. Game Theory 4: 61-94.

\bibitem{Harsanyi and Selten (1988)} J.C. Harsanyi \& R. Selten (1988). A General Theory of
Euqilibrium Selection in Games, MIT Press.

\bibitem{Herings and Peeters (2001)} P.J.J. Herings \& R.J.A.P. Peeters (2001). A differentiable
homotopy to compute Nash equilibria of $n$-person games, Econ.
Theory 18: 159-185.

\bibitem{Herings and Peeters (2010)} P.J.J. Herings \&  R.J.A.P. Peeters (2010). Homotopy methods to
compute equilibria in game theory, Econ. Theory 42(1): 119-156.
 
 \bibitem{Koller and Megiddo (1992)} N. Koller and N. Megiddo (1992). The complexity of two-person zero-sum games in extensive form, Games and Economic Behavior 4: 528-552.
 
 \bibitem{Koller et al. (1996)} N. Koller, N. Megiddo, and B. von Stengel (1996).
 Efficient computation of equilibria for extensive two-person games, Games and Economic Behavior 14: 247-259.
 
\bibitem{Kreps and Wilson (1982)} D.M. Kreps \& R. Wilson (1982). Sequential equilibria, Econometrica 50 (4): 863-894.

\bibitem{Kuhn (1953)} H.W. Kuhn (1953). Extensive games and the problem of information. Annals of Mathematics Studies, 28: 193-216.

\bibitem{van der Laan and Talman (1979)} G. van der Laan \&  A.J.J. Talman  (1979).  A restart algorithm
for computing fixed points without an extra dimension, Math.
Prog. 17: 74-84.

\bibitem{Lemke and Howson (1964)} C.E. Lemke \& J.T. Howson, Jr. (1964). Equilibrium points in
bimatrix games, SIAM J. Applied Math. 12: 413-423.

\bibitem{Mas-Colell (1974)} A. Mas-Colell (1974). A note on a theorem of F. Browder. Mathematical Programming 6: 229-233.

\bibitem{Mas-Colell et al. (1995)} A. Mas-Colell, M.D. Whinston \& J.R. Green (1995). Microeconomic Theory. New York: Oxford University Press.

\bibitem{McKelvey and Palfrey (1998)} R. D.  McKelvey \& T. R. Palfrey (1998). Quantal response equilibria for extensive form games. Experimental economics 1: 9-41.

\bibitem{Milgrom and Mollner (2021)} P. Milgrom \& J. Mollner (2021). Extended proper equilibrium. Journal of Economic Theory 194: 105258.

\bibitem{Miltersen and Sorensen (2010)} P. B. Miltersen \& T. B. S{\o}rensen (2010). Computing a quasi-perfect equilibrium of a two-player game. Economic Theory 42: 175-192.

\bibitem{Myerson (1978)} R.B. Myerson (1978). Refinement of the Nash equilibrium
concept, Intern. J. Game Theory 7(2): 73-80.

\bibitem{Myerson (1991)} R.B. Myerson (1991). Game Theory: Analysis of Conflict, Harvard University Press.

\bibitem{Nash (1951)} J.F. Nash, Jr. (1951). Noncooperative games, Annals of
Math. 54: 289-295.
 
\bibitem{Osborne and Rubinstein (1994)} M. Osborne and A. Rubinstein (1994). A Course in Game Theory, MIT Press.
 
 \bibitem{Reny (1992)}  P.J. Reny (1992). Backward induction, normal-form perfection and explicable equilibria, Econometrica 60(3): 627-649.
   
\bibitem{Scarf (1967)} H.E. Scarf  (1967). The approximation of fixed points of a continuous mapping, SIAM J. Applied Math. 15: 1328-1343.

\bibitem{Selten (1975)} R. Selten (1975). Reexamination of the perfectness concept for
equilibrium points in extensive games, Intern. J. Game Theory 4: 25-55.

\bibitem{Selten (1965)} R. Selten (1965). Spieltheoretische behandlung eines oligopolmodells mit nachfragetragheit. Zeitschrift fur die Gesamte Staatswissenschaft. Journal of Institutional and Theoretical Economics 121(3): 301-324.

\bibitem{Siniscalchi (2022)} M. Siniscalchi (2022).  Structural rationality in dynamic games, Econometrica 90 (5): 2437-2469.

\bibitem{von Stengel (1996)} B. von Stengel (1996). Efficient computation of behavior strategies. Games and Economic Behavior 14(2): 220-246.

\bibitem{von Stengel et al. (2002)} B. von Stengel, A. van den Elzen, \& D. Talman (2002). Computing normal-form perfect equilibria for extensive two-person games, Econometrica 70(2): 693-715.

\bibitem{Todd (1976)} M.J. Todd   (1976). The Computation of Fixed Points and Applications, Lecture Notes in Econ. Math. Sys. 124, Springer-Verlag, Berlin.

\bibitem{Turocy (2010)} T.L. Turocy (2010). Computing sequential equilibria using agent quantal response equilibria, Econ. Theory 42: 255-269.

\bibitem{Wilson (1972)} R. Wilson (1972). Computing equilibria of two-person games from the extensive form. Manag. Sci. 18: 448-460.

\end{thebibliography}
\end{document}